\documentclass{article}
\usepackage{iclr2027_conference,times}
\usepackage{amsmath,amssymb,amsthm,mathtools,booktabs,graphicx,enumitem}
\usepackage{algorithm,algpseudocode}
\usepackage{hyperref}
\usepackage{url}
\newtheorem{theorem}{Theorem}[section]
\newtheorem{proposition}[theorem]{Proposition}
\newtheorem{lemma}[theorem]{Lemma}
\newtheorem{corollary}[theorem]{Corollary}
\newtheorem{definition}[theorem]{Definition}
\theoremstyle{remark}\newtheorem{remark}[theorem]{Remark}
\newcommand{\Hdot}{\dot H}
\makeatletter\newcommand{\tabcaption}{\def\@captype{table}\caption}\makeatother
\DeclareMathOperator{\KL}{KL}
\graphicspath{{figures/}}
\title{Backward Kolmogorov Transport:\\ Sampling Invariant Laws of Reversible\\ Diffusions from Trajectory Data}
\author{%
\parbox[t]{0.35\textwidth}{\centering\textbf{Yuanchao Xu}\thanks{Corresponding author: \texttt{xu.yuanchao.3a@kyoto-u.ac.jp}.}}%
\parbox[t]{0.35\textwidth}{\centering\textbf{Isao Ishikawa}}\\[5pt]
Center for Science Adventure and Collaborative Research Advancement (SACRA)\\
Graduate School of Science, Kyoto University, Kyoto 606-8502, Japan}

\iclrfinalcopy 

\makeatletter
\ificlrfinal
\def\@maketitle{\vbox{\hsize\textwidth
  {\centering\LARGE\sc\@title\par}\lhead{Preprint}%
  \vskip 0.25in
  {\centering\@author\par}%
  \vskip 0.2in}}
\fi
\makeatother

\begin{document}\maketitle

\begin{abstract}
We consider a reversible diffusion with unknown drift, observed through trajectories, and construct samplers of its invariant law. The trajectories contain samples of this law, but a sampler built from them by density or score estimation inherits the errors of these estimates, in particular on the barriers between metastable states. The trajectories also contain the dynamics, and for a reversible diffusion the ratio between the law of the process and its invariant law solves the backward Kolmogorov equation, so that the eigenpairs of the generator propagate it in closed form. Backward Kolmogorov transport (BKT) evaluates this ratio with eigenpairs estimated from the trajectories and coefficients computed once from the initial particles, and moves the particles along the Wasserstein gradient flow of the Kullback-Leibler divergence. The particles do not interact, BKT defines a transport map, and no density, score or drift is estimated. With exact eigenpairs BKT transports the spectral projection of the initial law exactly when this projection is positive. For estimated eigenpairs we prove a transport identity in which truncation enters only through the initial law and estimation only through the residual of the eigenpairs, and a local stability bound in the 2-Wasserstein distance. With coefficients computed from the transported particles, the sampling error of the retained modes cancels, and a central limit theorem identifies the asymptotic variance. Experiments on Ornstein-Uhlenbeck processes, multi-well potentials, separable products and alanine dipeptide, including comparisons with particle methods based on density estimates, agree with these results, and with an accurate spectrum the errors of BKT reach or fall below those of independent samples.
\end{abstract}

\section{Introduction}\label{sec:intro}

Consider a diffusion process with unknown drift that is observed through trajectories, as in molecular dynamics, stochastic partial differential equations and climate models, where the object of interest is the invariant law $\pi$ of the process. The theory of Koopman and transfer operators provides estimators of the generator from trajectories, with error bounds \citep{williams2015,klus2020,noe2013,wu2020,kostic2023,kostic2024,colbrook2023,kostic2022,meanti2023}. A stationary trajectory also contains samples of $\pi$, and one may ask what a sampler adds to them. The recorded samples are finite in number, correlated and fixed, whereas a sampler generates any number of new samples of $\pi$, from a prescribed initial law, along the relaxation path of the process and without further simulation, as generative models produce new samples from their training data \citep{song2021,lipman2023}. The question is how the trajectory data should be used for this purpose.

The Wasserstein gradient flow of the Kullback-Leibler (KL) divergence transports any initial law to $\pi$ along the Fokker-Planck path of the diffusion \citep{jko1998,otto2001,ambrosio2008}, and its velocity is the negative gradient of the logarithm of the ratio between the current law and $\pi$. Particle methods estimate this ratio statistically during the evolution, by kernel density estimation \citep{degond1990,russo1990,lions2001,carrillo2019,chertock2017,klebanov2024} or by a fitted score \citep{maoutsa2020,boffi2023,ilin2025,zhou2025}, and they require the drift or the score of the target, as does the weak generative sampler of \citet{cai2024}. From data, both the target and the current law of the particles must then be estimated from samples. Such estimates are biased by smoothing and noisy where the density is small, in particular on the barriers between metastable states through which the mass must be transported, and their errors do not decrease along the flow. A trajectory contains more than samples of $\pi$, namely the dynamics of the process. For a reversible diffusion the ratio between the current law and $\pi$ solves the backward Kolmogorov equation, so that the dynamics propagate it in closed form, through the eigenfunction expansion of the Koopman semigroup, whose spectrum encodes the relaxation rates and the metastable structure \citep{risken1989,bakry2014,pavliotis2014}. This spectrum is estimated in weak form, by averages over the data, through the Dirichlet form $\int|\nabla f|^2d\pi=\lim_{\tau\to0}(2\tau)^{-1}\mathbb E_\pi|f(X_\tau)-f(X_0)|^2$ of the generator, which measures how observables decorrelate along the trajectories, and no density needs to be estimated. For the dihedral angles of alanine dipeptide, a kernel particle method built on the frames of a molecular dynamics trajectory stops at about three times the error of direct sampling, whereas the transport built on the spectrum estimated from the same frames reaches it (Figure~\ref{fig:ala_main}).

With eigenpairs estimated from trajectories, this expansion underlies density forecasting \citep{berry2015,giannakis2019,zhao2023,wang2025} and the logarithmic-gradient control fields of rare-event simulation \citep{zhang2022koopman}. The transport of a measure along the negative logarithmic gradient of the evolved ratio is the heat-flow map of \citet{kim2012}, whose Lipschitz constant and stability are studied in \citet{neeman2022,mikulincer2023,brigati2024,chewi2025stability,chewi2026}, and density-equalizing maps compute it for a uniform target with the Fourier expansion of the heat kernel \citep{gastner2004} or with the Laplace-Beltrami operator of a surface \citep{choi2018}. In diffusion models the expansion of a known operator gives the score of a noising process that is run backward toward the data \citep{lou2023,shen2024,scarvelis2023}, with the moments of the data computed once in \citet{khoo2025}, and Koopman representations linearise trained generative flows \citep{turan2025}. In these works the operator is known, or the expansion is not used to sample the invariant law. Estimating the expansion from trajectories and using it to transport particles raises a question that does not arise in density forecasting, namely how the accuracy of the eigenpairs transfers to the particles, since the truncated ratio must remain accurate and positive along the path and its logarithmic gradient amplifies the error.

The samplers built so far on an estimated spectrum, the diffusion map particle systems of \citet{li2025} and the Koopman spectral Wasserstein gradient descent of \citet{xu2025}, are data-driven versions of the Laplacian-adjusted Wasserstein gradient descent (LAWGD) of \citet{chewi2020}, which kernelises the Wasserstein gradient flow of the $\chi^2$ divergence as Stein variational gradient descent does \citep{liu2016,teolis2026}. Their particles interact, since the velocity of each particle depends on all current positions through moments re-estimated at every step. In backward Kolmogorov transport (BKT) the velocity field is determined once by the source and the spectrum, the particles do not interact, and the method defines a transport map. Each particle solves its own equation, the sampling error of the source is damped by the semigroup, and the error analysis concerns a single flow map, without propagation of chaos, as for probability-flow samplers and flow matching \citep{song2021,chen2023,chen2023pfode,huang2025,deveney2023,shen2022,kwon2022,lipman2023,albergo2023,benton2023}.

The paper makes three contributions. First, we introduce BKT (Algorithm~\ref{alg:sfm}), which is, to our knowledge, the first method that moves non-interacting particles along the Fokker-Planck path of a process with the spectrum of its generator estimated from trajectories, and we show that with exact eigenpairs it transports a positive spectral projection of the source exactly (Proposition~\ref{prop:exact}). Second, for estimated eigenpairs we prove a transport identity (Proposition~\ref{prop:identity}) in which truncation enters only through the initial law, estimation only through the residual of the eigenpairs and the safeguards through a flux, a bound in total variation in which these terms carry no stability constant (Corollary~\ref{cor:tv}), and, for truncations that change sign, where global stability estimates for probability flows fail, a bound in the 2-Wasserstein distance by localisation on validity regions (Theorem~\ref{thm:main}). Third, for coefficients computed from the transported particles, we show that the sampling error of the source cancels on the retained modes (Proposition~\ref{prop:cancel}) and prove a central limit theorem (Theorem~\ref{thm:clt}) in which the asymptotic variance is that of an independent sample for population coefficients, is larger for coefficients from an independent sample, and for same-sample coefficients is that of the component orthogonal to the retained modes, at most of order $\lambda_{r+1}^{-1}$ for exact eigenpairs up to terms due to the finite horizon, the truncation and the safeguards. This is the mechanism of moment matching in Monte Carlo integration \citep{glasserman2004}, which also underlies kernel herding, support points, kernel thinning and the fixed points of Stein variational gradient descent \citep{chen2010herding,mak2018,dwivedi2024,liu2018}, and it explains the errors below independent sampling in Section~\ref{sec:exp}.

The subject of the paper is the transport and its analysis, and the eigenpairs are taken as input. In high dimension, and in particular for many metastable states, the estimation of the slow eigenpairs is the main difficulty. It is a long-standing problem of transfer operator methods \citep{schutte2023} that is separate from the transport, and since the bounds of Section~\ref{sec:theory} hold for arbitrary estimated eigenpairs, the kernel and neural spectral estimators developed for it can be combined with BKT without change. Section~\ref{sec:exp} discusses where the dimension enters.

\section{Setting and notation}\label{sec:setting}

Let $\pi$ be a probability measure on $\mathbb R^d$ with density proportional to $e^{-V}$, and let $X$ solve the stochastic differential equation
\begin{equation}\label{eq:sde}
  dX_t=-\nabla V(X_t)\,dt+\sqrt2\,dW_t ,
\end{equation}
whose invariant law is $\pi$. The generator of $X$ is $-L$ with $L\coloneqq-\Delta+\nabla V\cdot\nabla$, a nonnegative self-adjoint operator on $L^2(\pi)$ whose Dirichlet form is $\langle f,Lg\rangle_\pi=\int\nabla f\cdot\nabla g\,d\pi$, and $P_t\coloneqq e^{-tL}$ is the associated symmetric Markov semigroup, also called the Koopman semigroup, $(P_tg)(x)=\mathbb E[g(X_t)\mid X_0=x]$ \citep{bakry2014}. We assume that $L$ has discrete spectrum $0=\lambda_0<\lambda_1\le\lambda_2\le\cdots$ with orthonormal eigenfunctions $\varphi_k$, $\varphi_0\equiv1$, and write $\Pi_r$ for the orthogonal projection onto the span of $\varphi_0,\dots,\varphi_r$. Since $P_t$ is self-adjoint, for $\mu_0\ll\pi$ with $\rho_0=d\mu_0/d\pi\in L^2(\pi)$ the law $\mu_t$ of $X_t$ satisfies $d\mu_t/d\pi=P_t\rho_0$, that is
\begin{equation}\label{eq:rep}
  \rho_t\coloneqq\frac{d\mu_t}{d\pi}=1+\sum_{k\ge1}e^{-\lambda_kt}c_k\varphi_k,\qquad c_k\coloneqq\mathbb E_{\mu_0}[\varphi_k],
\end{equation}
and $\chi^2(\mu_t\|\pi)=\sum_{k\ge1}e^{-2\lambda_kt}c_k^2$ (Proposition~\ref{prop:rep}).

The family $(\mu_t)$ solves the Fokker-Planck equation
\begin{equation}\label{eq:fp}
  \partial_t\mu_t=\mathrm{div}\big(\mu_t\nabla\log(\mu_t/\pi)\big),
\end{equation}
which is the Wasserstein gradient flow of the KL divergence $\KL(\mu\|\pi)\coloneqq\int\log(d\mu/d\pi)\,d\mu$, $\mu\ll\pi$ \citep{jko1998}, and $\mu_t$ is the law at time $t$ of the solution of the ordinary differential equation
\begin{equation}\label{eq:odeexact}
  \dot x=v_t(x),\qquad v_t=-\nabla\log\rho_t,\qquad \rho_t=d\mu_t/d\pi,
\end{equation}
started from $\mu_0$ \citep[Ch.~8]{ambrosio2008}, a construction also used in generative modelling \citep{rozen2021}. We denote by $\Phi_t$ the flow map of \eqref{eq:odeexact}, so that $\mu_t=\Phi_t\#\mu_0$. Given pairs $(X_0,X_\tau)$ of a stationary trajectory, extended dynamic mode decomposition and its variational, generator and semigroup versions (gEDMD, SDMD) \citep{williams2015,noe2013,klus2018,wu2020,klus2020,kostic2024,xu2025sdmd}, as well as diffusion maps built from samples of $\pi$ \citep{coifman2006,berry2016}, return approximations $(\hat\lambda_k,\hat\varphi_k,\nabla\hat\varphi_k)_{k\le r}$ of the leading eigenpairs as linear combinations of a dictionary of functions. If the trajectory is not stationary, the same estimators apply after the reweighting of \citet{wu2017,nuske2017}. The accuracy of an approximate eigenfunction is measured in $L^2(\pi)$ and in the energy norm $\|f\|_{\Hdot^1}^2\coloneqq\int|\nabla f|^2d\pi$, and error bounds are given in \citet{kostic2023,nuske2023,kostic2024}.

\paragraph{Notation.} We write $\|f\|$ for the norm of $L^2(\pi)$. For $g$ with $\langle g,1\rangle_\pi=0$ we set $\|g\|_{\Hdot^{-1}}^2\coloneqq\sum_{k\ge1}\langle g,\varphi_k\rangle_\pi^2/\lambda_k$, and $W_2(\mu,\pi)\le2\|d\mu/d\pi-1\|_{\Hdot^{-1}}$ for $\mu\ll\pi$ \citep{peyre2018}. We denote by $\mathcal P_2(\mathbb R^d)$ the probability measures with finite second moment, by $W_2$ the 2-Wasserstein distance, by $\Phi\#\nu$ the image of $\nu$ under $\Phi$, by $\nu_0\coloneqq M^{-1}\sum_j\delta_{x_j}$ the empirical measure of $M$ samples of the source and by $B$ a compact convex domain to which the particles are confined, and we set $\|\mu-\nu\|_{\rm TV}\coloneqq\sup_{|f|\le1}\int f\,d(\mu-\nu)$ and $\mathrm{Clip}_v(u)\coloneqq u\min\{1,v/|u|\}$ for $u\ne0$, with $\mathrm{Clip}_v(0)=0$. The assumptions are stated in Appendix~\ref{app:assump}.

\section{Backward Kolmogorov transport}\label{sec:method}

By \eqref{eq:rep} the velocity of the KL flow is $v_t=-\nabla\log\rho_t$ with $\rho_t=P_t\rho_0$, and the density ratio is determined by the eigenpairs and by the averages of the eigenfunctions over the source. Let $(\hat\lambda_k,\hat\varphi_k)_{k\le r}$ be approximations of the first $r$ nonconstant eigenpairs, with $\hat\lambda_0=0$, $\hat\varphi_0=1$ and $0<\hat\lambda_1\le\dots\le\hat\lambda_r$, let $x_1,\dots,x_M$ be independent samples of the source $\mu_0$, an initial law of the particles that may be chosen freely, and let $T>0$ be the time horizon. The approximate density ratio and the velocity are defined by
\begin{equation}\label{eq:ode}
  \hat\rho_t\coloneqq\sum_{k=0}^{r}e^{-\hat\lambda_kt}\hat c_k\hat\varphi_k,\qquad \hat c_k\coloneqq\int\hat\varphi_k\,d\nu_0,\qquad
  \hat v_t\coloneqq-\nabla\log\hat\rho_t .
\end{equation}
The coefficient $\hat c_k$ is the empirical counterpart of $c_k$ in \eqref{eq:rep}, computed from the particles that are subsequently transported. Each particle $j$ follows the ordinary differential equation
\begin{equation}\label{eq:particle}
  \dot X^j_t=\hat v_t(X^j_t),\qquad X^j_0=x_j,\qquad t\in[0,T],
\end{equation}
with reflection at the boundary of $B$, and the output is the empirical measure $\hat\mu_T$ of the transported particles $X^1_T,\dots,X^M_T$. The velocity $\hat v_t$ depends on the particles only through the coefficients $\hat c_k$, which are computed once at $t=0$. For $t>0$ the particles therefore do not interact, and \eqref{eq:particle} is a system of $M$ decoupled equations, which Algorithm~\ref{alg:sfm} integrates with the classical fourth-order Runge-Kutta method, although any integrator may be used. The spectral summation costs $O(Mrd)$ operations per time step, in addition to the evaluation of the eigenfunctions and their gradients.

\paragraph{Coefficient regimes.}\label{par:coeff-regimes}
The average of $\hat\varphi_k$ over the source can be computed in three ways. In regime (P) it is computed exactly, $\bar c_k=\mathbb E_{\mu_0}[\hat\varphi_k]$, by quadrature when $\mu_0$ is known in closed form. In regime (I) it is computed from an independent sample $y_1,\dots,y_{M'}$ of $\mu_0$, $\check c_k=M'^{-1}\sum_j\hat\varphi_k(y_j)$, and in regime (S) from the transported particles themselves, $\hat c_k=M^{-1}\sum_j\hat\varphi_k(x_j)$, as in \eqref{eq:ode} and Algorithm~\ref{alg:sfm}. In all regimes the trajectory data used for the eigenpairs are independent of the transported sample. Theorem~\ref{thm:main} covers (P) and (I), in which the field is independent of the transported particles. Regime (S) is the regime of the experiments, and Proposition~\ref{prop:cancel} and Theorem~\ref{thm:clt} describe the cancellation of the sampling error on the retained modes.

\begin{remark}\label{rem:safeguard}
A truncated expansion need not be positive, and the velocity in \eqref{eq:ode} is defined only where $\hat\rho_t>0$. The computations therefore use the velocity
\[
  \mathrm{Clip}_{v_{\max}}\Big(-\frac{\nabla\hat\rho_t}{\max\{\hat\rho_t,\varepsilon_0\}}\Big)
\]
with a threshold $\varepsilon_0>0$ and a bound $v_{\max}$. This field coincides with \eqref{eq:ode} where $\hat\rho_t\ge\varepsilon_0$ and $|\nabla\log\hat\rho_t|\le v_{\max}$, and it is continuous, locally Lipschitz and bounded by $v_{\max}$ on $\mathbb R^d$. On the set where it differs from \eqref{eq:ode}, Theorem~\ref{thm:main} uses only this bound and controls the mass of this set. The threshold is inactive at large times (Appendix~\ref{app:assump}), and the safeguards act on few paths (Appendix~\ref{app:exp}).
\end{remark}

\begin{algorithm}[t]
\caption{Backward Kolmogorov transport}\label{alg:sfm}
\begin{algorithmic}[1]
\Require eigenpairs $(\hat\lambda_k,\hat\varphi_k,\nabla\hat\varphi_k)_{k\le r}$ estimated from trajectories, source samples $x_1,\dots,x_M$, horizon $T$, step $h$
\State $\hat c_k\gets M^{-1}\sum_{j}\hat\varphi_k(x_j)$ for $1\le k\le r$, $\hat c_0\gets1$, and $X^j\gets x_j$
\For{$t=0,h,\dots,T-h$}
  \State $\hat v_s\gets-\nabla\hat\rho_s/\hat\rho_s$ with $\hat\rho_s=\sum_{k=0}^{r}e^{-\hat\lambda_ks}\hat c_k\hat\varphi_k$ for $s\in\{t,t+\frac h2,t+h\}$ (Remark~\ref{rem:safeguard})
  \State advance each $X^j$ by one fourth-order Runge-Kutta step for $\dot x=\hat v_s(x)$ and reflect at $\partial B$
\EndFor
\State \Return $\hat\mu_T=M^{-1}\sum_j\delta_{X^j}$ and the spectral diagnostic $\hat\chi^2(t)=\sum_{k=1}^{r}e^{-2\hat\lambda_kt}\hat c_k^2$
\end{algorithmic}
\end{algorithm}

\paragraph{Coefficients, diagnostic and source.} The coefficients decay at the rates of the dynamics, and the slow modes describe the transitions between metastable states. This propagation rule, known from density forecasting \citep{zhao2023,shen2024}, here determines the velocity of the particles and damps the sampling error of the high modes, whereas recomputing the coefficients from the particles at each step, as in the data-driven $\chi^2$ flow of \citet{xu2025} without its preconditioning, feeds it back. Along the exact flow the two rules coincide, since $\mathbb E_{\mu_t}[\varphi_k]=e^{-\lambda_kt}c_k$. The diagnostic $\hat\chi^2(t)$ equals $\chi^2(\pi\hat\rho_t\|\pi)$ when $\hat\rho_t\ge0$ and the eigenfunctions are orthonormal. We call the source admissible when $\rho_0$ is bounded, and for a double well a source of compact support is admissible while a Gaussian source with full support is not.

\begin{lemma}[Residual identity]\label{lem:residual}
Let $\hat q_t\coloneqq\pi\hat\rho_t$, $R_t\coloneqq\sum_{k\le r}e^{-\hat\lambda_kt}\hat c_k(L-\hat\lambda_k)\hat\varphi_k$, and $\Omega_t\coloneqq\{x\in\mathbb R^d\mid\hat\rho_t(x)>0\}$ the positivity set of the truncated ratio. On $\Omega_t$,
\[
  \partial_t\hat q_t+\mathrm{div}\big(\hat q_t\,(-\nabla\log\hat\rho_t)\big)=\pi R_t .
\]
\end{lemma}
If the eigenpairs are exact, then $R\equiv0$ for any coefficients, and the truncated spectral density solves the continuity equation on $\Omega_t$. When the truncated ratio is positive, the conclusion holds on the whole space and for the transport map.

\begin{proposition}[Exact eigenpairs]\label{prop:exact}
Let the eigenpairs be exact, with $\varphi_k\in C^2$, and let $\hat\rho_0=1+\sum_{k=1}^ra_k\varphi_k$ be positive for some $a\in\mathbb R^r$. Then $\hat\rho_t=P_t\hat\rho_0$ and $\inf\hat\rho_t\ge\inf\hat\rho_0$ for $t\ge0$, and the flow of the field \eqref{eq:ode} satisfies $\hat\Phi_{0,t}\#(\pi\hat\rho_0)=\pi\hat\rho_t$, so that BKT is the KL gradient flow started from $\pi\hat\rho_0$. For population coefficients, $a_k=c_k$,
\[
  \|\hat\Phi_{0,T}\#\mu_0-\pi\|_{\rm TV}\le\|(I-\Pi_r)\rho_0\|_{L^1(\pi)}+e^{-\lambda_1T}\|\Pi_r\rho_0-1\| .
\]
\end{proposition}
Truncation thus replaces the source by its spectral projection, a classical property of eigenfunction expansions \citep{risken1989,pavliotis2014} which here holds for the transport map, and by the maximum principle the threshold of Remark~\ref{rem:safeguard} remains inactive for all times when $\hat\rho_0\ge\varepsilon_0$. A truncated expansion can nevertheless change sign, as $1+e^{-t}x$ does for one Ornstein-Uhlenbeck mode, and Section~\ref{sec:theory} treats the general case.

\section{Error analysis}\label{sec:theory}

\subsection{A transport identity}\label{sec:identity}

For the field of Remark~\ref{rem:safeguard}, Lemma~\ref{lem:residual} extends to the whole path with an additional flux supported where the safeguards act. We state this extension for particles that are not confined, on $\mathbb R^d$ or on a torus, and for a field $\hat v$ that is bounded and Lipschitz in $x$ uniformly in $t\le T$, with flow $\hat\Phi_{s,t}$. We set $\gamma_t\coloneqq\hat\rho_t\hat v_t+\nabla\hat\rho_t$, which vanishes where $\hat v_t=-\nabla\log\hat\rho_t$, and $F_T(f)\coloneqq\int_0^T\!\int\nabla(f\circ\hat\Phi_{s,T})\cdot\gamma_s\,d\pi\,ds$.

\begin{proposition}[Transport identity]\label{prop:identity}
Let $\hat\rho_t=\sum_{k=0}^re^{-\hat\lambda_kt}\hat c_k\hat\varphi_k$, with coefficients of any regime, $\hat\varphi_k\in C^2$ and $\hat\varphi_k,|\nabla\hat\varphi_k|,L\hat\varphi_k\in L^1(\pi)$, and $R_t$ as in Lemma~\ref{lem:residual}. For every probability measure $\nu$ and $f\in C^1_b$,
\begin{equation}\label{eq:identity}
\int f\,d(\hat\Phi_{0,T}\#\nu)=\int f\hat\rho_T\,d\pi+\int f\circ\hat\Phi_{0,T}\,d(\nu-\pi\hat\rho_0)-\int_0^T\langle f\circ\hat\Phi_{s,T},R_s\rangle_\pi\,ds+F_T(f).
\end{equation}
\end{proposition}
The particles carry the spectral density $\pi\hat\rho_T$, corrected by the transport of the initial discrepancy $\nu-\pi\hat\rho_0$, of the residual and of the flux of the safeguards (Appendix~\ref{app:identity}). The identity is a duality argument for the continuity equation, of the kind used in a posteriori error estimation \citep{becker2001}, and it separates the three sources of error. With exact eigenpairs and population coefficients, $\mu_0-\pi\hat\rho_0=\pi(I-\Pi_r)\rho_0$, and truncation enters only through the initial law.

\begin{corollary}[Total variation]\label{cor:tv}
In regimes (P) and (I), if the safeguards are inactive on $[0,T]$,
$$
\|\hat\Phi_{0,T}\#\mu_0-\pi\|_{\rm TV}\le\|\rho_0-\hat\rho_0\|_{L^1(\pi)}+\|\hat\rho_T-1\|_{L^1(\pi)}+\sum_{k=1}^r|\hat c_k|\,\frac{\|(L-\hat\lambda_k)\hat\varphi_k\|_{L^1(\pi)}}{\hat\lambda_k}.
$$
Otherwise the same bound holds in the bounded-Lipschitz distance with the additional term $\sup_{s\le T}\mathrm{Lip}(\hat\Phi_{s,T})\int_0^T\|\gamma_s\|_{L^1(\pi)}ds$.
\end{corollary}
Unlike in Theorem~\ref{thm:main} below, no stability constant multiplies the truncation and residual terms, and a Lipschitz constant enters only through the flux of the safeguards. Each residual is weighted by the lifetime $1/\hat\lambda_k$ of its mode, and since $\partial_t\hat\rho_t=-L\hat\rho_t+R_t$, Duhamel's formula shows that the threshold is inactive when $\inf\hat\rho_0-\int_0^T\|R_s\|_\infty ds\ge\varepsilon_0$ (Appendix~\ref{app:assump}).

\subsection{Stability in the Wasserstein distance}\label{sec:stability}

The bound concerns an approximation $\hat\rho$ of the density ratio that is twice differentiable in $x$ and either deterministic or fixed after conditioning on data independent of the transported particles. It is expressed through the error $e_t\coloneqq\hat\rho_t-\rho_t$ and the relative error $\eta_t\coloneqq e_t/\rho_t$. The assumptions (A0) to (A4) of Appendix~\ref{app:assump} include $\rho_0\in L^2(\pi)\cap L^\infty(\pi)$, a compact convex domain $B$ on which $\|\nabla^j\varphi_k\|_{\infty,B}\le K_B(1+\lambda_k)^{\alpha+j/2}$ for $j\le2$, and a one-sided Lipschitz bound $-\nabla^2\log\rho_t\preceq L^+_tI$ for the exact velocity on the region defined below, with $|\nabla\log\rho_t|\le V_G$ and $\|\nabla^2\log\rho_t\|\le H_G$ there. The bound is local because, for a source of compact support, the global one-sided Lipschitz constant of the exact flow itself is not integrable near $t=0$.

\begin{definition}[Validity region]\label{def:good}
Let $\delta>0$. A \emph{validity region} for $\hat\rho$ is a family of closed sets $G_t\subset B$, $t\in[0,T]$, at distance at least $2\delta$ from $B^c$, together with numbers $\theta_t\ge2\varepsilon_0$, such that
$$
  \rho_t\ge\theta_t,\qquad |e_t|\le\tfrac12\rho_t,\qquad |\nabla\log\hat\rho_t|\le v_{\max}\qquad\text{on }G_t .
$$
For such a region, the \emph{expansion of the exact flow} on $G$ is $\Lambda_{s,t}\coloneqq\int_s^t(L^+_u)_+\,du$ with $L^+_u\coloneqq\sup_{G_u}\lambda_{\max}(-\nabla^2\log\rho_u)$, and the \emph{expansion of the approximate flow} is $\hat\Lambda_{s,t}\coloneqq\Lambda_{s,t}+\int_s^t\mathrm{pert}_u\,du$ with $\mathrm{pert}_u\coloneqq\sup_{G_u}\big(2\|\nabla^2\eta_u\|+4|\nabla\eta_u|^2\big)$. The \emph{mass leaving the region} is $p_{\rm exit}\coloneqq\mu_0\big(\{x\mid\Phi_t(x)\notin G_t^\delta\text{ for some }t\le T\}\big)$, where $G_t^\delta\coloneqq\{y\in G_t\mid\bar B_\delta(y)\subset G_t\}$, and the \emph{velocity errors} are
$$
  J_p\coloneqq\int_0^Te^{\hat\Lambda_{s,T}}\,\big\|(\hat v_s-v_s)\mathbf 1_{G_s}\big\|_{L^p(\mu_s)}\,ds,\qquad p\in\{1,2\}.
$$
\end{definition}
On $G_t$ the safeguards are inactive, and the quantity $\hat\Lambda$ measures the separation of nearby trajectories, $p_{\rm exit}$ the mass for which the region cannot be used, and $J_p$ the velocity error seen by the transported mass.

\begin{theorem}[Error bound for BKT]\label{thm:main}
Assume (A0) to (A4) and $\pi\in\mathcal P_2(\mathbb R^d)$, let $\hat\rho$ be deterministic and $(G_t,\theta_t)$ a validity region for it, and consider the confined particle system of Appendix~\ref{app:assump} with i.i.d.\ initial particles $x_j\sim\mu_0$. Let $S_M(P)\coloneqq\mathbb E\,W_2(M^{-1}\sum_{j\le M}\delta_{Z_j},P)$ for $Z_j\sim P$ i.i.d.\ denote the error of $M$ independent samples of a law $P$. Then
\begin{equation}\label{eq:main}
\begin{aligned}
\mathbb E\,W_2(\hat\mu_T,\pi)
&\le\underbrace{W_2(\mu_T,\pi)}_{\text{(i) finite time}}+\underbrace{S_M(\pi)}_{\text{(ii) sampling}}+\underbrace{C_T}_{\text{(iii) approximation}}\\
&\le2\Big(\sum_{k\ge1}\frac{e^{-2\lambda_kT}c_k^2}{\lambda_k}\Big)^{1/2}+S_M(\pi)+\big[4\mathcal J_2^2+4R_B^2q_*+2m_B\big]^{1/2},
\end{aligned}
\end{equation}
where $C_T\coloneqq[J_2^2+4R_B^2q_*+2m_B]^{1/2}$, $q_*\coloneqq\min\{1,\,p_{\rm exit}+\min(J_1/\delta,J_2^2/\delta^2)\}$, $R_B\coloneqq\sup_B|x|$, $m_B\coloneqq\|\rho_0\|_\infty\int_{B^c}|x|^2d\pi$, and $J_p\le2\mathcal J_p$ with
\[
\mathcal J_p\coloneqq\int_0^Te^{\hat\Lambda_{s,T}}\,\theta_s^{-(p-1)/2}\big(\|e_s\|_{\Hdot^1}+V_G\|e_s\|\big)\,ds,\qquad p\in\{1,2\}.
\]
\end{theorem}
For a field constructed from data independent of the particles, the bound holds conditionally on the data.

Term (i) is the distance of the exact flow to the target at time $T$, whose spectral bound decays at least like $e^{-\lambda_1T}$. Term (ii) is the error of $M$ independent samples of the target \citep{fournier2015} and carries no stability factor. Term (iii) consists of the velocity error on the validity region, the mass leaving the region and the mass outside the domain, and the errors of the eigenpairs and of the coefficients enter through $e_t$ (Corollaries~\ref{cor:exact} and~\ref{cor:est}). The amplification factor $e^{\hat\Lambda_{s,T}}$ is that of the exact flow up to a term that vanishes as $r\to\infty$, and it does not multiply the truncation and residual terms of Corollary~\ref{cor:tv}. The proof is a stopped Gr\"onwall argument, as in the analysis of probability flows \citep{chen2023pfode,kwon2022,benton2023} (Appendix~\ref{app:proofmain}). The bound does not depend explicitly on $d$, and its version for the sliced distance of Section~\ref{sec:exp} has a sampling term of order $M^{-1/2}$ in every dimension (Remark~\ref{rem:dimension}).

Validity regions are obtained from envelope bounds $|\nabla^je_t|\le w\,a_jS_t$ on $B$, $j\le2$, where $w$ is the growth envelope of the eigenfunctions and $S_t$ majorises the tail $\sum_{k>r}e^{-\lambda_kt}|c_k|$ and the sampling error of the coefficients. The set where $wS_t\le\kappa\rho_t$ is a validity region, on which $\mathrm{pert}_t=O(a_2\kappa+a_1^2\kappa^2)$, so that $\int_0^T\mathrm{pert}_t\,dt\to0$ as $r\to\infty$ for $\kappa$ of order $\lambda_r^{-3/2}$ (Lemma~\ref{lem:construct}), and $p_{\rm exit}$ is bounded by tail masses of $\mu_t$ (Lemma~\ref{lem:exit}). With exact eigenpairs the error splits as $e_t=e^{\rm tr}_t+e^{\rm fl}_t$ into a truncation error, with $\|e^{\rm tr}_t\|\le e^{-\lambda_{r+1}t}\tau_r$ and $\tau_r\coloneqq\|(I-\Pi_r)\rho_0\|$, and a fluctuation of the coefficients, of order $M'^{-1/2}$ in regime (I) and zero in regime (P) (Corollary~\ref{cor:exact}). With estimated eigenpairs two further terms appear, the change of the coefficients due to the approximate eigenfunctions and the error of the retained modes, whose energy norms are bounded by weighted sums of $|\hat\lambda_k-\lambda_k|$ and $\|\hat\varphi_k-\varphi_k\|_{\Hdot^1}$, $k\le r$ (Corollary~\ref{cor:est}). The velocity error is thus controlled by the estimation theory of the spectral methods, and for Ornstein-Uhlenbeck targets with Gaussian sources these bounds decay geometrically in $r$ (Theorem~\ref{thm:ou}).

\subsection{Same-sample coefficients}\label{sec:samesample}

\begin{proposition}[Same-sample cancellation]\label{prop:cancel}
Let $\hat\varphi_0=1,\hat\varphi_1,\dots,\hat\varphi_r$ be orthonormal in $L^2(\pi)$, $\hat\Pi_r$ the orthogonal projection onto their span and $\hat c_k=M^{-1}\sum_j\hat\varphi_k(x_j)$ (regime (S)). For $f\in C^1_b$ and $g=f\circ\hat\Phi_{0,T}$ the particles $X^j_T=\hat\Phi_{0,T}(x_j)$ satisfy
$$
\frac1M\sum_{j=1}^Mf(X^j_T)-\int f\,d\pi=\int f(\hat\rho_T-1)\,d\pi+\frac1M\sum_{j=1}^M\big[(I-\hat\Pi_r)g\big](x_j)+E_T(f),
$$
where $E_T(f)\coloneqq F_T(f)-\int_0^T\langle f\circ\hat\Phi_{s,T},R_s\rangle_\pi\,ds$ collects the residual and flux terms of \eqref{eq:identity}.
\end{proposition}
In regime (S) the sampling error of the source thus enters only through the component of $g$ orthogonal to the retained modes, and the fluctuation of the coefficients explicitly only through the first term, damped by $e^{-\hat\lambda_1T}$. When $E_T(f)=0$, the particles integrate exactly, with respect to $\pi\hat\rho_T$, every $f$ for which $f\circ\hat\Phi_{0,T}$ lies in the span of the retained eigenfunctions. Since $g$ depends on the sample through $\hat c$, the size of the remaining term is given by a central limit theorem. For $c\in\mathbb R^r$ let $\Phi^c$ be the flow of the field of Remark~\ref{rem:safeguard} with coefficients $(1,c)$, for particles that are not confined, $E^c_T(f)$ the corresponding term of Proposition~\ref{prop:cancel}, $m(c)\coloneqq\int f\circ\Phi^c_{0,T}\,d\mu_0$, $\bar c_k\coloneqq\mathbb E_{\mu_0}\hat\varphi_k$, $\bar g\coloneqq f\circ\Phi^{\bar c}_{0,T}$ and $\hat\varphi\coloneqq(\hat\varphi_1,\dots,\hat\varphi_r)$.

\begin{theorem}[Asymptotic variances]\label{thm:clt}
In the setting of Proposition~\ref{prop:cancel}, assume (A2), let $\hat\varphi_k\in L^2(\mu_0)$, and assume that on a bounded neighbourhood $U$ of $\bar c$ the fields satisfy the hypotheses of Proposition~\ref{prop:identity}, that $|f\circ\Phi^c_{0,T}-f\circ\Phi^{c'}_{0,T}|\le\kappa\,|c-c'|$ for some $\kappa\in L^2(\mu_0)\cap L^2(\pi)$, and that $m$ is differentiable at $\bar c$. Then, with $M'=M$ in regime (I), $\sqrt M\big(\int f\,d\hat\mu_T-m(\bar c)\big)$ converges in law to $N(0,\sigma^2)$, where
\[
\sigma_P^2=\mathrm{Var}_{\mu_0}(\bar g),\qquad\sigma_I^2=\sigma_P^2+\mathrm{Var}_{\mu_0}\big(\nabla m(\bar c)\cdot\hat\varphi\big),\qquad\sigma_S^2=\mathrm{Var}_{\mu_0}\big(\bar g+\nabla m(\bar c)\cdot\hat\varphi\big).
\]
Moreover $\nabla m(\bar c)=-(\langle\bar g,\hat\varphi_k\rangle_\pi)_{1\le k\le r}+\beta$ with $|\beta|\le e^{-\hat\lambda_1T}\|f\|+\|\kappa\|\hat\tau_r+\ell_T(f)$, where $\hat\tau_r\coloneqq\|(I-\hat\Pi_r)\rho_0\|$ and $\ell_T(f)\coloneqq\limsup_{c\to\bar c}|E^c_T(f)-E^{\bar c}_T(f)|/|c-\bar c|$. Hence
\[
\sigma_S\le\|(I-\hat\Pi_r)\bar g\|_{L^2(\mu_0)}+\|\rho_0\|_\infty^{1/2}|\beta|,
\]
and $\|(I-\Pi_r)\bar g\|_{L^2(\mu_0)}^2\le\|\rho_0\|_\infty\lambda_{r+1}^{-1}\|\bar g\|_{\Hdot^1}^2$ for exact eigenpairs.
\end{theorem}
The variance in regime (P) is that of $f$ under the output law $\Phi^{\bar c}_{0,T}\#\mu_0$, close to $\mathrm{Var}_\pi(f)$ by Sections~\ref{sec:identity} and~\ref{sec:stability}, and regime (I) adds the independent fluctuation of the coefficients. In regime (S) the component of $\bar g$ along the retained modes is removed, and $\beta$ collects the finite horizon, the truncation and the residual and flux terms, with $\ell_T(f)=0$ for exact eigenpairs when the safeguards are inactive on the whole space for $c$ near $\bar c$. For exact eigenpairs $\sigma_S^2\le2\|\rho_0\|_\infty(\lambda_{r+1}^{-1}\|\bar g\|_{\Hdot^1}^2+|\beta|^2)$, against $\mathrm{Var}_\pi(f)\le\lambda_1^{-1}\|f\|_{\Hdot^1}^2$ for an independent sample. This is the mechanism of moment matching in Monte Carlo integration \citep{glasserman2004}, in which a sample is transformed so that prescribed moments take their exact values, here those of the retained eigenfunctions transported by the flow. In the matched Ornstein-Uhlenbeck setting the observed standard deviations follow Theorem~\ref{thm:clt}, and regime (S) reduces the standard deviation of an independent sample by factors between $4$ and $70$ for smooth test functions (Table~\ref{tab:clt}). The gain is smaller in the Wasserstein distance, which involves test functions at all scales, and Corollary~\ref{cor:samesample} bounds the cost of the same-sample dependence in $W_2$ by a term of order $M^{-1/2}$ and the mass leaving the validity region.

\begin{remark}[Interacting spectral samplers]\label{rem:lawgd}
With the same eigenpairs, the $M$ particles of LAWGD follow $\dot x_i=-\sum_{k=1}^r\hat\lambda_k^{-1}m_k\nabla\hat\varphi_k(x_i)$ with $m_k=M^{-1}\sum_j\hat\varphi_k(x_j)$, the gradient flow of $\mathcal E=\frac M2\sum_{k=1}^r\hat\lambda_k^{-1}m_k^2$, a multiple of the squared maximum mean discrepancy between their empirical measure and $\pi$ for the kernel $\sum_{k=1}^r\hat\lambda_k^{-1}\hat\varphi_k\otimes\hat\varphi_k$ \citep{arbel2019}. Its zeros integrate the retained eigenfunctions exactly, so that LAWGD attains the cancellation of Proposition~\ref{prop:cancel} through the interaction when it reaches a zero of $\mathcal E$, whereas BKT imposes it at $t=0$.
\end{remark}

\section{Numerical results}\label{sec:exp}

The experiments examine, for systems with a known invariant law, how BKT compares with particle methods that use the same data through density estimates, how the accuracy of the eigenpairs, computed with $r$ nonconstant modes, $J$ dictionary functions and $n$ trajectory pairs, transfers to the $M$ transported particles, and how the dimension and the number of metastable states enter. Unless stated otherwise the coefficients are those of regime (S). We write Koopman (RBF) and Koopman (Legendre) for eigenpairs estimated by reversible gEDMD \citep{klus2020} with Gaussian and Legendre dictionaries, and compare them with a finite-difference discretisation of the Dirichlet form on $N$ points per direction (FD-$N$) or with the analytical Hermite spectrum. The error is the sliced $W_2$ distance, with $64$ fixed directions, between the transported particles and a target sample. The reference level is the mean $\pm$ standard deviation (sd) over $10$ realisations of the same distance between two independent target samples, and an error below this mean plus one standard deviation is said to reach the reference level. The comparisons are with particle methods that, like BKT, involve no training of a neural network, and the horizon is $T=8/\lambda_1$ (Appendix~\ref{app:exp}).

\begin{figure}[t]
\begin{minipage}[b]{0.57\textwidth}\centering\scriptsize\setlength{\tabcolsep}{2.2pt}
\begin{tabular}{lccc}\toprule
system & FD & Koopman & reference\\\midrule
2D, $\beta=0$ & $0.027\pm0.002$ & $0.028\pm0.003$ & $0.040\pm0.008$\\
2D, $\beta=0.25$ & $0.028\pm0.002$ & $0.029\pm0.003$ & $0.041\pm0.008$\\
2D, $\beta=0.5$ & $0.028\pm0.002$ & $0.031\pm0.004$ & $0.042\pm0.010$\\
2D, $\beta=1$ & $0.027\pm0.002$ & $0.027\pm0.003$ & $0.042\pm0.012$\\
2D four wells & $0.024\pm0.003$ & $0.028\pm0.002$ & $0.043\pm0.008$\\
10D & $0.053\pm0.002$ & $0.061\pm0.002$ & $0.059\pm0.004$\\
50D & $0.0158\pm0.0002$ & $0.0188\pm0.0012$ & $0.0162\pm0.0008$\\\bottomrule
\end{tabular}
\tabcaption{Sliced $W_2$ error of BKT with the FD and the Koopman (RBF) eigenpairs, mean $\pm$ sd over ten realisations, for the double well with coupling $\beta$ (2D), the four-well product, the double well with nine harmonic coordinates (10D) and the fifty-fold product (50D).}\label{tab:cells}
\end{minipage}\hfill
\begin{minipage}[b]{0.4\textwidth}\centering
\includegraphics[width=\textwidth]{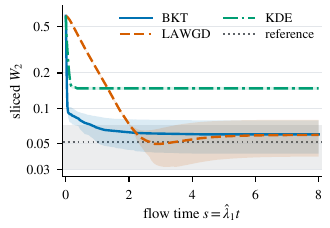}
\caption{Alanine dipeptide, sliced $W_2$ against the flow time, mean $\pm$ sd over three realisations, and the reference level.}\label{fig:ala_main}
\end{minipage}
\end{figure}

\paragraph{Dynamics against density estimates.} For the dihedral angles $(\phi,\psi)$ of alanine dipeptide, the eigenpairs of a reversible surrogate with unit mobility are estimated from $250{,}000$ frames of a public molecular dynamics trajectory with a Fourier dictionary of degree $28$, and $256$ modes transport a source localised in one basin. The kernel particle method uses the same frames through a kernel estimate of the target density and a kernel estimate of the density of the particles \citep{degond1990,chertock2017}. It reaches a sliced error of $0.148$ by $s=\hat\lambda_1t=1$ and remains there, about three times the error $0.051\pm0.021$ of direct sampling (Figure~\ref{fig:ala_main}), since its velocity depends on two density estimates whose errors do not decrease along the flow. BKT reaches $0.060\pm0.019$ at $s=8$, within the fluctuation of the reference, with a total variation of the basin masses of $0.033\pm0.018$ against $0.023\pm0.014$. LAWGD, which uses the same eigenpairs, reaches a comparable final error (Remark~\ref{rem:lawgd}), and BKT attains its late-time level in $34$ against $54$ seconds (Appendix~\ref{app:ala}). The same effect appears for the two-dimensional double wells, where the kernel method uses a drift regressed from the increments of the trajectory and gives $0.039$ and $0.048$ at $\beta=0$ and $0.5$, against $0.028$ and $0.031$ for BKT, with twice the spread at $\beta=0.5$ and five times the cost (Figure~\ref{fig:comparison}).

\paragraph{Ornstein-Uhlenbeck process.} With the analytical spectrum, $M=10^4$ and $T=6$, the error equals $0.074$, $0.045$, $0.023$, $0.011$ and $0.011$ at $r=10,15,20,30,60$, against $0.019\pm0.005$ for $10^4$ independent samples, and lies below the reference level beyond $r=20$ to $30$ (Figure~\ref{fig:ou}). Matched comparisons order the errors as (S) $<$ (P) $<$ (I), in agreement with Theorem~\ref{thm:clt} (Appendix~\ref{app:ou_path_residual}).

\paragraph{Two-dimensional double well with coupling.} For $V=(x_1^2-1)^2+x_2^2+\frac\beta2(x_2-x_1)^2$ and a source equal to a bump in the right well of $x_1$ times $N(0,0.3)$ in $x_2$, BKT with the FD-201 and Koopman (RBF) eigenpairs lies below the reference level at $r=64$ for every $\beta\in\{0,0.25,0.5,1\}$, with differences between the values of $\beta$ within the fluctuation (Table~\ref{tab:cells}). The reference level is reached between $r=32$ and $64$, since the ordering by eigenvalue interleaves the double-well modes with the harmonic ones (Figure~\ref{fig:cells}). Relative to the reference mean, same-sample coefficients give $0.65$ to $0.73$ and coefficients from an independent sample $1.02$ to $1.23$, close to the values $2^{-1/2}$ and $(3/2)^{1/2}$ expected from Theorem~\ref{thm:clt} when the retained modes carry most of the fluctuation and the distance is evaluated against one target sample (Appendices~\ref{app:resolution} and~\ref{app:comp}).

\paragraph{Multi-well potentials.} For the quartic four-well potential $V=(x^2-1)^2+(y^2-1)^2$ with the source in one well, both spectra lie below the reference level at $r=64$, and the four basin masses are recovered. With the FD spectrum BKT also reaches it for a nine-well potential ($0.033\pm0.003$ against $0.034\pm0.008$), and with the estimated spectrum it gives $0.055\pm0.006$, which isolates the error of the spectral estimation (Appendix~\ref{app:comp}).

\paragraph{Dimension and metastability.} The construction of BKT and the bounds of Section~\ref{sec:theory} do not depend on the dimension explicitly, and the dimension enters through the spectrum in two ways. First, the retained modes must represent the source. When the barriers are high, a potential with $n$ local minima has $n-1$ nonzero eigenvalues that are exponentially small in the barrier heights and separated from the rest of the spectrum \citep{bovier2005}, and a source concentrated in one well has in general nonzero coefficients on all of them, with $n=2^d$ for the product $V=\sum_{i\le d}(x_i^2-1)^2$ of $d$ double wells. Second, these slow eigenpairs must be estimated, and with a fixed dictionary the number of functions grows exponentially with $d$. Both difficulties belong to the spectral description of metastable dynamics, the central problem of transfer operator methods in molecular dynamics \citep{schutte2023}, and in the experiments the transport reaches the reference level whenever the spectrum is accurate, as for the coupled ten-dimensional Ornstein-Uhlenbeck process treated in its eigenbasis ($0.0195\pm0.0005$ against $0.0235\pm0.0018$, Appendix~\ref{app:ou_hd}) and for the fifty-fold product, where the product structure replaces the $17^{50}$ products of one-dimensional modes of a joint expansion by $16$ modes per coordinate (Table~\ref{tab:cells}). For $V=(x_1^2-1)^2+\sum_{i=2}^{10}x_i^2$, eigenpairs estimated in $\mathbb R^{10}$ ($J=657$, $r=16$) also reach the reference level, since the source differs from $\pi$ only in the slow coordinate, and in the fifty-fold product the largest marginal errors ($0.039$ with the estimated factor, $0.020$ with the FD factor, $0.030$ for the reference) place the excess in the one-dimensional spectral estimation (Appendix~\ref{app:comp}). The number of modes is thus set by the slow coordinates in which the source differs from $\pi$, and the source can be chosen accordingly. In applications the slow dynamics of high-dimensional systems is described on a few collective variables, as for alanine dipeptide, and coupled systems with many wells require eigenpairs from kernel or neural estimators \citep{meanti2023,mardt2018,kostic2024,zhang2022eigen}, which BKT accepts without change, since its bounds transfer the errors of any eigenpairs in energy norm and eigenvalue (Corollary~\ref{cor:est}).

\section{Concluding remarks}\label{sec:limits}

BKT turns the spectrum of the generator, estimated from trajectories, into a transport map that moves non-interacting particles along the KL gradient flow of a reversible diffusion, and it uses the trajectories through the dynamics rather than through density estimates. With exact eigenpairs it transports a positive spectral projection of the source exactly, its error for estimated eigenpairs is controlled by the decay of the exact flow, the sampling error of the source, the truncation and the residual of the eigenpairs, and same-sample coefficients remove the retained modes from the fluctuation of the source. For a nonreversible diffusion with the same invariant law and unit diffusion, BKT built on Dirichlet-form estimators still samples $\pi$ (Appendix~\ref{app:assump}). In high dimension the limiting factor is the estimation of the slow eigenpairs of metastable dynamics, and the combination of BKT with neural spectral estimators, to which the bounds of Section~\ref{sec:theory} apply directly, is the next step, together with a non-asymptotic version of Theorem~\ref{thm:clt} in the Wasserstein distance.

\ificlrfinal
\subsubsection*{Acknowledgments}
Y.X. and I.I. acknowledge support from JST CREST Grant No.~JPMJCR24Q1, including the AIP Challenge Program.
\fi

\subsubsection*{Reproducibility statement}
The code, configurations and data needed to reproduce all experiments are available at
\ificlrfinal\url{https://github.com/TalkingDoll/BKT}\else\url{https://anonymous.4open.science/r/BKT-1E06}\fi.

\subsubsection*{AI use statement}
An AI assistant was used to assist with the writing.

\bibliography{refs}

@inproceedings{albergo2023,
  author = {Michael S. Albergo and Eric Vanden-Eijnden},
  title = {{Building normalizing flows with stochastic interpolants}},
  booktitle = {International Conference on Learning Representations},
  year = {2023},
}

@book{ambrosio2008,
  author = {Luigi Ambrosio and Nicola Gigli and Giuseppe Savar{\'e}},
  title = {{Gradient Flows in Metric Spaces and in the Space of Probability Measures}},
  year = {2008},
  edition = {Second},
  series = {Lectures in Mathematics. ETH Z{\"u}rich},
  publisher = {Birkh{\"a}user},
  doi = {10.1007/978-3-7643-8722-8},
}

@book{bakry2014,
  author = {Dominique Bakry and Ivan Gentil and Michel Ledoux},
  title = {{Analysis and Geometry of Markov Diffusion Operators}},
  year = {2014},
  volume = {348},
  series = {Grundlehren der mathematischen Wissenschaften},
  publisher = {Springer},
  doi = {10.1007/978-3-319-00227-9},
}

@article{benton2023,
  author = {Joe Benton and George Deligiannidis and Arnaud Doucet},
  title = {{Error bounds for flow matching methods}},
  journal = {Transactions on Machine Learning Research},
  year = {2024},
  eprint = {2305.16860},
  archivePrefix = {arXiv},
}

@article{berry2016,
  author = {Tyrus Berry and John Harlim},
  title = {{Variable bandwidth diffusion kernels}},
  journal = {Applied and Computational Harmonic Analysis},
  year = {2016},
  volume = {40},
  number = {1},
  pages = {68--96},
  doi = {10.1016/j.acha.2015.01.001},
}

@article{berry2015,
  author = {Tyrus Berry and Dimitrios Giannakis and John Harlim},
  title = {{Nonparametric forecasting of low-dimensional dynamical systems}},
  journal = {Physical Review E},
  year = {2015},
  volume = {91},
  number = {3},
  pages = {032915},
  doi = {10.1103/PhysRevE.91.032915},
}

@article{bobkov2019,
  author = {Sergey Bobkov and Michel Ledoux},
  title = {{One-dimensional empirical measures, order statistics, and Kantorovich transport distances}},
  journal = {Memoirs of the American Mathematical Society},
  year = {2019},
  volume = {261},
  number = {1259},
  doi = {10.1090/memo/1259},
}

@article{boffi2023,
  author = {Nicholas M. Boffi and Eric Vanden-Eijnden},
  title = {{Probability flow solution of the Fokker--Planck equation}},
  journal = {Machine Learning: Science and Technology},
  year = {2023},
  volume = {4},
  number = {3},
  pages = {035012},
  doi = {10.1088/2632-2153/ace2aa},
  eprint = {2206.04642},
  archivePrefix = {arXiv},
}

@article{brigati2024,
  author = {Giovanni Brigati and Francesco Pedrotti},
  title = {{Heat flow, log-concavity, and Lipschitz transport maps}},
  journal = {Electronic Communications in Probability},
  year = {2025},
  volume = {30},
  pages = {1--12},
  doi = {10.1214/25-ECP717},
  eprint = {2404.15205},
  archivePrefix = {arXiv},
}

@article{cai2024,
  author = {Zhiqiang Cai and Yu Cao and Yuanfei Huang and Xiang Zhou},
  title = {{Weak generative sampler to efficiently sample invariant distribution of stochastic differential equation}},
  journal = {SIAM Journal on Scientific Computing},
  year = {2026},
  volume = {48},
  number = {4},
  pages = {C708--C735},
  doi = {10.1137/24M1665271},
  eprint = {2405.19256},
  archivePrefix = {arXiv},
}

@article{carrillo2019,
  author = {Jos{\'e} Antonio Carrillo and Katy Craig and Francesco S. Patacchini},
  title = {{A blob method for diffusion}},
  journal = {Calculus of Variations and Partial Differential Equations},
  year = {2019},
  volume = {58},
  number = {2},
  pages = {53},
  doi = {10.1007/s00526-019-1486-3},
}

@inproceedings{chen2023pfode,
  author = {Sitan Chen and Sinho Chewi and Holden Lee and Yuanzhi Li and Jianfeng Lu and Adil Salim},
  title = {{The probability flow ODE is provably fast}},
  booktitle = {Advances in Neural Information Processing Systems},
  year = {2023},
  volume = {36},
  pages = {68552--68575},
  publisher = {Curran Associates, Inc.},
}

@inproceedings{chen2023,
  author = {Sitan Chen and Sinho Chewi and Jerry Li and Yuanzhi Li and Adil Salim and Anru R. Zhang},
  title = {{Sampling is as easy as learning the score: theory for diffusion models with minimal data assumptions}},
  booktitle = {International Conference on Learning Representations},
  year = {2023},
}

@inproceedings{chewi2020,
  author = {Chewi, Sinho and Le Gouic, Thibaut and Lu, Chen and Maunu, Tyler and Rigollet, Philippe},
  title = {{SVGD as a kernelized Wasserstein gradient flow of the chi-squared divergence}},
  booktitle = {Advances in Neural Information Processing Systems},
  year = {2020},
  volume = {33},
  pages = {2098--2109},
  publisher = {Curran Associates, Inc.},
  eprint = {2006.02509},
  archivePrefix = {arXiv},
}

@article{chewi2025stability,
  author = {Sinho Chewi and Aram-Alexandre Pooladian and Matthew S. Zhang},
  title = {{Stability of the Kim--Milman flow map}},
  journal = {Electronic Communications in Probability},
  year = {2026},
  volume = {31},
  pages = {1--13},
  doi = {10.1214/26-ECP800},
  eprint = {2511.01154},
  archivePrefix = {arXiv},
}

@article{chewi2026,
  author = {Sinho Chewi and Katharina Eichinger and Aram-Alexandre Pooladian},
  title = {{Near-Lipschitz stability of the Kim--Milman flow map}},
  journal = {arXiv preprint arXiv:2606.23383},
  year = {2026},
  eprint = {2606.23383},
  archivePrefix = {arXiv},
}

@article{coifman2006,
  author = {Ronald R. Coifman and St{\'e}phane Lafon},
  title = {{Diffusion maps}},
  journal = {Applied and Computational Harmonic Analysis},
  year = {2006},
  volume = {21},
  number = {1},
  pages = {5--30},
  doi = {10.1016/j.acha.2006.04.006},
}

@article{colbrook2023,
  author = {Matthew J. Colbrook and Lorna J. Ayton and M{\'a}t{\'e} Sz{\H o}ke},
  title = {{Residual dynamic mode decomposition: robust and verified Koopmanism}},
  journal = {Journal of Fluid Mechanics},
  year = {2023},
  volume = {955},
  pages = {A21},
  doi = {10.1017/jfm.2022.1052},
}

@article{degond1990,
  author = {Pierre Degond and Francisco-Jos{\'e} Mustieles},
  title = {{A deterministic approximation of diffusion equations using particles}},
  journal = {SIAM Journal on Scientific and Statistical Computing},
  year = {1990},
  volume = {11},
  number = {2},
  pages = {293--310},
  doi = {10.1137/0911018},
}

@article{deveney2023,
  author = {Teo Deveney and Jan Stanczuk and Lisa Kreusser and Chris Budd and Carola-Bibiane Sch{\"o}nlieb},
  title = {{Closing the ODE--SDE gap in score-based diffusion models through the Fokker--Planck equation}},
  journal = {Philosophical Transactions of the Royal Society A: Mathematical, Physical and Engineering Sciences},
  year = {2025},
  volume = {383},
  number = {2298},
  pages = {20240503},
  doi = {10.1098/rsta.2024.0503},
  eprint = {2311.15996},
  archivePrefix = {arXiv},
}

@inproceedings{devergne2024,
  author = {Timoth{\'e}e Devergne and Vladimir R. Kostic and Michele Parrinello and Massimiliano Pontil},
  title = {{From biased to unbiased dynamics: an infinitesimal generator approach}},
  booktitle = {Advances in Neural Information Processing Systems},
  year = {2024},
  volume = {37},
  pages = {75495--75521},
  publisher = {Curran Associates, Inc.},
}

@article{fournier2015,
  author = {Nicolas Fournier and Arnaud Guillin},
  title = {{On the rate of convergence in Wasserstein distance of the empirical measure}},
  journal = {Probability Theory and Related Fields},
  year = {2015},
  volume = {162},
  number = {3--4},
  pages = {707--738},
  doi = {10.1007/s00440-014-0583-7},
}

@article{gastner2004,
  author = {Michael T. Gastner and M. E. J. Newman},
  title = {{Diffusion-based method for producing density-equalizing maps}},
  journal = {Proceedings of the National Academy of Sciences},
  year = {2004},
  volume = {101},
  number = {20},
  pages = {7499--7504},
  doi = {10.1073/pnas.0400280101},
}

@article{giannakis2019,
  author = {Dimitrios Giannakis},
  title = {{Data-driven spectral decomposition and forecasting of ergodic dynamical systems}},
  journal = {Applied and Computational Harmonic Analysis},
  year = {2019},
  volume = {47},
  number = {2},
  pages = {338--396},
  doi = {10.1016/j.acha.2017.09.001},
}

@article{huang2025,
  author = {Daniel Zhengyu Huang and Jiaoyang Huang and Zhengjiang Lin},
  title = {{Convergence analysis of probability flow ODE for score-based generative models}},
  journal = {IEEE Transactions on Information Theory},
  year = {2025},
  volume = {71},
  number = {6},
  pages = {4581--4601},
  doi = {10.1109/TIT.2025.3557050},
  eprint = {2404.09730},
  archivePrefix = {arXiv},
}

@article{ilin2025,
  author = {Vasily Ilin and Peter Sushko and Jingwei Hu},
  title = {{Score-based deterministic density sampling}},
  journal = {Communications on Pure and Applied Analysis},
  year = {2026},
  volume = {32},
  pages = {153--170},
  doi = {10.3934/cpaa.2026028},
  eprint = {2504.18130},
  archivePrefix = {arXiv},
}

@article{indritz1961,
  author = {Jack Indritz},
  title = {{An inequality for Hermite polynomials}},
  journal = {Proceedings of the American Mathematical Society},
  year = {1961},
  volume = {12},
  number = {6},
  pages = {981--983},
  doi = {10.1090/S0002-9939-1961-0132852-2},
}

@article{jko1998,
  author = {Richard Jordan and David Kinderlehrer and Felix Otto},
  title = {{The variational formulation of the Fokker--Planck equation}},
  journal = {SIAM Journal on Mathematical Analysis},
  year = {1998},
  volume = {29},
  number = {1},
  pages = {1--17},
  doi = {10.1137/S0036141096303359},
}

@article{khoo2025,
  author = {Yuehaw Khoo and Mathias Oster and Yifan Peng},
  title = {{Optimization-free diffusion model---a perturbation theory approach}},
  journal = {arXiv preprint arXiv:2505.23652},
  year = {2025},
  eprint = {2505.23652},
  archivePrefix = {arXiv},
}

@article{kim2012,
  author = {Young-Heon Kim and Emanuel Milman},
  title = {{A generalization of Caffarelli's contraction theorem via (reverse) heat flow}},
  journal = {Mathematische Annalen},
  year = {2012},
  volume = {354},
  number = {3},
  pages = {827--862},
  doi = {10.1007/s00208-011-0749-x},
}

@article{klebanov2024,
  author = {Ilja Klebanov},
  title = {{Deterministic Fokker--Planck transport---with applications to sampling, variational inference, kernel mean embeddings \& sequential Monte Carlo}},
  journal = {arXiv preprint arXiv:2410.18993},
  year = {2024},
  eprint = {2410.18993},
  archivePrefix = {arXiv},
}

@article{klus2018,
  author = {Stefan Klus and Feliks N{\"u}ske and P{\'e}ter Koltai and Hao Wu and Ioannis Kevrekidis and Christof Sch{\"u}tte and Frank No{\'e}},
  title = {{Data-driven model reduction and transfer operator approximation}},
  journal = {Journal of Nonlinear Science},
  year = {2018},
  volume = {28},
  number = {3},
  pages = {985--1010},
  doi = {10.1007/s00332-017-9437-7},
}

@article{klus2020,
  author = {Stefan Klus and Feliks N{\"u}ske and Sebastian Peitz and Jan-Hendrik Niemann and Cecilia Clementi and Christof Sch{\"u}tte},
  title = {{Data-driven approximation of the Koopman generator: model reduction, system identification, and control}},
  journal = {Physica D: Nonlinear Phenomena},
  year = {2020},
  volume = {406},
  pages = {132416},
  doi = {10.1016/j.physd.2020.132416},
}

@inproceedings{kostic2022,
  author = {Vladimir Kostic and Pietro Novelli and Andreas Maurer and Carlo Ciliberto and Lorenzo Rosasco and Massimiliano Pontil},
  title = {{Learning dynamical systems via Koopman operator regression in reproducing kernel Hilbert spaces}},
  booktitle = {Advances in Neural Information Processing Systems},
  year = {2022},
  volume = {35},
  pages = {4017--4031},
  publisher = {Curran Associates, Inc.},
}

@inproceedings{kostic2023,
  author = {Vladimir Kostic and Karim Lounici and Pietro Novelli and Massimiliano Pontil},
  title = {{Sharp spectral rates for Koopman operator learning}},
  booktitle = {Advances in Neural Information Processing Systems},
  year = {2023},
  volume = {36},
  pages = {32328--32339},
  publisher = {Curran Associates, Inc.},
}

@inproceedings{kostic2024,
  author = {Vladimir R. Kostic and Karim Lounici and H{\'e}l{\`e}ne Halconruy and Timoth{\'e}e Devergne and Massimiliano Pontil},
  title = {{Learning the infinitesimal generator of stochastic diffusion processes}},
  booktitle = {Advances in Neural Information Processing Systems},
  year = {2024},
  volume = {37},
  pages = {137806--137846},
  publisher = {Curran Associates, Inc.},
  eprint = {2405.12940},
  archivePrefix = {arXiv},
}

@inproceedings{kwon2022,
  author = {Dohyun Kwon and Ying Fan and Kangwook Lee},
  title = {{Score-based generative modeling secretly minimizes the Wasserstein distance}},
  booktitle = {Advances in Neural Information Processing Systems},
  year = {2022},
  volume = {35},
  pages = {20205--20217},
  publisher = {Curran Associates, Inc.},
}

@article{li2025,
  author = {Fengyi Li and Youssef Marzouk},
  title = {{Diffusion map particle systems for generative modeling}},
  journal = {Foundations of Data Science},
  year = {2025},
  volume = {7},
  number = {3},
  pages = {814--837},
  doi = {10.3934/fods.2024054},
}

@article{lions2001,
  author = {Pierre-Louis Lions and Sylvie Mas-Gallic},
  title = {{Une m{\'e}thode particulaire d{\'e}terministe pour des {\'e}quations diffusives non lin{\'e}aires}},
  journal = {Comptes Rendus de l'Acad{\'e}mie des Sciences, S{\'e}rie I},
  year = {2001},
  volume = {332},
  number = {4},
  pages = {369--376},
  doi = {10.1016/S0764-4442(00)01795-X},
}

@article{lions1984,
  author = {Pierre-Louis Lions and Alain-Sol Sznitman},
  title = {{Stochastic differential equations with reflecting boundary conditions}},
  journal = {Communications on Pure and Applied Mathematics},
  year = {1984},
  volume = {37},
  number = {4},
  pages = {511--537},
  doi = {10.1002/cpa.3160370408},
}

@inproceedings{lipman2023,
  author = {Yaron Lipman and Ricky T. Q. Chen and Heli Ben-Hamu and Maximilian Nickel and Matt Le},
  title = {{Flow matching for generative modeling}},
  booktitle = {International Conference on Learning Representations},
  year = {2023},
}

@inproceedings{liu2016,
  author = {Qiang Liu and Dilin Wang},
  title = {{Stein variational gradient descent: a general purpose Bayesian inference algorithm}},
  booktitle = {Advances in Neural Information Processing Systems},
  year = {2016},
  volume = {29},
  pages = {2378--2386},
  publisher = {Curran Associates, Inc.},
}

@inproceedings{lou2023,
  author = {Aaron Lou and Minkai Xu and Adam Farris and Stefano Ermon},
  title = {{Scaling Riemannian diffusion models}},
  booktitle = {Advances in Neural Information Processing Systems},
  year = {2023},
  volume = {36},
  pages = {80291--80305},
  publisher = {Curran Associates, Inc.},
  eprint = {2310.20030},
  archivePrefix = {arXiv},
}

@article{lusch2018,
  author = {Bethany Lusch and J. Nathan Kutz and Steven L. Brunton},
  title = {{Deep learning for universal linear embeddings of nonlinear dynamics}},
  journal = {Nature Communications},
  year = {2018},
  volume = {9},
  number = {1},
  pages = {4950},
  doi = {10.1038/s41467-018-07210-0},
}

@incollection{chertock2017,
  author = {Alina Chertock},
  title = {{A practical guide to deterministic particle methods}},
  booktitle = {Handbook of Numerical Analysis},
  volume = {18},
  pages = {177--202},
  publisher = {Elsevier},
  year = {2017},
  doi = {10.1016/bs.hna.2016.11.004},
}

@article{maoutsa2020,
  author = {Dimitra Maoutsa and Sebastian Reich and Manfred Opper},
  title = {{Interacting particle solutions of Fokker--Planck equations through gradient-log-density estimation}},
  journal = {Entropy},
  year = {2020},
  volume = {22},
  number = {8},
  pages = {802},
  doi = {10.3390/e22080802},
}

@article{mardt2018,
  author = {Andreas Mardt and Luca Pasquali and Hao Wu and Frank No{\'e}},
  title = {{VAMPnets for deep learning of molecular kinetics}},
  journal = {Nature Communications},
  year = {2018},
  volume = {9},
  number = {1},
  pages = {5},
  doi = {10.1038/s41467-017-02388-1},
}

@inproceedings{meanti2023,
  author = {Giacomo Meanti and Antoine Chatalic and Vladimir Kostic and Pietro Novelli and Massimiliano Pontil and Lorenzo Rosasco},
  title = {{Estimating Koopman operators with sketching to provably learn large scale dynamical systems}},
  booktitle = {Advances in Neural Information Processing Systems},
  year = {2023},
  volume = {36},
  pages = {77242--77276},
  publisher = {Curran Associates, Inc.},
}

@incollection{mikulincer2023,
  author = {Dan Mikulincer and Yair Shenfeld},
  title = {{On the Lipschitz properties of transportation along heat flows}},
  booktitle = {Geometric Aspects of Functional Analysis},
  year = {2023},
  volume = {2327},
  pages = {269--290},
  series = {Lecture Notes in Mathematics},
  publisher = {Springer},
  doi = {10.1007/978-3-031-26300-2_9},
  eprint = {2201.01382},
  archivePrefix = {arXiv},
}

@article{neeman2022,
  author = {Joe Neeman},
  title = {{Lipschitz changes of variables via heat flow}},
  journal = {arXiv preprint arXiv:2201.03403},
  year = {2022},
  eprint = {2201.03403},
  archivePrefix = {arXiv},
}

@article{noe2013,
  author = {Frank No{\'e} and Feliks N{\"u}ske},
  title = {{A variational approach to modeling slow processes in stochastic dynamical systems}},
  journal = {Multiscale Modeling \& Simulation},
  year = {2013},
  volume = {11},
  number = {2},
  pages = {635--655},
  doi = {10.1137/110858616},
}

@article{nuske2017,
  author = {Feliks N{\"u}ske and Hao Wu and Jan-Hendrik Prinz and Christoph Wehmeyer and Cecilia Clementi and Frank No{\'e}},
  title = {{Markov state models from short non-equilibrium simulations---analysis and correction of estimation bias}},
  journal = {The Journal of Chemical Physics},
  year = {2017},
  volume = {146},
  number = {9},
  pages = {094104},
  doi = {10.1063/1.4976518},
}

@article{nuske2023,
  author = {Feliks N{\"u}ske and Sebastian Peitz and Friedrich Philipp and Manuel Schaller and Karl Worthmann},
  title = {{Finite-data error bounds for Koopman-based prediction and control}},
  journal = {Journal of Nonlinear Science},
  year = {2023},
  volume = {33},
  number = {1},
  pages = {14},
  doi = {10.1007/s00332-022-09862-1},
}

@article{otto2001,
  author = {Felix Otto},
  title = {{The geometry of dissipative evolution equations: the porous medium equation}},
  journal = {Communications in Partial Differential Equations},
  year = {2001},
  volume = {26},
  number = {1--2},
  pages = {101--174},
  doi = {10.1081/PDE-100002243},
}

@book{pavliotis2014,
  author = {Grigorios A. Pavliotis},
  title = {{Stochastic Processes and Applications: Diffusion Processes, the Fokker--Planck and Langevin Equations}},
  year = {2014},
  volume = {60},
  series = {Texts in Applied Mathematics},
  publisher = {Springer},
  doi = {10.1007/978-1-4939-1323-7},
}

@article{peyre2018,
  author = {R{\'e}mi Peyre},
  title = {{Comparison between $W_2$ distance and $\dot H^{-1}$ norm, and localization of Wasserstein distance}},
  journal = {ESAIM: Control, Optimisation and Calculus of Variations},
  year = {2018},
  volume = {24},
  number = {4},
  pages = {1489--1501},
  doi = {10.1051/cocv/2017050},
}

@inproceedings{rozen2021,
  author = {Noam Rozen and Aditya Grover and Maximilian Nickel and Yaron Lipman},
  title = {{Moser flow: divergence-based generative modeling on manifolds}},
  booktitle = {Advances in Neural Information Processing Systems},
  year = {2021},
  volume = {34},
  pages = {17669--17680},
  publisher = {Curran Associates, Inc.},
}

@article{russo1990,
  author = {Giovanni Russo},
  title = {{Deterministic diffusion of particles}},
  journal = {Communications on Pure and Applied Mathematics},
  year = {1990},
  volume = {43},
  number = {6},
  pages = {697--733},
  doi = {10.1002/cpa.3160430602},
}

@article{scarvelis2023,
  author = {Christopher Scarvelis and {S{\'a}ez de Oc{\'a}riz Borde}, Haitz and Justin Solomon},
  title = {{Closed-form diffusion models}},
  journal = {Transactions on Machine Learning Research},
  year = {2025},
  eprint = {2310.12395},
  archivePrefix = {arXiv},
}

@inproceedings{shen2022,
  author = {Zebang Shen and Zhenfu Wang and Satyen Kale and Alejandro Ribeiro and Amin Karbasi and Hamed Hassani},
  title = {{Self-consistency of the Fokker Planck equation}},
  booktitle = {Proceedings of Thirty Fifth Conference on Learning Theory},
  year = {2022},
  volume = {178},
  pages = {817--841},
  series = {Proceedings of Machine Learning Research},
  publisher = {PMLR},
}

@article{shen2024,
  author = {Zheyang Shen and Huihui Wang and Marina Riabiz and Chris J. Oates},
  title = {{Operator-informed score matching for Markov diffusion models}},
  journal = {arXiv preprint arXiv:2406.09084v2},
  year = {2025},
  eprint = {2406.09084v2},
  archivePrefix = {arXiv},
}

@inproceedings{song2021,
  author = {Yang Song and Jascha Sohl-Dickstein and Diederik P. Kingma and Abhishek Kumar and Stefano Ermon and Ben Poole},
  title = {{Score-based generative modeling through stochastic differential equations}},
  booktitle = {International Conference on Learning Representations},
  year = {2021},
}

@book{villani2003,
  author = {C{\'e}dric Villani},
  title = {{Topics in Optimal Transportation}},
  year = {2003},
  volume = {58},
  series = {Graduate Studies in Mathematics},
  publisher = {American Mathematical Society},
  doi = {10.1090/gsm/058},
}

@article{wang2025,
  author = {Ziyue Wang and Yuko Araki},
  title = {{Functional time series forecasting of distributions: a Koopman--Wasserstein approach}},
  journal = {Behaviormetrika},
  year = {2026},
  volume = {53},
  number = {2},
  pages = {571--597},
  doi = {10.1007/s41237-025-00278-1},
  eprint = {2507.07570},
  archivePrefix = {arXiv},
}

@article{williams2015,
  author = {Matthew O. Williams and Ioannis G. Kevrekidis and Clarence W. Rowley},
  title = {{A data-driven approximation of the Koopman operator: extending dynamic mode decomposition}},
  journal = {Journal of Nonlinear Science},
  year = {2015},
  volume = {25},
  number = {6},
  pages = {1307--1346},
  doi = {10.1007/s00332-015-9258-5},
}

@article{wu2020,
  author = {Hao Wu and Frank No{\'e}},
  title = {{Variational approach for learning Markov processes from time series data}},
  journal = {Journal of Nonlinear Science},
  year = {2020},
  volume = {30},
  number = {1},
  pages = {23--66},
  doi = {10.1007/s00332-019-09567-y},
}

@article{wu2017,
  author = {Hao Wu and Feliks N{\"u}ske and Fabian Paul and Stefan Klus and P{\'e}ter Koltai and Frank No{\'e}},
  title = {{Variational Koopman models: slow collective variables and molecular kinetics from short off-equilibrium simulations}},
  journal = {The Journal of Chemical Physics},
  year = {2017},
  volume = {146},
  number = {15},
  pages = {154104},
  doi = {10.1063/1.4979344},
}

@article{xu2025,
  author = {Yuanchao Xu and Fengyi Li and Masahiro Fujisawa and Xiaoyuan Cheng and Youssef Marzouk and Isao Ishikawa},
  title = {{Generative modeling through Koopman spectral analysis: an operator-theoretic perspective}},
  journal = {arXiv preprint arXiv:2512.18837},
  year = {2025},
  eprint = {2512.18837},
  archivePrefix = {arXiv},
}

@article{xu2025sdmd,
  author = {Yuanchao Xu and Kaidi Shao and Isao Ishikawa and Yuka Hashimoto and Nikos Logothetis and Zhongwei Shen},
  title = {{A data-driven framework for Koopman semigroup estimation in stochastic dynamical systems}},
  journal = {Chaos: An Interdisciplinary Journal of Nonlinear Science},
  year = {2025},
  volume = {35},
  number = {10},
  pages = {103123},
  doi = {10.1063/5.0283640},
}

@article{zhang2022koopman,
  author = {Benjamin J. Zhang and Tuhin Sahai and Youssef M. Marzouk},
  title = {{A Koopman framework for rare event simulation in stochastic differential equations}},
  journal = {Journal of Computational Physics},
  year = {2022},
  volume = {456},
  pages = {111025},
  doi = {10.1016/j.jcp.2022.111025},
}

@article{zhang2022eigen,
  author = {Wei Zhang and Tiejun Li and Christof Sch{\"u}tte},
  title = {{Solving eigenvalue PDEs of metastable diffusion processes using artificial neural networks}},
  journal = {Journal of Computational Physics},
  year = {2022},
  volume = {465},
  pages = {111377},
  doi = {10.1016/j.jcp.2022.111377},
}

@article{zhao2023,
  author = {Meng Zhao and Lijian Jiang},
  title = {{Data-driven probability density forecast for stochastic dynamical systems}},
  journal = {Journal of Computational Physics},
  year = {2023},
  volume = {492},
  pages = {112422},
  doi = {10.1016/j.jcp.2023.112422},
  eprint = {2210.03418},
  archivePrefix = {arXiv},
}

@article{zhou2025,
  author = {Mo Zhou and Stanley Osher and Wuchen Li},
  title = {{Simulating Fokker--Planck equations via mean field control of score-based normalizing flows}},
  journal = {Journal of Computational Physics},
  year = {2026},
  volume = {566},
  pages = {115274},
  doi = {10.1016/j.jcp.2026.115274},
  eprint = {2506.05723},
  archivePrefix = {arXiv},
}

@inproceedings{turan2025,
  author = {Erkan Turan and Ari Siozopoulos and Louis Martinez and Julien Gaubil and Emery Pierson and Maks Ovsjanikov},
  title = {{Unfolding generative flows with Koopman operators: trajectory-preserving linearization}},
  booktitle = {Proceedings of the 43rd International Conference on Machine Learning},
  series = {Proceedings of Machine Learning Research},
  volume = {306},
  year = {2026},
  eprint = {2506.22304},
  archivePrefix = {arXiv},
}

@article{teolis2026,
  author = {Trevor Teolis and Maarten V. de Hoop},
  title = {{Target-adapted Green--Bessel SVGD: uniform-in-time propagation of chaos and last-iterate consistency}},
  journal = {arXiv preprint arXiv:2609.08122},
  year = {2026},
  eprint = {2609.08122},
  archivePrefix = {arXiv},
}

@inproceedings{chen2010herding,
  author = {Yutian Chen and Max Welling and Alexander J. Smola},
  title = {{Super-samples from kernel herding}},
  booktitle = {Proceedings of the Twenty-Sixth Conference on Uncertainty in Artificial Intelligence},
  year = {2010},
  pages = {109--116},
}

@article{mak2018,
  author = {Simon Mak and V. Roshan Joseph},
  title = {{Support points}},
  journal = {The Annals of Statistics},
  year = {2018},
  volume = {46},
  number = {6A},
  pages = {2562--2592},
  doi = {10.1214/17-AOS1629},
}

@article{dwivedi2024,
  author = {Raaz Dwivedi and Lester Mackey},
  title = {{Kernel thinning}},
  journal = {Journal of Machine Learning Research},
  year = {2024},
  volume = {25},
  number = {152},
  pages = {1--77},
}

@inproceedings{liu2018,
  author = {Qiang Liu and Dilin Wang},
  title = {{Stein variational gradient descent as moment matching}},
  booktitle = {Advances in Neural Information Processing Systems},
  year = {2018},
  volume = {31},
  publisher = {Curran Associates, Inc.},
}

@book{risken1989,
  author = {Hannes Risken},
  title = {{The Fokker--Planck Equation: Methods of Solution and Applications}},
  edition = {second},
  series = {Springer Series in Synergetics},
  volume = {18},
  publisher = {Springer},
  address = {Berlin, Heidelberg},
  year = {1989},
  doi = {10.1007/978-3-642-61544-3},
}

@article{choi2018,
  author = {Gary P. T. Choi and Chris H. Rycroft},
  title = {{Density-equalizing maps for simply connected open surfaces}},
  journal = {SIAM Journal on Imaging Sciences},
  year = {2018},
  volume = {11},
  number = {2},
  pages = {1134--1178},
  doi = {10.1137/17M1124796},
}

@article{becker2001,
  author = {Roland Becker and Rolf Rannacher},
  title = {{An optimal control approach to a posteriori error estimation in finite element methods}},
  journal = {Acta Numerica},
  year = {2001},
  volume = {10},
  pages = {1--102},
  doi = {10.1017/S0962492901000010},
}

@book{glasserman2004,
  author = {Paul Glasserman},
  title = {{Monte Carlo Methods in Financial Engineering}},
  series = {Applications of Mathematics},
  volume = {53},
  publisher = {Springer},
  address = {New York},
  year = {2004},
  doi = {10.1007/978-0-387-21617-1},
}

@inproceedings{arbel2019,
  author = {Michael Arbel and Anna Korba and Adil Salim and Arthur Gretton},
  title = {{Maximum mean discrepancy gradient flow}},
  booktitle = {Advances in Neural Information Processing Systems},
  volume = {32},
  year = {2019},
  publisher = {Curran Associates, Inc.},
  eprint = {1906.04370},
  archivePrefix = {arXiv},
}

@book{vandervaart1998,
  author = {A. W. van der Vaart},
  title = {{Asymptotic Statistics}},
  series = {Cambridge Series in Statistical and Probabilistic Mathematics},
  volume = {3},
  publisher = {Cambridge University Press},
  address = {Cambridge},
  year = {1998},
  doi = {10.1017/CBO9780511802256},
}

@article{bovier2005,
  author = {Anton Bovier and V{\'e}ronique Gayrard and Markus Klein},
  title = {{Metastability in reversible diffusion processes II: precise asymptotics for small eigenvalues}},
  journal = {Journal of the European Mathematical Society},
  year = {2005},
  volume = {7},
  number = {1},
  pages = {69--99},
  doi = {10.4171/JEMS/22},
}

@article{schutte2023,
  author = {Christof Sch{\"u}tte and Stefan Klus and Carsten Hartmann},
  title = {{Overcoming the timescale barrier in molecular dynamics: transfer operators, variational principles and machine learning}},
  journal = {Acta Numerica},
  year = {2023},
  volume = {32},
  pages = {517--673},
  doi = {10.1017/S0962492923000016},
}
\bibliographystyle{iclr2027_conference}

\newpage
\appendix
\renewcommand{\topfraction}{0.9}\renewcommand{\bottomfraction}{0.8}\renewcommand{\textfraction}{0.07}\renewcommand{\floatpagefraction}{0.75}\setcounter{topnumber}{3}\setcounter{totalnumber}{5}
\section{Assumptions and preliminary results}\label{app:assump}

Throughout, $\pi\in\mathcal P_2(\mathbb R^d)$ has a positive $C^2$ density proportional to $e^{-V}$, $L=-\Delta+\nabla V\cdot\nabla$, so that $\langle f,Lg\rangle_\pi=\int\nabla f\cdot\nabla g\,d\pi$, and $P_t=e^{-tL}$. For $g$ with $\langle g,1\rangle_\pi=0$ we write $\|g\|_{\dot H^{-1}}=\sup\{\langle g,f\rangle_\pi\mid\|f\|_{\dot H^1}\le1\}$, and $\|\mu-\pi\|_{\dot H^{-1}}=\|d\mu/d\pi-1\|_{\dot H^{-1}}$ for $\mu\ll\pi$. Section~\ref{sec:theory} uses the following assumptions.
\begin{itemize}
\item[(A0)] \emph{Regularity.} The semigroup $P_t$ has a positive symmetric kernel $p_t(x,y)$ of class $C^2$ in $x$ for $t>0$. The ratio $\rho_t=P_t\rho_0$ is $C^2$ and positive for $t>0$, $v_t=-\nabla\log\rho_t$ is locally Lipschitz, and $\mu_t=P_t^*\mu_0$ is the unique solution of $\partial_t\mu+\mathrm{div}(\mu v_t)=0$ with $\mu|_{t=0}=\mu_0$ among curves with $d\mu_t/d\pi\in L^\infty$ \citep[Prop.~8.1.8, Thm.~8.3.1]{ambrosio2008}.
\item[(A1)] \emph{Spectrum.} The operator $L$ has discrete spectrum $0=\lambda_0<\lambda_1\le\lambda_2\le\cdots$ with orthonormal eigenfunctions $\varphi_k$, $\varphi_0\equiv1$, and $\sum_k(1+\lambda_k)^{a}e^{-\lambda_kt}<\infty$ for all $t>0$ and $a\ge0$.
\item[(A2)] \emph{Source.} The density ratio $\rho_0=d\mu_0/d\pi$ lies in $L^2(\pi)\cap L^\infty(\pi)$, and $c_k=\langle\rho_0,\varphi_k\rangle_\pi$.
\item[(A3)] \emph{Domain.} The set $B\subset\mathbb R^d$ is compact and convex, $R_B=\sup_B|x|$, $m_B=\|\rho_0\|_\infty\int_{B^c}|x|^2d\pi$, and $\sup_B\|\nabla^j\varphi_k\|\le K_B(1+\lambda_k)^{\alpha+j/2}$ for $j=0,1,2$.
\item[(A4)] \emph{Exact flow.} On a neighbourhood of the sets $G_t$ of Definition~\ref{def:good}, $-\nabla^2\log\rho_t\preceq L^+_tI$ with $L^+\in L^1(0,T)$, $|\nabla\log\rho_t|\le V_G$ and $\|\nabla^2\log\rho_t\|\le H_G$, and we set $\Lambda_{s,t}\coloneqq\int_s^t(L^+_u)_+du$. The one-sided bound holds on the whole space for an Ornstein-Uhlenbeck target with a Gaussian source, but not for a source of compact support, since outside the support $-\nabla^2\log\rho_t\approx I/(2t)$ as $t\to0$. The locality of Theorem~\ref{thm:main} is therefore imposed by the exact flow.
\end{itemize}

\paragraph{The confined particle system.}
The approximate ratio $\hat\rho$ is a function on $[0,T]\times\mathbb R^d$ of class $C^2$ in $x$, deterministic or fixed by conditioning on data independent of the transported sample, namely the trajectory data of the eigenpairs in regime (P) and, in regime (I), also the coefficient sample. We set $e_t\coloneqq\hat\rho_t-\rho_t$ and $\eta_t\coloneqq e_t/\rho_t$. The field $\hat v_t$ equals $\mathrm{Clip}_{v_{\max}}(-\nabla\log\hat\rho_t)$ where $\hat\rho_t\ge\varepsilon_0$, and elsewhere it is bounded and locally Lipschitz, as the field of Remark~\ref{rem:safeguard}. The particle $\hat X_t(x)$ solves $\dot{\hat X}=\hat v_t(\hat X)$ with reflection at $\partial B$ \citep{lions1984}, starting from $x$ if $x\in B$ and otherwise from a point of $B$ chosen independently of the sample, and $\hat\mu_T=M^{-1}\sum_j\delta_{\hat X_T(x_j)}$. A computation without confinement coincides with this system for any domain that its particles do not reach, and the componentwise projection of Appendix~\ref{app:comp} is a time discretisation of the reflection. Since $\hat\lambda_k>0$ for $k\ge1$, $\hat\rho_t\to1$ uniformly on $B$, and the threshold $\varepsilon_0$ is inactive after a finite time. We refer to the conditions of Definition~\ref{def:good} as (G1) $\rho_t\ge\theta_t$, (G2) $|e_t|\le\frac12\rho_t$ and (G3) $|\nabla\log\hat\rho_t|\le v_{\max}$. By (G2), $\hat\rho_t\ge\frac12\theta_t\ge\varepsilon_0$ on $G_t$, so that $\hat v_t=-\nabla\log\hat\rho_t$ there.

\paragraph{Preliminary results.}
\begin{proposition}\label{prop:rep}
Let $\rho_0=d\mu_0/d\pi\in L^2(\pi)$ and let $\mu_t$ be the law of $X_t$ when $X_0\sim\mu_0$. Then $d\mu_t/d\pi=P_t\rho_0=1+\sum_{k\ge1}e^{-\lambda_kt}c_k\varphi_k$ with $c_k=\mathbb E_{\mu_0}[\varphi_k]$, $\chi^2(\mu_t\|\pi)=\sum_{k\ge1}e^{-2\lambda_kt}c_k^2$ and $\rho_t(x)=\mathbb E_{y\sim\mu_0}[p_t(x,y)]$ with $p_t(x,y)=\sum_ke^{-\lambda_kt}\varphi_k(x)\varphi_k(y)$. Under (A0), $\mu_t=\Phi_t\#\mu_0$, and $\rho_t$ solves the backward Kolmogorov equation $\partial_t\rho_t=-L\rho_t$.
\end{proposition}
The first identity follows from $\int g\,d\mu_t=\int P_tg\,\rho_0\,d\pi=\int g\,P_t\rho_0\,d\pi$, and the others from Parseval's identity and (A0). The method takes its name from the last equation. The Lebesgue density $q_t=\pi\rho_t$ of $\mu_t$ solves the forward equation $\partial_tq_t=\Delta q_t+\mathrm{div}(q_t\nabla V)$, and $\nabla q_t+q_t\nabla V=\pi\nabla\rho_t$ gives $\pi\,\partial_t\rho_t=\mathrm{div}(\pi\nabla\rho_t)=-\pi L\rho_t$. The density ratio therefore evolves under the semigroup of the conditional expectations $(P_tg)(x)=\mathbb E[g(X_t)\mid X_0=x]$, which the estimators of Section~\ref{sec:setting} approximate from trajectories. The estimators of Section~\ref{sec:exp} are built from the Gram matrix of the dictionary and from the matrix of its Dirichlet form, which are averages over the sampled states. The dynamics enter through the Dirichlet form, that is, through the invariant law and the diffusion matrix, which together determine the generator of a reversible diffusion. For a nonreversible diffusion with the same invariant law and unit diffusion matrix, these estimators approximate the eigenpairs of the symmetric part of its generator, which is $-L$, and BKT still has $\pi$ as its target.

\begin{lemma}\label{lem:A2}
For $t\ge0$, $\|\mu_t-\pi\|_{\dot H^{-1}}^2=\sum_{k\ge1}e^{-2\lambda_kt}c_k^2/\lambda_k$ and $W_2(\mu_t,\pi)\le2\|\mu_t-\pi\|_{\dot H^{-1}}$.
\end{lemma}
\begin{proof}
For $f=\sum f_k\varphi_k$ one has $\langle\rho_t-1,f\rangle_\pi=\sum_{k\ge1}e^{-\lambda_kt}c_kf_k$, and the Cauchy-Schwarz inequality with the weights $\lambda_k$ gives the identity, with equality for $f_k\propto e^{-\lambda_kt}c_k/\lambda_k$. The inequality is that of \citet{peyre2018}, with $\pi$ as reference measure.
\end{proof}

\begin{lemma}\label{lem:A4}
Let $e^{\rm tr}_t=-\sum_{k>r}e^{-\lambda_kt}c_k\varphi_k$, $\tau_r=\|(I-\Pi_r)\rho_0\|$ and $m_r(t)=\sup_{\lambda\ge\lambda_{r+1}}\lambda e^{-2\lambda t}$. Then $\|e^{\rm tr}_t\|\le e^{-\lambda_{r+1}t}\tau_r$, $\|e^{\rm tr}_t\|_{\dot H^1}\le m_r(t)^{1/2}\tau_r$ and $\sup_B\|\nabla^je^{\rm tr}_t\|\le K_B\sum_{k>r}(1+\lambda_k)^{\alpha+j/2}e^{-\lambda_kt}|c_k|$ for $j\le2$.
\end{lemma}
The first two bounds follow from Parseval's identity and the third from (A3), and the last sum is finite by (A1).

\begin{lemma}\label{lem:endpoint}
Let $P,Q\in\mathcal P_2(\mathbb R^d)$ and let $P_M$ be the empirical measure of $M$ i.i.d.\ samples of $P$. Then $\mathbb E W_2(P_M,Q)\le W_2(P,Q)+\min\{S_M(P),S_M(Q)\}$.
\end{lemma}
\begin{proof}
Let $(U_j,Z_j)$ be i.i.d.\ pairs whose law is an optimal coupling of $P$ and $Q$, and let $Q_M$ be the empirical measure of the $Z_j$. Pairing equal indices and Jensen's inequality give $\mathbb E W_2(P_M,Q_M)\le W_2(P,Q)$, and the triangle inequality through $Q_M$ or through $P$ gives the claim.
\end{proof}

\begin{proof}[Proof of Lemma~\ref{lem:residual}]
We have $\mathrm{div}(\hat q_t\nabla\log\hat\rho_t)=\mathrm{div}(\pi\nabla\hat\rho_t)=-\pi L\hat\rho_t$ and $\partial_t\hat q_t=-\pi\sum_k\hat\lambda_ke^{-\hat\lambda_kt}\hat c_k\hat\varphi_k$, and the identity follows by subtraction. For exact eigenpairs $L\hat\varphi_k=\lambda_k\hat\varphi_k$, so that $R\equiv0$.
\end{proof}

\begin{proof}[Proof of Proposition~\ref{prop:exact}]
For exact eigenpairs $\hat\rho_t=1+\sum_{k=1}^re^{-\lambda_kt}a_k\varphi_k=P_t\hat\rho_0$. Since $P_t$ is a Markov operator with a positive kernel, $P_t(\hat\rho_0-\inf\hat\rho_0)\ge0$, which gives $\inf\hat\rho_t\ge\inf\hat\rho_0$, and $\hat\rho_t>0$ on $[0,T]\times\mathbb R^d$. The measure $\pi\hat\rho_0$ is a probability measure, since $\int\varphi_k\,d\pi=0$ for $k\ge1$, and by Proposition~\ref{prop:rep} $\pi\hat\rho_t$ is the law at time $t$ of the diffusion started from it. The field $\hat v_t=-\nabla\log\hat\rho_t$ is locally Lipschitz on $[0,T]\times\mathbb R^d$, since $\hat\rho$ is continuous and positive there and $\varphi_k\in C^2$, and $\int_0^T\!\int|\hat v_t|\hat\rho_t\,d\pi\,dt=\int_0^T\|\nabla\hat\rho_t\|_{L^1(\pi)}dt\le T\sum_k|a_k|\lambda_k^{1/2}$, since $\|\nabla\varphi_k\|^2=\lambda_k$. By Lemma~\ref{lem:residual} with $R\equiv0$, $\pi\hat\rho_t$ solves the continuity equation with velocity $\hat v_t$ on the whole space, and \citet[Prop.~8.1.8]{ambrosio2008} gives $\pi\hat\rho_t=\hat\Phi_{0,t}\#(\pi\hat\rho_0)$, with a flow defined $\pi\hat\rho_0$-almost everywhere, hence $\mu_0$-almost everywhere since $\hat\rho_0>0$. For population coefficients $\hat\rho_0=\Pi_r\rho_0$, and
\[
\hat\Phi_{0,T}\#\mu_0-\pi=\hat\Phi_{0,T}\#(\mu_0-\pi\hat\rho_0)+\pi(\hat\rho_T-1).
\]
The image of a signed measure under a measurable map has at most its total variation, $\|\mu_0-\pi\hat\rho_0\|_{\rm TV}=\|(I-\Pi_r)\rho_0\|_{L^1(\pi)}$, and $\|\hat\rho_T-1\|_{L^1(\pi)}\le\|\hat\rho_T-1\|=(\sum_{k=1}^re^{-2\lambda_kT}c_k^2)^{1/2}\le e^{-\lambda_1T}\|\Pi_r\rho_0-1\|$.
\end{proof}
For estimated eigenpairs, $\partial_t\hat\rho_t=-\sum_k\hat\lambda_ke^{-\hat\lambda_kt}\hat c_k\hat\varphi_k=-L\hat\rho_t+R_t$. When the $\hat\varphi_k$ lie in the domain of $L$ and the residuals $(L-\hat\lambda_k)\hat\varphi_k$ are bounded, the variation of constants formula gives $\hat\rho_t=P_t\hat\rho_0+\int_0^tP_{t-s}R_s\,ds$, and since $P_{t-s}$ is a Markov operator, $\inf\hat\rho_t\ge\inf\hat\rho_0-\int_0^t\|R_s\|_\infty\,ds$.

\section{Proof of the error bound and of its corollaries}\label{app:proofs}

\subsection{Proof of Theorem~\ref{thm:main}}\label{app:proofmain}

Fix $x$ and write $X_t=\Phi_t(x)$, $\hat X_t=\hat X_t(x)$, $D_t=|\hat X_t-X_t|$, $\sigma_1=\inf\{t\le T\mid X_t\notin G_t^\delta\}$, $\sigma_2=\inf\{t\le T\mid D_t>\delta\}$ and $\sigma=\sigma_1\wedge\sigma_2$. For $t<\sigma$ one has $X_t\in G_t^\delta$ and $\hat X_t\in\bar B_\delta(X_t)\subset G_t$, so that the reflection is inactive and the segment $[X_t,\hat X_t]$ lies in $G_t$, where $\hat v_t=-\nabla\log\hat\rho_t$. In particular, the sets $G_t$ need not be convex.

\begin{lemma}\label{lem:A8}
On $G_t$ one has $\hat v_t=v_t-\nabla\eta_t/(1+\eta_t)$,
\[
  |\hat v_t-v_t|\le2|\nabla\eta_t|\le\frac{2\big(|\nabla e_t|+V_G|e_t|\big)}{\rho_t},\qquad
  -\nabla^2\log\hat\rho_t\preceq\big(L^+_t+\mathrm{pert}_t\big)I ,
\]
and $J_p\le2\mathcal J_p$ for $p\in\{1,2\}$.
\end{lemma}
\begin{proof}
By (G2), $\hat\rho=\rho(1+\eta)$ with $1+\eta\ge\frac12$, $\nabla\eta=(\nabla e-e\nabla\log\rho)/\rho$ and $-\nabla^2\log(1+\eta)=-\nabla^2\eta/(1+\eta)+\nabla\eta\otimes\nabla\eta/(1+\eta)^2$. By (G1), $\int_{G_s}|\hat v_s-v_s|^2d\mu_s\le4\theta_s^{-1}\int(|\nabla e_s|+V_G|e_s|)^2d\pi$, while in $L^1(\mu_s)$ the factor $\rho_s$ cancels the denominator, and Minkowski's inequality gives $J_p\le2\mathcal J_p$.
\end{proof}

\begin{lemma}\label{lem:A10}
For $t<\sigma$ one has $D_t\le I(x)\coloneqq\int_0^Te^{\hat\Lambda_{s,T}}|\hat v_s-v_s|(X_s)\mathbf 1_{G_s}(X_s)\,ds$. Moreover, $\mathbb E_{\mu_0}I=J_1$, $(\mathbb E_{\mu_0}I^2)^{1/2}\le J_2$, and the event $E_{\rm bad}=\{\sigma\le T\}$ satisfies $\mu_0(E_{\rm bad})\le q_*$.
\end{lemma}
\begin{proof}
Along the segment $[X_t,\hat X_t]\subset G_t$ the Jacobian of $\hat v_t$ is symmetric and bounded above by $(L^+_t+\mathrm{pert}_t)I$ (Lemma~\ref{lem:A8}), so that $\frac d{dt}D_t\le(L^+_t+\mathrm{pert}_t)D_t+|\hat v_t-v_t|(X_t)$ where $D_t>0$, and Gr\"onwall's inequality with $D_0=0$ gives the first bound. Tonelli's theorem and Minkowski's integral inequality give the moments of $I$, since $X_s\sim\mu_s$. Outside the event $\{\sigma_1\le T\}$, of probability $p_{\rm exit}$, a separation larger than $\delta$ requires $I\ge\delta$, and Markov's inequality applied to $I$ and $I^2$ bounds its probability by $\min\{J_1/\delta,J_2^2/\delta^2\}$.
\end{proof}

\begin{proof}[Proof of Theorem~\ref{thm:main}]
Let $\widetilde\mu_T=\Phi_T\#\nu_0$. The points $\Phi_T(x_j)$ are i.i.d.\ with law $\mu_T$, so that the triangle inequality and Lemma~\ref{lem:endpoint} give
\[
 \mathbb E W_2(\hat\mu_T,\pi)\le\mathbb E W_2(\hat\mu_T,\widetilde\mu_T)+W_2(\mu_T,\pi)+\min\{S_M(\pi),S_M(\mu_T)\},
\]
and the sampling error of the source is not propagated by the flow. Since the field is independent of the sample, pairing equal indices gives $\mathbb E W_2^2(\hat\mu_T,\widetilde\mu_T)\le\mathbb E_{x\sim\mu_0}D_T(x)^2$, and it is this step that fails for same-sample coefficients. On $E_{\rm bad}^c$, $D_T\le I$ by Lemma~\ref{lem:A10}. On $E_{\rm bad}$, confinement gives $D_T^2\le2R_B^2+2|X_T|^2$, where $|X_T|\le R_B$ unless $X_T\notin B$ and $\mathbb E[|X_T|^2;X_T\notin B]\le m_B$ since $\|\rho_T\|_\infty\le\|\rho_0\|_\infty$. Hence $\mathbb E D_T^2\le J_2^2+4R_B^2q_*+2m_B=C_T^2$, and Lemma~\ref{lem:A2} gives the spectral bound on $W_2(\mu_T,\pi)$.
\end{proof}
For a field constructed from data $\mathcal D$ independent of the particles, the argument holds conditionally on $\mathcal D$, and all quantities in the bound depend on $\mathcal D$.

\begin{remark}[Dimension]\label{rem:dimension}
The bound does not depend explicitly on $d$. The dimension enters through the number of modes needed to represent the source, which is small when the source differs from $\pi$ only in slow coordinates, through $S_M(\pi)$ and through the estimation of the eigenpairs, for which kernel, neural and generator-learning estimators remain applicable in high dimension \citep{meanti2023,mardt2018,lusch2018,kostic2024,devergne2024,zhang2022eigen}. For the sliced distance $SW_2(\mu,\nu)^2\coloneqq|\Theta|^{-1}\sum_{\theta\in\Theta}W_2^2(\theta_\#\mu,\theta_\#\nu)\le W_2(\mu,\nu)^2$ of Section~\ref{sec:exp}, with $\theta_\#$ the projection on a unit vector $\theta$, the proof gives \eqref{eq:main} with $SW_2$ on the left and $S_M(\pi)$ replaced by the sliced error of an independent sample, of order $M^{-1/2}$ up to logarithmic factors in every dimension \citep{bobkov2019}.
\end{remark}

\subsection{Errors of the approximate density ratio}\label{app:cor}

\begin{corollary}\label{cor:exact}
Let the eigenpairs be exact and let the coefficients be computed from an i.i.d.\ sample $\nu'$ of $\mu_0$ of size $M'$, independent of the transported sample as in regime (I). Expectations below are over $\nu'$. Then $e_t=e^{\rm tr}_t+e^{\rm fl}_t$, with
$$
  \|e^{\rm tr}_t\|\le e^{-\lambda_{r+1}t}\tau_r,\qquad \|e^{\rm tr}_t\|_{\Hdot^1}\le m_r(t)^{1/2}\tau_r,
$$
$$
  \mathbb E\|e^{\rm fl}_t\|_{\Hdot^1}^2=\frac1{M'}\sum_{k\le r}\lambda_ke^{-2\lambda_kt}\mathrm{Var}_{\mu_0}(\varphi_k),\qquad \mathbb E\|e^{\rm fl}_t\|^2\le\frac{\mathbb E_{\mu_0}[p_{2t}(X,X)]}{M'},
$$
where $\tau_r\coloneqq\|(I-\Pi_r)\rho_0\|$ and $m_r(t)\coloneqq\sup_{\lambda\ge\lambda_{r+1}}\lambda e^{-2\lambda t}$. In regime (P), $e^{\rm fl}_t=0$.
\end{corollary}
\begin{proof}
The bounds on $e^{\rm tr}_t$ are Lemma~\ref{lem:A4}. For $e^{\rm fl}_t=\sum_{k\le r}e^{-\lambda_kt}(\hat c_k-c_k)\varphi_k$, Parseval's identity and $\mathbb E(\hat c_k-c_k)^2=\mathrm{Var}_{\mu_0}(\varphi_k)/M'$ give the identity, and $\sum_ke^{-2\lambda_kt}\mathrm{Var}_{\mu_0}(\varphi_k)\le\mathbb E_{\mu_0}[p_{2t}(X,X)]$ gives the second bound. The same argument with (A3) bounds $\sup_B\|\nabla^je_t\|$ for $j\le2$, as needed for the validity region, and it requires only an i.i.d.\ coefficient sample, so that it also applies to the source sample of regime (S).
\end{proof}

\begin{corollary}[Estimated eigenpairs]\label{cor:est}\label{cor:est_full}
Let the eigenpairs and the rank be fixed by data independent of the transported sample, with $\hat\lambda_k>0$ and $\int\hat\varphi_k\,d\pi=0$, and let $\hat c_k=\int\hat\varphi_k\,d\nu'$ as in regime (I). Write $\psi_k=\hat\varphi_k-\varphi_k$, $\tilde c_k=\int\varphi_k\,d\nu'$ and
\[
 a_k(t)\coloneqq|e^{-\hat\lambda_kt}-e^{-\lambda_kt}|\le t|\hat\lambda_k-\lambda_k|e^{-(\lambda_k\wedge\hat\lambda_k)t},\qquad
 A_{M'}\coloneqq\Big(\frac{\|\rho_0-1\|^2+\|\rho_0\|_\infty/M'}{\lambda_1}\Big)^{1/2}.
\]
Then $e_t=e^{\rm tr}_t+e^{\rm fl}_t+d_t+e^{\rm est}_t$, where the first two terms are those of Corollary~\ref{cor:exact} for the coefficients $\tilde c_k$, $d_t=\sum_{k=1}^re^{-\lambda_kt}(\hat c_k-\tilde c_k)\varphi_k$ and $e^{\rm est}_t=\sum_{k=1}^r\hat c_k(e^{-\hat\lambda_kt}\hat\varphi_k-e^{-\lambda_kt}\varphi_k)$, and
\[
\begin{aligned}
 \|e^{\rm est}_t\|_{\Hdot^1}&\le\sum_{k=1}^r|\hat c_k|e^{-\hat\lambda_kt}\|\psi_k\|_{\Hdot^1}+\Big(\sum_{k=1}^r\lambda_k|\hat c_k|^2a_k(t)^2\Big)^{1/2},\\
 \big(\mathbb E_{\nu'}\|d_t\|_{\Hdot^1}^2\big)^{1/2}&\le A_{M'}\Big(\sum_{k=1}^r\lambda_ke^{-2\lambda_kt}\|\psi_k\|_{\Hdot^1}^2\Big)^{1/2}.
\end{aligned}
\]
The bounds in $L^2(\pi)$ are analogous, and for population coefficients $M'=\infty$ and $e^{\rm fl}_t=0$.
\end{corollary}
\begin{proof}
The decomposition follows by adding and subtracting the two expansions in the exact basis. With $\beta_k=\mathbb E_{\mu_0}\psi_k=\langle\rho_0-1,\psi_k\rangle_\pi$ and $V_k=\mathrm{Var}_{\mu_0}(\psi_k)\le\|\rho_0\|_\infty\|\psi_k\|^2$, one has $\mathbb E_{\nu'}|\hat c_k-\tilde c_k|^2=\beta_k^2+V_k/M'\le A_{M'}^2\|\psi_k\|_{\Hdot^1}^2$ by the Poincar\'e inequality, and the orthogonality of the $\varphi_k$ in $L^2(\pi)$ and in $\Hdot^1$ gives the bound on $d_t$ without independence between modes. In $e^{\rm est}_t=\sum_k\hat c_ke^{-\hat\lambda_kt}\psi_k+\sum_k\hat c_k(e^{-\hat\lambda_kt}-e^{-\lambda_kt})\varphi_k$ the first sum is bounded by the triangle inequality and the second by orthogonality, and the bound on $a_k(t)$ is the mean value theorem.
\end{proof}

\subsection{Validity regions and the exit probability}\label{app:construct}

\begin{lemma}[Construction of a validity region]\label{lem:construct}
Suppose that for $t\le T$ there are a continuous function $w\ge1$ on $B$ and numbers $S_t>0$, $a_1,a_2$ such that $|\nabla^je_t|\le w\,a_jS_t$ on $B$ for $j=0,1,2$, with $a_0=1$. Let $\kappa\in(0,\tfrac12]$, $\kappa_t=\min\{\kappa,S_t/(2\varepsilon_0)\}$ and $G_t=\{x\in B_{2\delta}\mid w(x)S_t\le\kappa_t\rho_t(x)\}$, where $B_{2\delta}$ is the set of points of $B$ at distance at least $2\delta$ from its complement. Then $(G_t)$ is a validity region with $\theta_t\ge\max\{2\varepsilon_0,S_t/\kappa_t\}$ and $\mathrm{pert}_t\le2(a_2+2V_Ga_1+V_G^2+H_G)\kappa_t+4(a_1+V_G)^2\kappa_t^2$, and the bound $v_{\max}\ge V_G+2(a_1+V_G)\kappa$ suffices.
\end{lemma}
\begin{proof}
The set $G_t$ is closed and lies in $B_{2\delta}$. On $G_t$, $|e_t|\le wS_t\le\kappa_t\rho_t\le\rho_t/2$ and $\rho_t\ge S_t/\kappa_t\ge2\varepsilon_0$, which gives (G1) and (G2). Differentiating $\eta\rho=e$ twice and using $|\nabla^je|\le wa_jS_t\le a_j\kappa_t\rho_t$ gives $|\nabla\eta|\le(a_1+V_G)\kappa_t$ and $\|\nabla^2\eta\|\le(a_2+2V_Ga_1+V_G^2+H_G)\kappa_t$, hence the bound on $\mathrm{pert}_t$, and $|\nabla\log\hat\rho|\le V_G+2(a_1+V_G)\kappa_t$ gives (G3).
\end{proof}

\begin{lemma}\label{lem:exit}
(i) In dimension one, if $v_t$ is locally Lipschitz then $\Phi_t$ is increasing and $F_t\circ\Phi_t=F_0$, where $F_t$ is the distribution function of $\mu_t$. If $G_t^\delta=[a_t,b_t]$ is nonempty for all $t$, then $p_{\rm exit}\le\min\{1,\sup_tF_t(a_t^-)+\sup_t(1-F_t(b_t))\}$, with equality when the two families of tails are disjoint.
(ii) In any dimension, let $\bar v=\sup_{t\le T}\sup_B|v_t|$, $\Delta\le\delta/(2\bar v)$, $t_i=i\Delta$, and suppose that $G_{t_i}^{3\delta/2}\subset G_t^{\delta}$ for $t\in[t_i,t_{i+1}]$. Then $p_{\rm exit}\le\sum_{i\le T/\Delta}\mu_{t_i}\big((G_{t_i}^{2\delta})^c\big)+\mu_0(\{x\mid\Phi_t(x)\notin B\text{ for some }t\le T\})$.
\end{lemma}
\begin{proof}
In dimension one the flow is increasing and is the quantile coupling of optimal transport \citep[Ch.~2]{villani2003}, so that a particle exits $G_t^\delta$ at time $t$ if and only if $F_0(x)<F_t(a_t^-)$ or $F_0(x)>F_t(b_t)$, which gives (i). For (ii), a particle in $G_{t_i}^{2\delta}$ moves by at most $\delta/2$ on $[t_i,t_{i+1}]$ and therefore remains in $G_{t_i}^{3\delta/2}\subset G_t^\delta$, and a union bound over the grid gives the claim.
\end{proof}
In dimension one, part (i) gives $p_{\rm exit}$ exactly. In higher dimension the hypothesis of (ii) asks only that the region move by less than $\delta/2$ over one grid interval, which is a condition on the step $\Delta$.

\section{The transport identity and same-sample coefficients}\label{app:identity}

\begin{proof}[Proof of Proposition~\ref{prop:identity}]
Let $\hat q_t=\pi\hat\rho_t$ and $\hat\nu_t=\hat\Phi_{0,t}\#\nu$. As in the proof of Lemma~\ref{lem:residual}, and since $\hat q_t\hat v_t=\pi\gamma_t-\pi\nabla\hat\rho_t$,
\[
  \partial_t\hat q_t+\mathrm{div}(\hat q_t\hat v_t)=\pi R_t+\mathrm{div}(\pi\gamma_t),\qquad
  \partial_t\hat\nu_t+\mathrm{div}(\hat\nu_t\hat v_t)=0
\]
in the sense of distributions on the whole space, without any sign condition on $\hat\rho_t$. The function $g_s=f\circ\hat\Phi_{s,T}$ is bounded and Lipschitz and solves $\partial_sg_s+\hat v_s\cdot\nabla g_s=0$ with $g_T=f$. Testing both equations against $g_s$, after mollification in $x$, gives $\frac d{ds}\int g_s\,d\hat\nu_s=0$ and $\frac d{ds}\int g_s\hat\rho_s\,d\pi=\langle g_s,R_s\rangle_\pi-\int\nabla g_s\cdot\gamma_s\,d\pi$, and integration over $[0,T]$ gives \eqref{eq:identity}. Where the safeguards are inactive, $\hat v_t=-\nabla\hat\rho_t/\hat\rho_t$ and $\gamma_t=0$.
\end{proof}

\begin{proof}[Proof of Corollary~\ref{cor:tv}]
For $\nu=\mu_0$ and $|f|\le1$, the second term of \eqref{eq:identity} is $\int(f\circ\hat\Phi_{0,T})(\rho_0-\hat\rho_0)\,d\pi$, bounded by $\|\rho_0-\hat\rho_0\|_{L^1(\pi)}$, and $|\int f\hat\rho_T\,d\pi-\int f\,d\pi|\le\|\hat\rho_T-1\|_{L^1(\pi)}$. The residual term is at most $\int_0^T\|R_s\|_{L^1(\pi)}ds\le\sum_{k\le r}|\hat c_k|\,\|(L-\hat\lambda_k)\hat\varphi_k\|_{L^1(\pi)}/\hat\lambda_k$. If the safeguards are active, $f$ is also taken with $\mathrm{Lip}(f)\le1$, and the flux term is at most $\int_0^T\mathrm{Lip}(\hat\Phi_{s,T})\|\gamma_s\|_{L^1(\pi)}ds$.
\end{proof}

\subsection{Same-sample coefficients and asymptotic variances}\label{app:cancel}

\begin{proof}[Proof of Proposition~\ref{prop:cancel}]
Apply Proposition~\ref{prop:identity} with $\nu=\nu_0$, let $\tilde c_k$ be the coefficients of the field, with $\tilde c_0=1$, and let $\hat\Pi_rg=\sum_{k\le r}\langle g,\hat\varphi_k\rangle_\pi\hat\varphi_k$. By orthonormality $\int\hat\varphi_k\hat\rho_0\,d\pi=\tilde c_k$, so that $\int\hat\Pi_rg\,d(\nu_0-\pi\hat\rho_0)=\sum_{k=1}^r\langle g,\hat\varphi_k\rangle_\pi\big(M^{-1}\sum_j\hat\varphi_k(x_j)-\tilde c_k\big)$, which vanishes in regime (S), while $\int(I-\hat\Pi_r)g\,\hat\rho_0\,d\pi=0$ because $\hat\rho_0$ lies in the span of the $\hat\varphi_k$. If the $\hat\varphi_k$ have a Gram matrix $\Gamma\ne I$ in $L^2(\pi)$, for instance when they are normalised with respect to the trajectory data, and $\hat\Pi_r$ denotes the orthogonal projection onto their span, the same computation gives the additional term $b^\top(\Gamma^{-1}-I)\hat m$, where $b_k=\langle g,\hat\varphi_k\rangle_\pi$ and $\hat m_k=M^{-1}\sum_j\hat\varphi_k(x_j)$ for $0\le k\le r$, whose fluctuation is of order $\|\Gamma-I\|M^{-1/2}$.
\end{proof}

\begin{proof}[Proof of Theorem~\ref{thm:clt}]
Write $G(c,x)=f(\Phi^c_{0,T}(x))$, $P$ for $\mu_0$ and $P_M$ for $\nu_0$, so that $\int f\,d\hat\mu_T=P_MG(c,\cdot)$ with $c=\bar c$, $\check c$ or $\hat c$ in regimes (P), (I) and (S). By the law of large numbers $\hat c,\check c\to\bar c$ almost surely, and $\sqrt M(\hat c-\bar c)=\sqrt M(P_M-P)\hat\varphi$ is asymptotically normal. Since $|G(c,\cdot)-G(c',\cdot)|\le\kappa|c-c'|$ on the bounded set $U$ with $\kappa\in L^2(\mu_0)$, the class $\{G(c,\cdot)\mid c\in U\}$ is $\mu_0$-Donsker \citep[Example~19.7 and Theorem~19.5]{vandervaart1998}. On the event $\hat c\in U$, whose probability tends to one, $P(G(\hat c,\cdot)-\bar g)^2\le\|\kappa\|_{L^2(\mu_0)}^2|\hat c-\bar c|^2$, which tends to zero in probability, and asymptotic equicontinuity \citep[Lemma~19.24]{vandervaart1998} gives $(P_M-P)(G(\hat c,\cdot)-\bar g)=o_P(M^{-1/2})$. The differentiability of $m$ at $\bar c$ gives $m(\hat c)-m(\bar c)=\nabla m(\bar c)\cdot(\hat c-\bar c)+o_P(M^{-1/2})$. Hence
\[
P_MG(\hat c,\cdot)-m(\bar c)=(P_M-P)\big(\bar g+\nabla m(\bar c)\cdot\hat\varphi\big)+o_P(M^{-1/2}),
\]
and the central limit theorem gives $\sigma_S^2$. In regime (I) the same expansion holds with $(P_M-P)\hat\varphi$ replaced by the independent fluctuation $\check c-\bar c$, and in regime (P) the second term is absent.

For the derivative of $m$, apply Proposition~\ref{prop:identity} with $\nu=\mu_0$ to the field with coefficients $c\in U$. Since $\hat\rho^c_0-\hat\rho^{\bar c}_0=\sum_{k=1}^r(c_k-\bar c_k)\hat\varphi_k$ and, by orthonormality, $\rho_0-\hat\rho^{\bar c}_0=(I-\hat\Pi_r)\rho_0$,
\[
\begin{aligned}
m(c)={}&\langle f,1\rangle_\pi+\sum_{k=1}^re^{-\hat\lambda_kT}c_k\langle f,\hat\varphi_k\rangle_\pi+\langle G(c,\cdot),(I-\hat\Pi_r)\rho_0\rangle_\pi\\
&-\sum_{k=1}^r(c_k-\bar c_k)\langle G(c,\cdot),\hat\varphi_k\rangle_\pi+E^c_T(f).
\end{aligned}
\]
Subtracting the same expression at $\bar c$ and using $|\langle G(c,\cdot)-\bar g,h\rangle_\pi|\le\|\kappa\|\,\|h\|\,|c-\bar c|$ gives
\[
m(c)-m(\bar c)=\sum_{k=1}^r(c_k-\bar c_k)\big(e^{-\hat\lambda_kT}\langle f,\hat\varphi_k\rangle_\pi-\langle\bar g,\hat\varphi_k\rangle_\pi\big)+\epsilon(c)
\]
with $|\epsilon(c)|\le\big(\|\kappa\|\hat\tau_r+\ell_T(f)+o(1)\big)|c-\bar c|$,
so that $\beta=\nabla m(\bar c)+(\langle\bar g,\hat\varphi_k\rangle_\pi)_k$ satisfies $|\beta|\le|(e^{-\hat\lambda_kT}\langle f,\hat\varphi_k\rangle_\pi)_k|+\|\kappa\|\hat\tau_r+\ell_T(f)$, and Bessel's inequality bounds the first term by $e^{-\hat\lambda_1T}\|f\|$. Finally $\bar g+\nabla m(\bar c)\cdot\hat\varphi=(I-\hat\Pi_r)\bar g+\langle\bar g,1\rangle_\pi+\beta\cdot\hat\varphi$, where $\|\beta\cdot\hat\varphi\|_{L^2(\mu_0)}\le\|\rho_0\|_\infty^{1/2}|\beta|$ by orthonormality, and for exact eigenpairs $\|(I-\Pi_r)\bar g\|_{L^2(\mu_0)}^2\le\|\rho_0\|_\infty\|(I-\Pi_r)\bar g\|^2\le\|\rho_0\|_\infty\lambda_{r+1}^{-1}\|\bar g\|_{\Hdot^1}^2$.
\end{proof}
The hypotheses on $U$ hold on a torus, or for dictionaries with bounded derivatives of order at most two, since the safeguarded field is then Lipschitz in $x$ and in $c$, and the differentiability of $m$ holds when the safeguards are inactive along the paths or are replaced by smooth functions. In regime (P), $\sigma_P^2$ is the variance of $f$ under $\Phi^{\bar c}_{0,T}\#\mu_0$, and $\mathrm{Var}_\pi(f)\le\lambda_1^{-1}\|f\|_{\Hdot^1}^2$ by the Poincar\'e inequality. The centring $m(\bar c)$ differs from $\int f\,d\pi$ by $\int f(\hat\rho^{\bar c}_T-1)\,d\pi+\langle\bar g,(I-\hat\Pi_r)\rho_0\rangle_\pi+E^{\bar c}_T(f)$, by Proposition~\ref{prop:identity} with $\nu=\mu_0$, and this bias is common to the three regimes.

\paragraph{The energy of LAWGD.} With $m_k=M^{-1}\sum_j\hat\varphi_k(x_j)$ and $\mathcal E=\frac M2\sum_{k=1}^r\hat\lambda_k^{-1}m_k^2$, one has $\nabla_{x_i}\mathcal E=M\sum_k\hat\lambda_k^{-1}m_k\nabla_{x_i}m_k=\sum_k\hat\lambda_k^{-1}m_k\nabla\hat\varphi_k(x_i)$, which is minus the velocity of Remark~\ref{rem:lawgd}, and $\frac d{dt}\mathcal E=-\sum_i|\dot x_i|^2$. Let $\nu_M$ be the empirical measure of the particles. Since $\int\hat\varphi_k\,d\pi=0$, $\frac2M\mathcal E=\sum_{k=1}^r\hat\lambda_k^{-1}\big(\int\hat\varphi_k\,d(\nu_M-\pi)\big)^2$ is the squared maximum mean discrepancy between $\nu_M$ and $\pi$ for the kernel $\sum_{k=1}^r\hat\lambda_k^{-1}\hat\varphi_k\otimes\hat\varphi_k$, and for exact eigenpairs it is the truncation to the first $r$ modes of $\|d\nu_M/d\pi-1\|_{\Hdot^{-1}}^2$, defined by the same series.

\subsection{Same-sample coefficients in the 2-Wasserstein distance}\label{app:samesample}

\begin{corollary}[Comparison with same-sample coefficients]\label{cor:samesample}
Fix the eigenpairs and the rank independently of the transported sample, with $\hat\lambda_k>0$ and $\operatorname{Var}_{\mu_0}(\hat\varphi_k)<\infty$ for $k\le r$, let $\bar c_k=\mathbb E_{\mu_0}\hat\varphi_k$ and $\hat c_k=M^{-1}\sum_j\hat\varphi_k(x_j)$, and let $\hat\mu^S_T$ be the measure transported with the coefficients $\hat c$ from the same confined initialisation. If Theorem~\ref{thm:main} applies to the population field, both fields admit a common validity region with bounds $\theta_s$, $\hat\Lambda_{s,t}$ and $V_s$ (the latter for $|\nabla\log\hat\rho^{\bar c}_s|$), and $\|\hat\varphi_k\|_{\infty,B}+\|\nabla\hat\varphi_k\|_{\infty,B}\le K_r$, then
\begin{equation}\label{eq:samesample}
\mathbb E\,W_2(\hat\mu^S_T,\pi)\le\mathcal B_T+\mathcal A_T\Delta_M+2R_B\,\hat p_{\rm exit}^{1/2},\qquad
\mathcal A_T\coloneqq2K_r\int_0^Te^{\hat\Lambda_{s,T}}\frac{1+V_s}{\theta_s}\,ds,
\end{equation}
with $\Delta_M\coloneqq\sum_{k\le r}(\operatorname{Var}_{\mu_0}(\hat\varphi_k)/M)^{1/2}$, where $\mathcal B_T$ is the bound of Theorem~\ref{thm:main} for the population field and $\hat p_{\rm exit}$ the expected fraction of particles whose population-field trajectory leaves $G^\delta$, or separates from its same-sample counterpart by more than $\delta$, before $T$.
\end{corollary}
The term $\mathcal A_T\Delta_M$ is of order $M^{-1/2}$ for fixed $r$, and the last term is controlled by the mass that leaves the validity region. The corollary bounds the cost of same-sample coefficients, and their gain is the cancellation of Proposition~\ref{prop:cancel}.

\begin{proof}[Proof of Corollary~\ref{cor:samesample}]
Condition on the eigenpair data and fix the source sample. With $\delta_s=\hat\rho^{\hat c}_s-\hat\rho^{\bar c}_s=\sum_{k\le r}e^{-\hat\lambda_ks}(\hat c_k-\bar c_k)\hat\varphi_k$ and $\hat\rho^{\hat c}_s\ge\theta_s/2$ on the common region,
\[
\sup_{G_s}|\hat v^{\hat c}_s-\hat v^{\bar c}_s|=\sup_{G_s}\Big|\frac{\nabla\delta_s-\delta_s\nabla\log\hat\rho^{\bar c}_s}{\hat\rho^{\hat c}_s}\Big|\le\frac{2K_r(1+V_s)}{\theta_s}\sum_{k=1}^r|\hat c_k-\bar c_k|.
\]
For a pair of trajectories that remains in the region, the argument of Lemma~\ref{lem:A10} gives $|X^S_{T,j}-X^P_{T,j}|\le\mathcal A_T\sum_k|\hat c_k-\bar c_k|$, and the other pairs are at distance at most $2R_B$. Pairing the particles, the bound $\mathbb E|\hat c_k-\bar c_k|\le(\operatorname{Var}_{\mu_0}(\hat\varphi_k)/M)^{1/2}$ and the triangle inequality through the population measure give the claim. The comparison is pathwise, so that no independence between the field and the particles is needed.
\end{proof}

\subsection{Resolution of a fixed evaluation sample}\label{app:resolution}

In Section~\ref{sec:exp} the $M$ transported particles are compared with one fixed target sample $Y_1$ of size $M$, and the reference level is the mean distance between two independent target samples of size $M$, over ten pairs. To first order, the sliced distance between point clouds is the $L^2(du\,d\theta)$ norm of the difference of their projected quantile fluctuations about those of $\pi$, and independent fluctuations add in squares, so that $\mathbb E\,SW_2^2(X,Y_1)\approx\mathbb E\,SW_2^2(X,\pi)+\mathbb E\,SW_2^2(Y_1,\pi)$ for particles $X$ independent of $Y_1$ and $\mathbb E\,SW_2^2(Y,Y')\approx2\,\mathbb E\,SW_2^2(Y,\pi)$ for independent samples $Y,Y'$ of size $M$. Particles without fluctuation therefore have an error of about $2^{-1/2}$ times the reference level, particles with the fluctuation of an independent sample the reference level, and particles with twice this fluctuation on the modes that dominate the distance about $(3/2)^{1/2}$ times the reference level. By Theorem~\ref{thm:clt} these are the situations of regimes (S), (P) and (I) when the retained modes carry most of the fluctuation.

\section{The Ornstein-Uhlenbeck case}\label{app:ou}

Let $\pi=N(0,1)$, $\lambda_k=k$, $\varphi_k=He_k/\sqrt{k!}$ and $\mu_0=N(m_0,s_0^2)$ with $s_0<1$, $q\coloneqq1-s_0^2$, $m_t\coloneqq m_0e^{-t}$ and $\sigma_t^2\coloneqq1-qe^{-2t}$. Then $\rho_t=N(\cdot;m_t,\sigma_t^2)/N(\cdot;0,1)$, the velocity $v_t(x)=(\sigma_t^{-2}-1)x-m_t\sigma_t^{-2}$ is affine, the flow map $\Phi_{s,t}(x)=m_t+\sigma_t\sigma_s^{-1}(x-m_s)$ has Lipschitz constant at most $1/s_0$, $L^+_t=\sigma_t^{-2}-1$, and $W_2(\mu_t,\pi)=(m_t^2+(\sigma_t-1)^2)^{1/2}$.

\begin{proposition}\label{prop:coeff}
The coefficients of the source are $c_k=q^{k/2}He_k(m_0/\sqrt q)/\sqrt{k!}$, and $c_k^2\le e^{m_0^2/(2q)}q^k$ and $\tau_r^2\le e^{m_0^2/(2q)}q^{r+1}/(1-q)$.
\end{proposition}
The formula follows from the generating function $\sum_kt^k\mathbb E[He_k(X)]/k!=\exp\{tm_0-qt^2/2\}$, and the bounds from Cram\'er's inequality $|He_k(y)|\le\sqrt{k!}\,e^{y^2/4}$ \citep{indritz1961}. For $m_0=1.5$ the formula gives $\tau_r<10^{-2}$ first at $r=9$, $32$ and $102$ when $s_0=0.8$, $0.5$ and $0.3$.

\begin{theorem}[Explicit rates]\label{thm:ou}
Assume $s_0<1$, $b\ge|m_0|+3$, exact eigenpairs and population coefficients $c_k$ (regime (P)), threshold $\varepsilon_0$, bound $v_{\max}\ge V_B+C'(q,V_B)$ with $V_B\coloneqq b(s_0^{-2}-1)+|m_0|s_0^{-2}$, rank $r$ and a fixed horizon $T$. Then, with constants $C,c$ depending only on $(m_0,s_0,b,\delta,T)$,
\begin{equation}\label{eq:main-ou}
  \mathbb E\,W_2(\hat\mu_T,\pi)\le W_2(\mu_T,\pi)+S_M(\mu_T)+C(r+1)^{c}\,\mathcal T_r(\varepsilon_0)+\mathcal E_{\rm tail}(b),
\end{equation}
where $\mathcal T_r(\varepsilon_0)\coloneqq q^{\frac{r+1}{2}}\min\{\varepsilon_0^{-1/2},q^{-\frac{r+1}{4}}\}+b\,q^{\frac{r+1}{2(2-s_0^2)}}+b\,\varepsilon_0^{\frac1{1+q}}$ up to factors $e^{O(\sqrt{\log(1/\varepsilon_0)})}$ and $e^{O(\sqrt{r})}$, subpolynomial in $1/\varepsilon_0$ and in $q^{-r}$, and $\mathcal E_{\rm tail}(b)=O(e^{-(b-|m_0|-1)^2/4})$ up to a polynomial factor in $b$.
\end{theorem}
\begin{proof}[Sketch of proof]
Assumptions (A0) to (A4) hold on $B=[-b,b]$ with $K_B=e^{b^2/4}$, $\alpha=0$, $V_G\le V_B$ and $H_G\le s_0^{-2}-1$, by Cram\'er's inequality and the explicit flow. With $z=\sqrt qe^{-t}$ and $K'=e^{m_0^2/(4q)}$, Proposition~\ref{prop:coeff} gives the envelope bounds of Lemma~\ref{lem:construct} with $w(x)=e^{x^2/4}$, $S_t=K'z^{r+1}/(1-z)$, $a_1=O(\sqrt{r+1})$ and $a_2=O(r+1)$. For $\kappa=(r+1)^{-3/2}$ the lemma gives a validity region with $\mathrm{pert}_t=O((r+1)^{-1/2})$, so that $e^{\hat\Lambda_{0,T}}\le s_0^{-1}e^{CT/\sqrt{r+1}}$. Since $\log\rho_t(x)-x^2/4$ is a concave quadratic, the region is an interval whose half-width $w_t$ satisfies $w_t^2\ge\big((r+1)\log(1/z)-\tfrac32\log(r+1)-c\big)/\big(\tfrac1{2\sigma_t^2}-\tfrac14\big)$. Lemma~\ref{lem:exit}(i) and Gaussian tails then give $p_{\rm exit}=O\big((r+1)^{c}q^{(r+1)/(2-s_0^2)}+\varepsilon_0^{2/(1+q)}+e^{-(b-|m_0|-1)^2/2}\big)$, where the term in $\varepsilon_0$ comes from the set $\{\rho_t<2\varepsilon_0e^{x^2/4}\}$. Lemma~\ref{lem:A4} and $\theta_t\ge\max\{2\varepsilon_0,S_t(r+1)^{3/2}\}$ give $J_2=O\big(\tau_r\min\{\varepsilon_0^{-1/2}(r+1)^{-1/2},(r+1)^{-5/4}q^{-(r+1)/4}\}\big)$. Inserting these bounds into Theorem~\ref{thm:main} with $q_*\le p_{\rm exit}+J_2^2/\delta^2$ and $S_M(\mu_T)=\sigma_TS_M(\pi)$ gives \eqref{eq:main-ou}, and the factors $e^{O(\sqrt{\log(1/\varepsilon_0)})}$ and $e^{O(\sqrt r)}$ come from the shift of the centre and the margin in the Gaussian tail.
\end{proof}
The per-mode ratio $\max\{q^{1/4},q^{1/(2(2-s_0^2))}\}$ equals $0.93$ for $s_0=0.5$. It comes from the conversion of $L^2(\mu_s)$ norms to $L^2(\pi)$ norms at the edge of the validity region. In the bound of Corollary~\ref{cor:tv} the truncation term decays with the sharper ratio $q^{1/2}=0.87$ per mode, since $\|(I-\Pi_r)\rho_0\|_{L^1(\pi)}\le\tau_r$, and the flux of the safeguards is the remaining term. With coefficients from a finite sample, regime (S) is covered by Proposition~\ref{prop:cancel}, Theorem~\ref{thm:clt} and Corollary~\ref{cor:samesample}, and the rank dependence observed in Section~\ref{sec:exp} is that of same-sample BKT. For $\pi=N(0,I_d)$ and a product Gaussian source all objects tensorise, and the number of modes needed for a given accuracy is the product of the numbers needed in each coordinate. When source and target factorise, the coordinate-wise application of BKT in the product experiments avoids this product.

\subsection{A matched comparison of the coefficient regimes}\label{app:ou_path_residual}

For the Ornstein-Uhlenbeck flow the deviation of computed paths from the exact ones is explicit. For points $x_j$ and absolutely continuous paths $Y_j$ with $Y_j(0)=x_j$, the residual $\mathcal R_j(t)=\dot Y_j(t)-v_t(Y_j(t))$ satisfies $Y_j(T)-\Phi_T(x_j)=\int_0^TK(s,T)\mathcal R_j(s)\,ds$ with $K(s,T)=\sigma_T/\sigma_s\le\max\{1,s_0^{-1}\}$, since $v_t$ is affine. With $\nu_0$ and $\hat\mu_T$ the empirical measures of the $x_j$ and of the $Y_j(T)$, this gives
\[
 W_2(\hat\mu_T,\pi)\le W_2(\Phi_T\#\nu_0,\pi)+D_T\le W_2(\Phi_T\#\nu_0,\pi)+\bar D_T,\quad
 \bar D_T=\int_0^TK(s,T)\|\mathcal R(s)\|_M\,ds,
\]
where $\|z\|_M=(M^{-1}\sum_j|z_j|^2)^{1/2}$ and $D_T=\|Y(T)-\Phi_T(x)\|_M$, for any paths and in particular for those of regime (S). For the Runge-Kutta paths the residual is evaluated on the piecewise linear interpolant, and it includes the effects of the step, of the threshold and of the velocity bound. The one-dimensional error against the Gaussian target is computed by the exact formula for the monotone coupling on the quantile intervals.

\paragraph{Matched experiment.} The parameters are $m_0=1.5$, $s_0=0.5$, $T=6$ and $M=2000$, with the analytical Hermite eigenpairs for $r\in\{10,20,40\}$, RK4 steps $h\in\{0.05,0.025\}$ and ten realisations of the source, matched across ranks, steps and regimes. In (P) the coefficients are the population values, in (I) they are estimated from an independent source sample of size $M'=2000$, and in (S) from the $M$ transported points. Within each realisation the transported sample, the coefficient sample and a target reference cloud of the same size are independent. The threshold and the velocity bound are those of Appendix~\ref{app:hyp}, dispersions are sample standard deviations over the ten realisations, and comparisons between regimes or steps are paired within a realisation.

\paragraph{Results.} Table~\ref{tab:ou_path_validation} reports the finer step and Figure~\ref{fig:ou_path_comparison} both steps. The exact-flow value $W_2(\Phi_T\#\nu_0,\pi)$ of the source clouds is $0.0448\pm0.0152$ and the independent-target reference $0.0391\pm0.0096$, and the population distance is $W_2(\mu_6,\pi)=0.0037$. At $h=0.025$, increasing $r$ from $10$ to $40$ reduces the distance to the target from $0.102$ to $0.045$ in (P), from $0.110$ to $0.061$ in (I) and from $0.096$ to $0.023$ in (S). At $r=40$ the population coefficients track the exact endpoints most closely ($D_T=0.0015\pm0.0008$ against $0.033\pm0.012$ for (I) and $0.041\pm0.016$ for (S)). The ordering (S) $<$ (P) $<$ (I) at every rank agrees with Theorem~\ref{thm:clt}, in which $\sigma_I\ge\sigma_P$ and $\sigma_S$ is small when the retained modes carry most of the fluctuation. The same-sample field corrects the sampling error of the retained modes of the source, so that its trajectories depart most from the exact paired trajectories and end closest to the target. The bound $W_2(\Phi_T\#\nu_0,\pi)+\bar D_T$ has the right scale ($0.065$, $0.095$, $0.099$ against measured $0.045$, $0.061$, $0.023$ at $r=40$). Halving the step changes the error by at most $0.001$ at $r=20$ and $40$, and by about $0.015$ at $r=10$, where the threshold and the velocity bound are active, while the population computations at $r=40$ activate neither.

\begin{table}[t]
\centering
\small
\caption{Matched Ornstein-Uhlenbeck computations at $h=0.025$, error to the target, paired endpoint discrepancy $D_T$ and residual-norm bound $\bar D_T$ (numerical quadrature), mean $\pm$ sd over ten matched realisations, for population (P), independent-sample (I) and same-sample (S) coefficients.}
\label{tab:ou_path_validation}
\begin{tabular}{ccrrr}
\toprule
$r$ & coefficients & $W_2$ & $D_T$ & $\bar D_T$ \\
\midrule
10 & P & $0.1020 \pm 0.0180$ & $0.0953 \pm 0.0058$ & $0.2529 \pm 0.0203$ \\
10 & I & $0.1104 \pm 0.0144$ & $0.1011 \pm 0.0108$ & $0.2664 \pm 0.0260$ \\
10 & S & $0.0961 \pm 0.0094$ & $0.1020 \pm 0.0096$ & $0.2563 \pm 0.0168$ \\
20 & P & $0.0533 \pm 0.0135$ & $0.0249 \pm 0.0029$ & $0.0433 \pm 0.0062$ \\
20 & I & $0.0668 \pm 0.0151$ & $0.0415 \pm 0.0116$ & $0.0590 \pm 0.0129$ \\
20 & S & $0.0382 \pm 0.0065$ & $0.0456 \pm 0.0164$ & $0.0629 \pm 0.0191$ \\
40 & P & $0.0447 \pm 0.0151$ & $0.0015 \pm 0.0008$ & $0.0200 \pm 0.0006$ \\
40 & I & $0.0609 \pm 0.0170$ & $0.0330 \pm 0.0119$ & $0.0504 \pm 0.0146$ \\
40 & S & $0.0227 \pm 0.0045$ & $0.0411 \pm 0.0163$ & $0.0545 \pm 0.0152$ \\
\bottomrule
\end{tabular}
\end{table}

\begin{figure}[t]
\centering
\includegraphics[width=\textwidth]{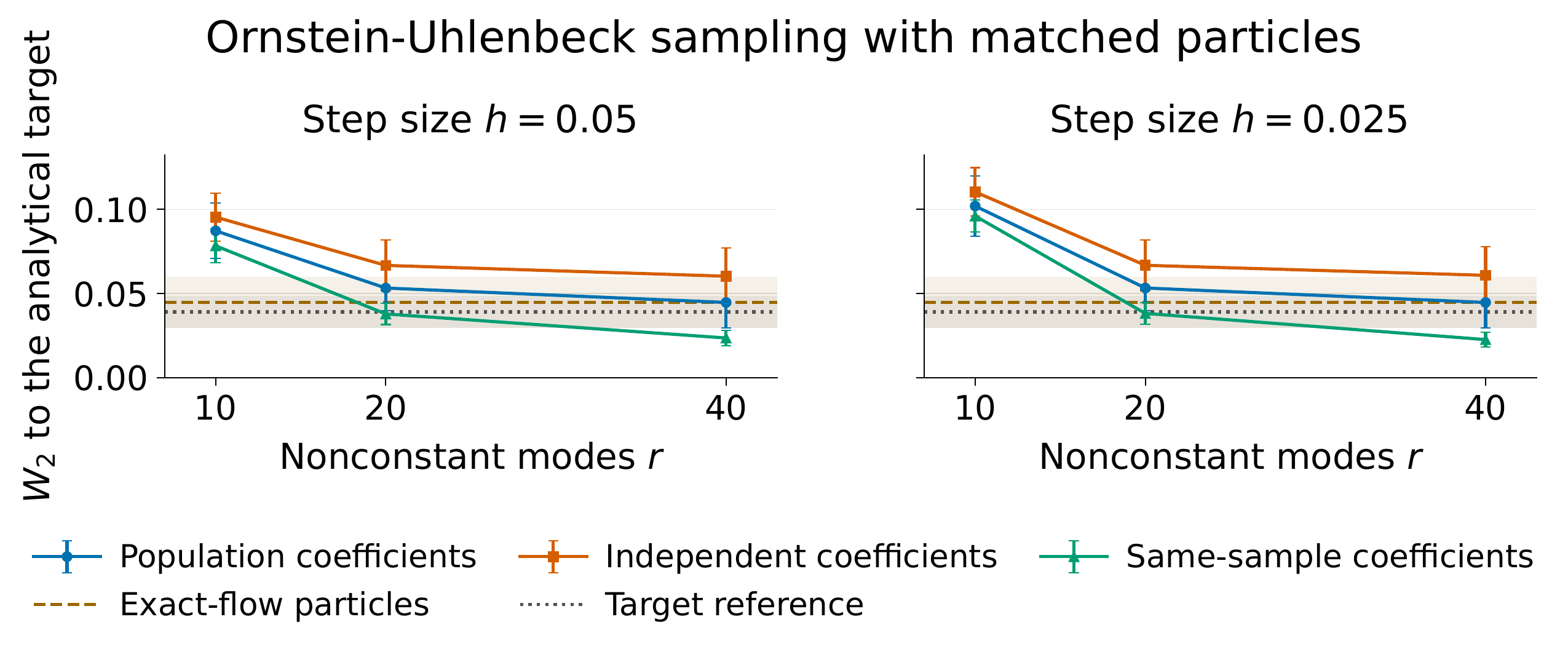}
\caption{Matched Ornstein-Uhlenbeck sampling with analytical eigenpairs. Distance $W_2$ to the target for population (P), independent-sample (I) and same-sample (S) coefficients at both steps. Bars and bands show the mean $\pm$ sd over ten realisations. The exact-flow value is computed from the source clouds of BKT, and the reference from an independent target cloud of the same size.}\label{fig:ou_path_comparison}
\end{figure}

\paragraph{Test functions.} For the smooth test functions $\sin x$, $\tanh x$ and $e^{-x^2}$, Table~\ref{tab:clt} gives the standard deviation of the particle mean $M^{-1}\sum_jf(X^j_T)$ over $200$ realisations of the source, with $M=M'=2000$, $h=0.025$ and the same threshold and velocity bound, together with the asymptotic standard deviations $\sigma/\sqrt M$ of Theorem~\ref{thm:clt}. These are computed from the population flow of $4\cdot10^5$ source points, with $\nabla m(\bar c)$ obtained by central differences of step $10^{-3}$ on $2\cdot10^4$ common source points. The relative sampling error of the observed standard deviations is about $5\%$, and the predictions agree with the observations within ten percent, except for $\sin x$ at $r=40$ in regime (S). There $|\beta|=0.17$, against $|\beta|<0.01$ for $\tanh x$ and $e^{-x^2}$, since $\sin x$ varies in the tails, where the safeguards act and the flux term of Proposition~\ref{prop:identity} enters $\beta$. Regime (P) has the standard deviation of an independent sample, regime (I) a larger one, and regime (S) reduces the standard deviation of an independent sample by factors between $4$ and $20$ at $r=20$ and between $14$ and $70$ at $r=40$. The factor is smallest for $e^{-x^2}$, whose composition with the flow is a narrow bump in the source variable.

\begin{table}[t]\centering\footnotesize\setlength{\tabcolsep}{4pt}
\begin{tabular}{lcccccccc}\toprule
& & & \multicolumn{2}{c}{(P)} & \multicolumn{2}{c}{(I)} & \multicolumn{2}{c}{(S)}\\
$f$ & $r$ & independent & observed & Theorem~\ref{thm:clt} & observed & Theorem~\ref{thm:clt} & observed & Theorem~\ref{thm:clt}\\\midrule
$\sin x$ & 20 & 14.7 & 15.8 & 14.7 & 22.3 & 20.9 & 0.75 & 0.69\\
$\tanh x$ & 20 & 14.0 & 15.4 & 14.0 & 21.3 & 19.8 & 0.76 & 0.80\\
$e^{-x^2}$ & 20 & 7.55 & 6.87 & 7.57 & 9.51 & 10.2 & 2.06 & 1.95\\
$\sin x$ & 40 & 14.7 & 14.1 & 14.7 & 19.1 & 21.0 & 0.35 & 0.48\\
$\tanh x$ & 40 & 14.0 & 13.5 & 14.0 & 18.4 & 19.8 & 0.20 & 0.19\\
$e^{-x^2}$ & 40 & 7.55 & 7.45 & 7.54 & 10.8 & 10.7 & 0.52 & 0.53\\\bottomrule
\end{tabular}
\caption{Matched Ornstein-Uhlenbeck setting with $M=2000$. Standard deviation, in units of $10^{-3}$, of the particle mean of $f$ over $200$ realisations in the three regimes, the asymptotic values $\sigma/\sqrt M$ of Theorem~\ref{thm:clt}, and $\mathrm{sd}_\pi(f)/\sqrt M$ for an independent sample.}\label{tab:clt}
\end{table}

\section{Details of the numerical experiments}\label{app:exp}

\paragraph{Cost of the velocity.} With eigenfunctions given as combinations of $J$ dictionary functions, the coefficients $e^{-\hat\lambda_kt}\hat c_k$, which are common to all particles, are contracted with the eigenvectors first, in $O(rJ)$ operations, and an evaluation of the velocity then costs one evaluation of the dictionary and of its gradients at the $M$ particles, $O(MJd)$ operations in general.

\subsection{Ornstein-Uhlenbeck processes}\label{app:hyp}

\paragraph{Settings.} The process is $dX=-X\,dt+\sqrt2\,dW$ with target $N(0,1)$, the source is $N(1.5,0.5^2)$, and the spectrum is the analytical one, $\lambda_k=k$ and $\varphi_k=He_k/\sqrt{k!}$. The $M=10^4$ particles are transported up to $T=6$ by the fourth-order Runge-Kutta method with step $0.05$, threshold $\varepsilon_0=10^{-3}$ and bound $v_{\max}=20$, without spatial clipping or reflection, and the stability diagnostic below uses a grid on $[-4,4]$. The coefficients $\hat c_k=\frac1M\sum_i\varphi_k(X_0^i)$, $\hat c_0=1$, are computed once from the initial particles that are then transported (regime (S)) and propagated by $e^{-\lambda_kt}$. The error is the empirical $W_2$ distance between the sorted particles and the quantiles of $N(0,1)$, and the reference is the same distance for an independent sample of $N(0,1)$ of the same size, with mean $\pm$ sd over $10$ independent samples ($0.0189\pm0.0049$ for $M=10^4$).

\begin{figure}[t]\centering\includegraphics[width=\textwidth]{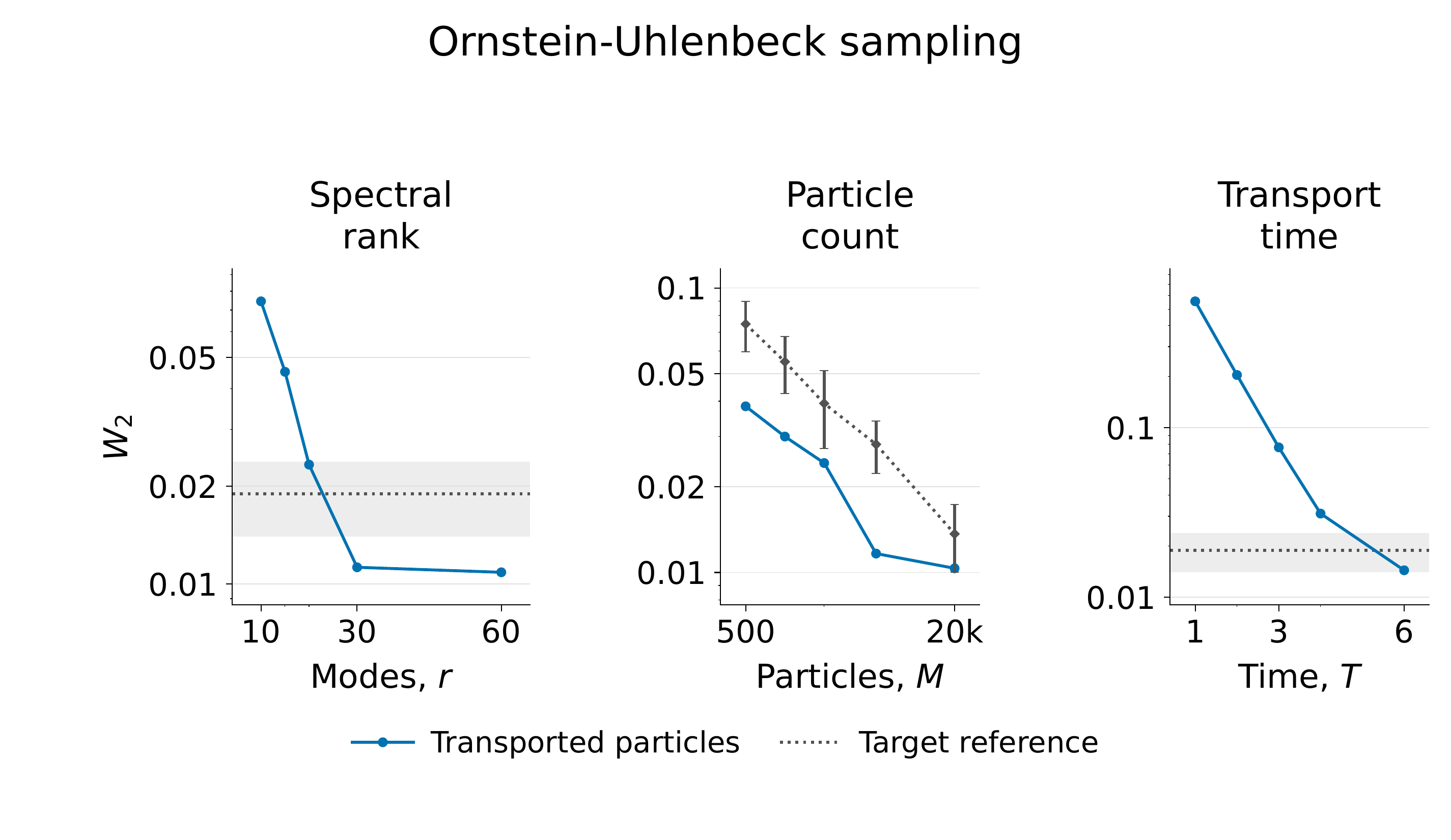}
\caption{Ornstein-Uhlenbeck process with the analytical spectrum and particle coefficients. Error against the number of modes ($M=10^4$, $T=6$, left), against the number of particles ($r=40$, $T=6$, middle, with the reference for each $M$) and against the horizon ($r=40$, $M=10^4$, right). The dotted line, band and bars give the reference mean $\pm$ sd over $10$ repetitions.}\label{fig:ou}
\end{figure}

\paragraph{Dependence on the parameters.} Figure~\ref{fig:ou} gives the three variations quoted in Section~\ref{sec:exp}. The error equals $0.0744$, $0.0451$, $0.0233$, $0.0112$ and $0.0109$ at $r=10,15,20,30,60$. For $r=40$ it equals $0.0384$, $0.0300$, $0.0242$, $0.0117$ and $0.0103$ at $M=500,10^3,2\cdot10^3,5\cdot10^3,2\cdot10^4$, against references $0.0747\pm0.0151$, $0.0550\pm0.0124$, $0.0393\pm0.0120$, $0.0282\pm0.0059$ and $0.0137\pm0.0037$, and $0.556$, $0.204$, $0.076$, $0.031$ and $0.014$ at $T=1,2,3,4,6$. The error of BKT lies below the reference at every $M$ in this variation. The coefficients are computed from the transported initial particles (regime (S)), and the reduction below the reference is the cancellation of Proposition~\ref{prop:cancel}.

\subsubsection*{A ten-dimensional process with nearest-neighbour coupling}\label{app:ou_hd}

\paragraph{Model and analytical eigenpairs.} Let $V_\gamma(x)=\frac12|x|^2+\frac\gamma2\sum_{j=1}^{9}(x_{j+1}-x_j)^2$, $A_\gamma=I+\gamma E^\top E$ and $\pi_\gamma=N(0,A_\gamma^{-1})$, with $\gamma=0.15$, for which the correlations of neighbouring coordinates lie between $0.1170$ and $0.1235$. For $\gamma=0$ the target is the product Gaussian. Write $A_\gamma=Q\,\mathrm{diag}(a_1,\dots,a_{10})Q^\top$ with $Q$ orthogonal and $z_j=\sqrt{a_j}\,q_j^\top x$. The eigenfunctions and decay rates of $L$ are $\varphi_{\boldsymbol n}(x)=\prod_jh_{n_j}(z_j)$ and $\lambda_{\boldsymbol n}=\sum_jn_ja_j$, with $h_k=He_k/\sqrt{k!}$. The isotropic source $N(0.5\mathbf1,0.25I)$ and the target remain products in the coordinates $z$, so the ratio factorises, $\hat\rho_t(x)=\prod_j\hat\rho_{j,t}(z_j)$ with $\hat\rho_{j,t}(z)=1+\sum_{k\le30}e^{-ka_jt}\hat c_{jk}h_k(z)$ and $\hat c_{jk}=M^{-1}\sum_ih_k(z^i_{j,0})$, the velocity is $\dot z_j=-a_j\partial_{z_j}\log\hat\rho_{j,t}$, and the particles are mapped back to the coordinates $x$ for the evaluation. The coupling is thus removed by the rotation, and the experiment concerns a correlated ten-dimensional target with analytical eigenpairs.

\paragraph{Protocol and results.} The computations use $M=5000$ particles, ten realisations of the source, with the same initial cloud across cases, step $0.05$, threshold $10^{-3}$, a componentwise velocity bound $20$ in the coordinates $z$ and no domain restriction. The error is the sliced $W_2$ ($64$ fixed directions) of the transported particles to the analytical projected Gaussian laws, and the reference is the same distance for $10$ independent target samples of size $M$ ($0.0235\pm0.0018$ at $\gamma=0.15$ and $0.0262\pm0.0026$ at $\gamma=0$). The error decreases at the rate $e^{-T}$ and then approaches the finite-particle level, from $0.643$ at $T=0$ to $0.0195\pm0.0005$ at $T=6$ (Figure~\ref{fig:ou_hd}, and $0.0219\pm0.0004$ against $0.0262\pm0.0026$ at $\gamma=0$). The relative covariance error at $T=6$ is $0.042$ against $0.047$ for the reference, and the neighbouring correlations and the $(x_1,x_2)$ marginal agree with the target (Figures~\ref{fig:ou_hd} and~\ref{fig:ou_marginals}). The paths of $0.26\pm0.10\%$ of the particles meet the threshold or the componentwise velocity bound at $\gamma=0.15$ ($0.38\pm0.10\%$ at $\gamma=0$), and the error of the remaining particles is $0.0198\pm0.0006$.

\begin{figure}[t]\centering\includegraphics[width=\textwidth]{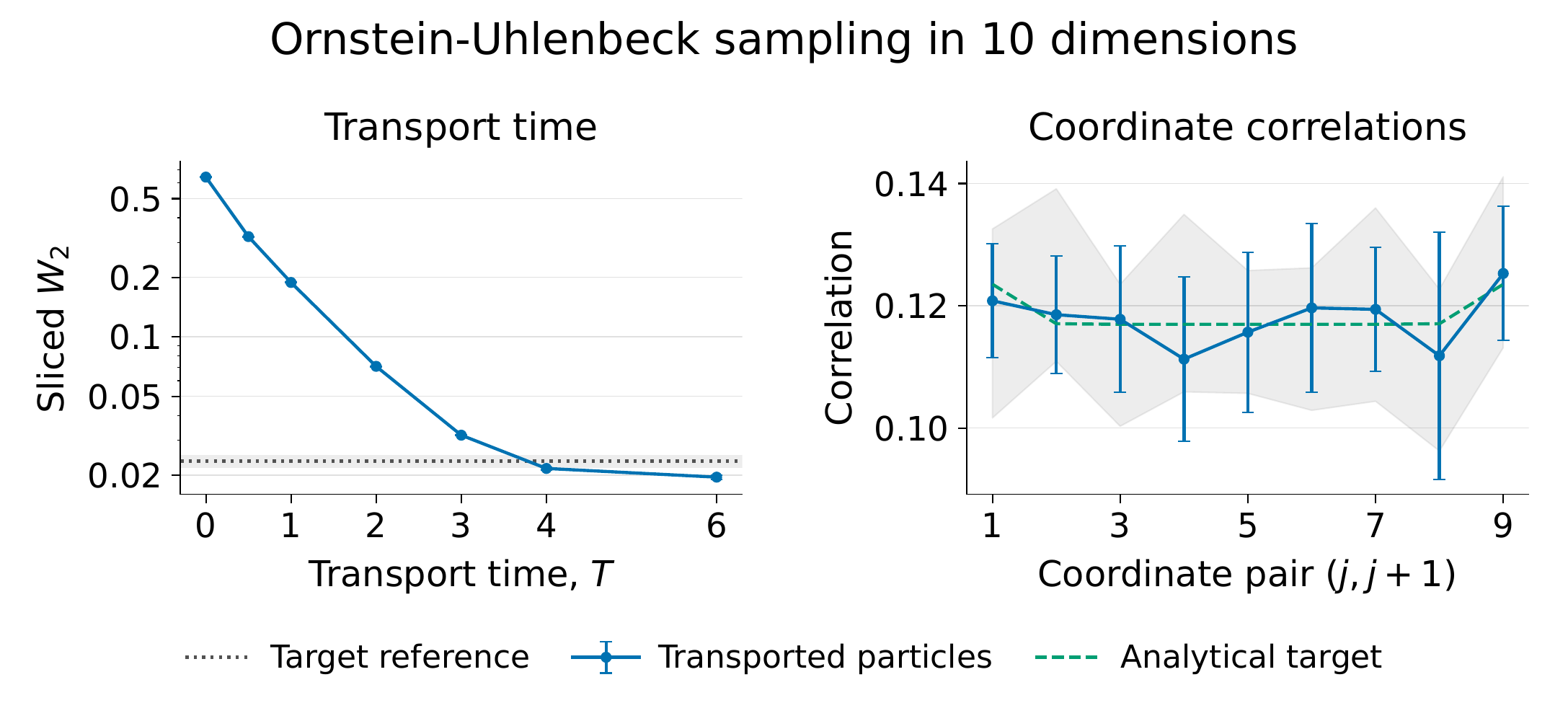}
\caption{Ten-dimensional Ornstein-Uhlenbeck process with coupling $\gamma=0.15$. Sliced $W_2$ against the transport time (left) and neighbouring-coordinate correlations of the transported particles against the target (right), mean $\pm$ sd over ten realisations. The reference is direct sampling from the analytical Gaussian.}\label{fig:ou_hd}
\end{figure}

\begin{figure}[t]\centering\includegraphics[width=\textwidth]{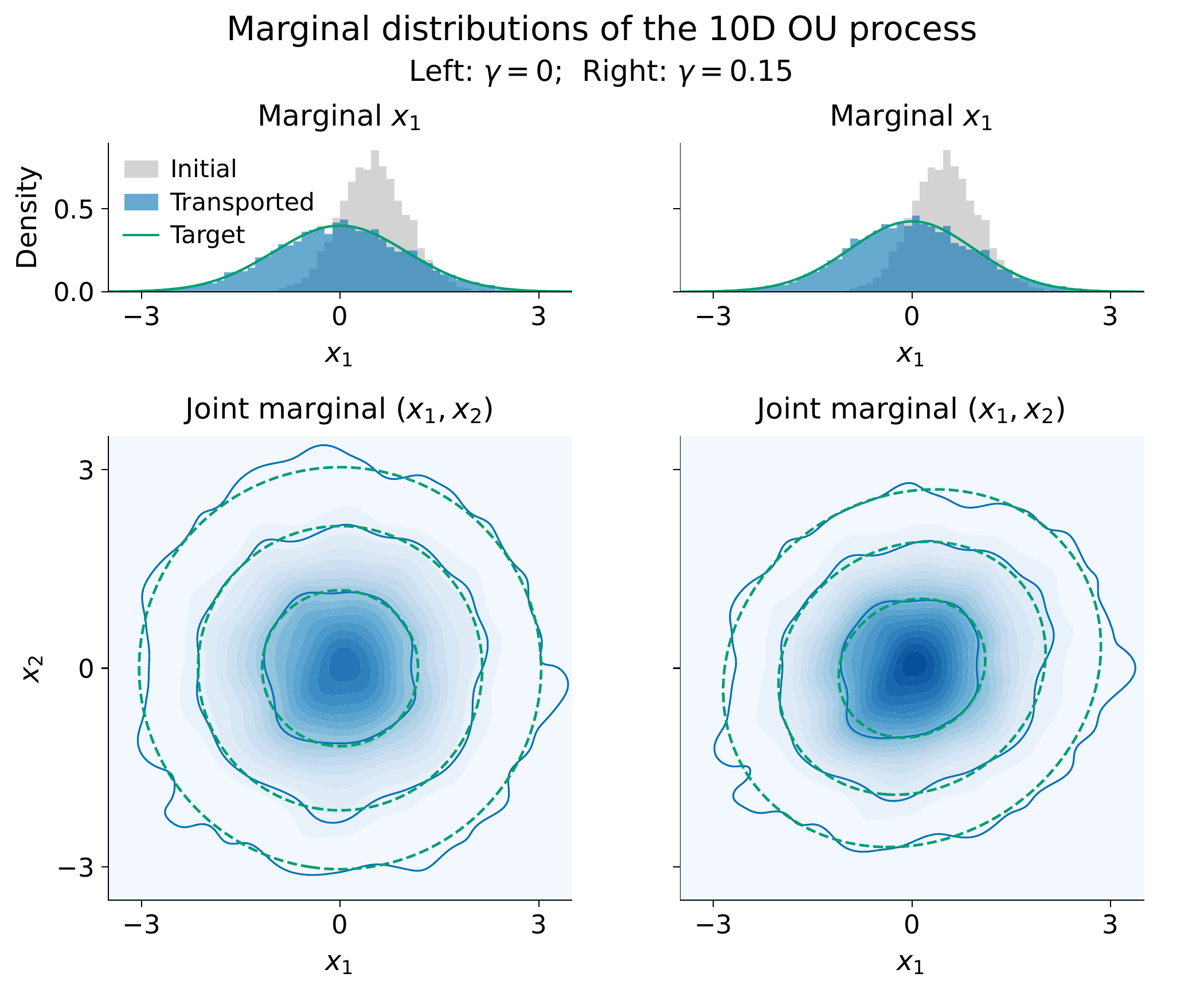}
\caption{Marginals of the transported particles in ten dimensions at $T=6$ (first realisation of the source), $\gamma=0$ (left) and $\gamma=0.15$ (right). The top row shows the $x_1$ marginal of the initial particles, of the transported particles and of the analytical target, the bottom row the $(x_1,x_2)$ joint marginal of the transported particles (kernel density estimate) against the analytical contours enclosing $50\%$, $90\%$ and $99\%$ of the probability.}\label{fig:ou_marginals}
\end{figure}

\paragraph{Verification of the expansion.} The truncated expansion $1+\sum_{k\le r}e^{-kt}c_k\varphi_k$ with population coefficients is compared with the closed-form ratio $\rho_t=N(1.5e^{-t},1-0.75e^{-2t})/N(0,1)$ on the grid, and the closed-form $\chi^2$ sum $\sum_ke^{-2kt}c_k^2$ with quadrature. The differences at finite $r$ are truncation and quadrature errors and decrease as $r$ grows.

\paragraph{Positivity and stability diagnostics (particle coefficients).} On the domain, the minimum of the truncated ratio is $-4.89$, $-0.44$, $-0.40$ at $t=0$ for $r=10,20,40$, becomes positive at $t\approx0.70$, $0.45$, $0.30$, and is $0.031$ at $t=1$ for every $r$. Truncation produces negative densities at early times, which is why Section~\ref{sec:theory} confines the analysis to a validity region. The one-sided Lipschitz integral $\int_0^T(L_t^+)_+dt$ of the approximate velocity is $213$, $97$, $48$ over the whole domain and $36.2$, $4.9$, $4.5$ over a screened region, against the exact value $\log(1/s_0)=0.69$. The screening thus removes almost all of the growth, which is concentrated where the truncated ratio is small and where the analysis uses only the bound of the field. With $\epsilon_t\coloneqq\max_{\rm grid}|\hat\rho_t-\rho_t|$, the grid statistics use the points where $\rho_t\ge\sqrt{\epsilon_t}$, and the particle statistics the particles with $|x|\le b$ and $\rho_t(x)\ge2\epsilon_t$. Along the particle paths, $99.6\%$ of the particle steps fall in the screened region for $r=20$ and $40$, with no activation of the threshold or of the velocity bound inside it. Counted per particle, the paths of $3.2\%$, $0.9\%$ and $0.02\%$ of the particles meet the threshold or the velocity bound for $r=10$, $15$ and $20$, none for $r=30$ and $0.02\%$ for $r=40$.

\subsection{Double wells, multi-wells and products}\label{app:comp}

\paragraph{Protocol.} The repeated experiments of this section have ten realisations, and the rank and dictionary scans use the first one. A realisation consists of a new source sample and, for the estimated spectra, of a new trajectory from which the eigenpairs are estimated. For the double-well, multi-well and product systems the transported particles are compared with one fixed target sample, and the reference compares two independent target samples. For the ten-dimensional Ornstein-Uhlenbeck process both are compared with the analytical projected Gaussian laws (Appendix~\ref{app:ou_hd}), and in one dimension both the transported and the reference particles are compared with the target quantiles. The main dictionary has $J=145$ functions in two dimensions and $J=657$ in ten, and a one-dimensional factor spectrum (FD-1401) is used for the separable systems. The diffusion $dX=-\nabla V\,dt+\sqrt2\,dW$ is simulated with step $2\cdot10^{-3}$ after $20$ time units of equilibration, with the drift-implicit Euler scheme for the polynomial multi-wells, and pairs are taken at lag $0.1$. The sample sizes $10^4$ and $10^5$ are initial segments of the trajectory of $2\cdot10^5$ pairs of each realisation. The eigenpairs are computed by reversible gEDMD from the matrices $G$ and $D$, with whitening of $G$ at tolerance $10^{-8}$ and normalisation in $L^2(\pi)$. In two dimensions the dictionaries are $12\times12$ Gaussians plus the constant ($J=145$, width $7.5\%$ of the domain) and tensor Legendre polynomials of total degree at most $16$ ($J=153$). In ten dimensions they are eight Gaussians per coordinate with products over neighbouring pairs ($J=657$) and Legendre polynomials of total degree at most $3$ ($J=286$). For the products they are $80$ Gaussians plus the constant on $[-2.8,2.8]$ (width $0.14$, whitening $10^{-10}$) or Legendre polynomials of degree $32$ per coordinate ($J=33$). BKT uses fourth-order Runge-Kutta with step $0.05$ (double wells) or $0.02$ (multi-wells and products), $\varepsilon_0=10^{-3}$, $v_{\max}=20$, horizon $T=8/\lambda_1$ with $\lambda_1$ from the FD spectrum or from the fitted generator. In ten dimensions the candidate modes are those whose generator and lag eigenvalues agree within a relative tolerance of $0.5$, and among them the $r$ modes with the largest coefficients, computed from an independent source sample of size $10^5$, are used, so that the selected modes do not depend on the transported particles. The comparisons across sample sizes, dictionaries and coefficient regimes are paired within a realisation. For the two-dimensional double wells the source is a bump of centre $1$ and width $0.6$ in $x_1$ times $N(0,0.3)$ in $x_2$, for which $\rho_0$ is bounded when $\beta\le1$ (the quadratic coefficient of $\log\rho_0$ in $x_2$ is $-(\tfrac23-\tfrac\beta2)<0$). In the ten-dimensional double well the source has the same bump in $x_1$ and nine independent $N(0,0.5)$ harmonic coordinates, which are at equilibrium. The domain $B=[-2.8,2.8]\times[-2.5,2.5]$ of the two-dimensional double wells is enforced by componentwise projection of the trajectory data, of the reference samples and of the transported particles, the latter at the arguments of the Runge-Kutta stages and at the end of each step, so that the finite-step reference law is that of a projected Euler scheme. The paths of fewer than $0.3\%$ of the particles meet the threshold, the velocity bound or the projection ($0.06$ to $0.28\%$ over the spectra and couplings), and the error of the remaining particles differs from that of all particles by at most $0.003$. Halving the transport step changes the $\beta=0.5$ error from $0.0305$ to $0.0309$ (RBF) and by at most $0.0010$ for the other spectra. The FD-201 eigenpairs of the multi-wells use Neumann tail truncation at the points where $U=30$. The polynomial targets are defined on the whole plane and the particles are never clipped or reflected. Their normalisation, distribution function, basin masses (separated by the saddles) and barriers are obtained by direct integration of $e^{-U}$ on $40\,001$ points up to $U=40$. The sample size is $n=2\cdot10^5$ for the two-dimensional systems and $10^5$ for the ten-dimensional ones, and Table~\ref{tab:samplesize} gives the error at all three sizes.

\begin{table}[t]\centering\footnotesize\setlength{\tabcolsep}{5pt}
\begin{tabular}{lcccl}\toprule
system & $n$ & $M$ & $r$ & well masses (target)\\\midrule
2D double well, $\beta=0,0.25,0.5,1$ & $2\cdot10^5$ & 2000 & 64 & --\\
2D four-well product & $2\cdot10^5$ & 2000 & 64 & $0.23$ to $0.27$ ($0.25$ each)\\
2D nine wells & $2\cdot10^5$ & 2000 & 64 & basins of Table~\ref{tab:cells_app}\\
10D double-well product, $\beta=0$ & $10^5$ & 1000 & 16 & left $0.55\pm0.01$ ($0.50$)\\
10D and 50D double-well products & $2\cdot10^5$ & $2\cdot10^4$ & 16 per coordinate & --\\\bottomrule
\end{tabular}
\caption{Settings of the experiments of Table~\ref{tab:cells}, with ten realisations each. The columns give the sample size $n$, the number $M$ of transported particles, the number $r$ of modes and the well masses of the transported particles (RBF, range or mean $\pm$ sd over the realisations) against the target.}\label{tab:settings}
\end{table}

\begin{figure}[t]\centering\includegraphics[width=0.8\textwidth]{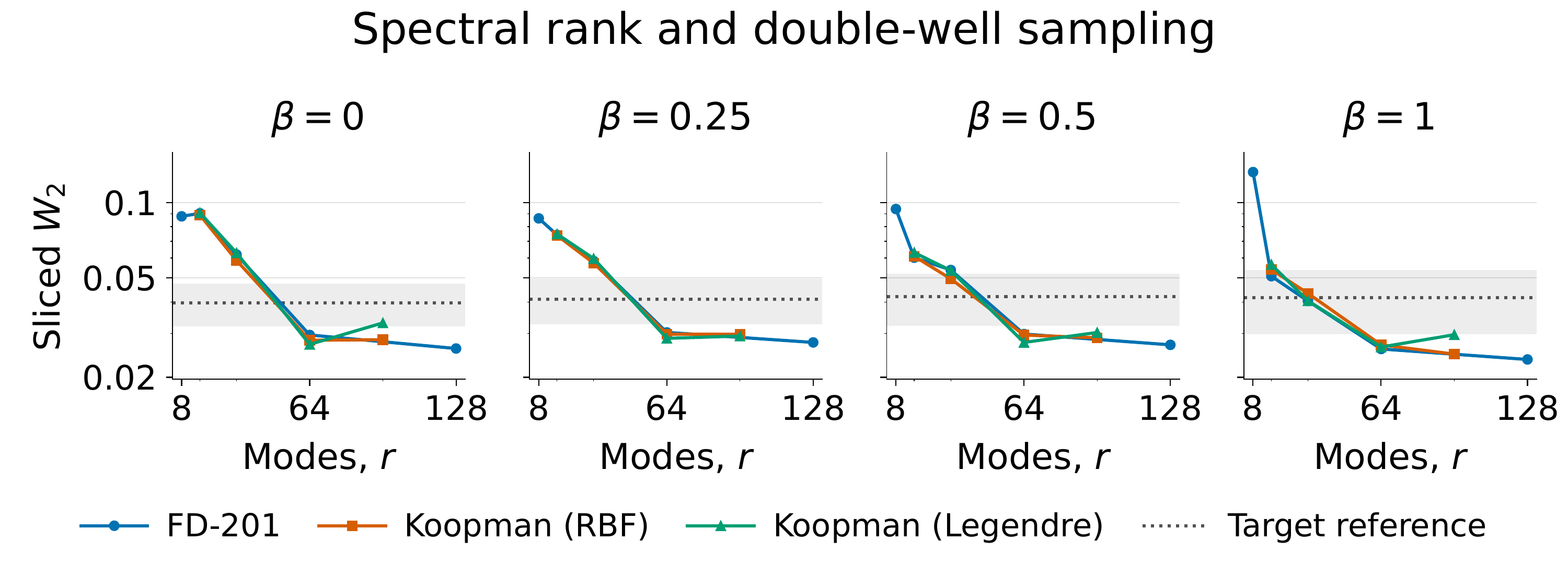}
\caption{Two-dimensional double well. Error against the number of modes for one realisation of the source ($n=2\cdot10^5$), with the FD-201, Koopman (RBF) and Koopman (Legendre) eigenpairs, for $\beta=0,0.25,0.5,1$. The dotted line and band show the reference mean $\pm$ sd over $10$ repetitions.}\label{fig:cells}
\end{figure}

Here \emph{product} refers to the invariant density, since an additive potential $V(x)=\sum_i V_i(x_i)$ gives $\pi(x)=\prod_i\pi_i(x_i)$. The uncoupled double well has one double-well factor and Gaussian factors, whereas the four-well target and the ten- and fifty-dimensional double-well products have two, ten and fifty double-well factors, respectively.

\paragraph{Dictionary sensitivity.} Table~\ref{tab:dict}, Table~\ref{tab:cells_app} and Figure~\ref{fig:cells} compare the two dictionaries at the sample sizes of Table~\ref{tab:settings}. Since BKT itself is unchanged, the comparison measures the sensitivity of the spectral estimation to the dictionary. On the two-dimensional double wells and the four-well the two dictionaries are equivalent within the fluctuation between realisations, and both are below the reference. On the nine-well the Legendre dictionary fails in five of the ten realisations (errors $0.90$ to $2.71$, against $0.036$ to $0.052$ in the other five), with a non-monotone dependence on $r$ ($0.9$ to $2.3$ for $r=8$ to $96$ in the first realisation). The polynomials grow without bound outside the domain, where the estimated eigenfunctions are unreliable, and a few particles are driven far into the tails. The paths of $0.2\%$ of the particles meet a safeguard, and the error of the remaining particles is $0.042\pm0.007$. The same instability appears on the four-well at $r=32$, in nine of ten realisations (median $2.51$), and not at $r=64$ ($0.026\pm0.002$). The RBF dictionary, whose functions are bounded, is stable in all realisations. In ten dimensions it gives the smaller sliced error ($0.061$ against $0.072$). The instability is therefore a property of the unbounded polynomial dictionary, and BKT inherits the stability of the bounded one.

\begin{table}[t]\centering\scriptsize\setlength{\tabcolsep}{3pt}
\begin{tabular}{lcccc}\toprule
system & FD & Koopman (RBF) & Koopman (Legendre) & reference\\\midrule
2D double-well product, $\beta=0$ & $0.027\pm0.002$ & $0.028\pm0.003$ & $0.028\pm0.003$ & $0.040\pm0.008$\\
2D double well, $\beta=0.25$ & $0.028\pm0.002$ & $0.029\pm0.003$ & $0.030\pm0.003$ & $0.041\pm0.008$\\
2D double well, $\beta=0.5$ & $0.028\pm0.002$ & $0.031\pm0.004$ & $0.029\pm0.003$ & $0.042\pm0.010$\\
2D double well, $\beta=1$ & $0.027\pm0.002$ & $0.027\pm0.003$ & $0.028\pm0.003$ & $0.042\pm0.012$\\
2D four-well product & $0.024\pm0.003$ & $0.028\pm0.002$ & $0.026\pm0.002$ & $0.043\pm0.008$\\
2D nine wells & $0.033\pm0.003$ & $0.055\pm0.006$ & fails in $5$ of $10$ & $0.034\pm0.008$\\
10D double-well product, $\beta=0$ & $0.053\pm0.002$ & $0.061\pm0.002$ & $0.072\pm0.008$ & $0.059\pm0.004$\\
10D double-well product & $0.0151\pm0.0003$ & $0.0195\pm0.0021$ & $0.0195\pm0.0021$ & $0.0156\pm0.0014$\\
50D double-well product & $0.0158\pm0.0002$ & $0.0188\pm0.0012$ & $0.0188\pm0.0012$ & $0.0162\pm0.0008$\\\bottomrule
\end{tabular}
\caption{Dictionary comparison (settings of Table~\ref{tab:settings}), sliced $W_2$ with the RBF and the Legendre dictionaries, mean $\pm$ sd over ten realisations, and the reference (mean $\pm$ sd over $10$ repetitions). The nine-well Legendre entry is discussed in the text.}\label{tab:dict}
\end{table}

\begin{table}[t]\centering\scriptsize\setlength{\tabcolsep}{3pt}
\begin{tabular}{llcccccc}\toprule
system & spectrum & $r=8$ & $16$ & $32$ & $64$ & $96$ & $128$\\\midrule
4 wells & FD-201 & 0.129 & 0.124 & 0.051 & 0.024 & -- & 0.022\\
 & Koopman (RBF) & -- & 0.125 & 0.050 & 0.029 & 0.029 & --\\
 & Koopman (Legendre) & -- & 0.128 & 2.44 & 0.027 & 0.027 & --\\
9 wells & FD-201 & 0.325 & 0.199 & 0.049 & 0.036 & -- & 0.032\\
 & Koopman (RBF) & 0.309 & 0.188 & 0.063 & 0.063 & 0.062 & --\\
 & Koopman (Legendre) & 0.91 & 1.56 & 2.34 & 1.64 & 1.88 & --\\\bottomrule
\end{tabular}
\caption{Dependence on the number of modes for the multi-wells, for one realisation of the source and the sample sizes of Table~\ref{tab:settings}, in sliced $W_2$ against $Y_1$. At $r=3$ the four-well gives $0.32$ with every spectrum, and the reference levels are $0.043$ and $0.034$. The double wells are shown in Figure~\ref{fig:cells}, and at $\beta=0.5$ the RBF errors are $0.061$, $0.050$, $0.030$ and $0.029$ at $r=16,32,64,96$.}\label{tab:cells_app}
\end{table}

\paragraph{Coefficients from an independent sample.} On the double wells with the RBF eigenpairs, coefficients from an independent source sample of size $M'=M=2000$ (regime (I)) give $0.049\pm0.014$, $0.049\pm0.014$, $0.048\pm0.011$ and $0.043\pm0.008$ at $\beta=0,0.25,0.5,1$, against $0.028$, $0.029$, $0.031$ and $0.027$ with the coefficients of the transported particles (regime (S)). The paired difference is $0.015$ to $0.021$ and is positive in $39$ of the $40$ pairs (Figure~\ref{fig:regimes}). Relative to the reference mean, regime (S) gives $0.65$ to $0.73$ and regime (I) $1.02$ to $1.23$, against $2^{-1/2}\approx0.71$ for an exact representation of $\pi$ and $(3/2)^{1/2}\approx1.22$ for a doubled sampling error of the retained modes (Appendix~\ref{app:resolution}), in agreement with Proposition~\ref{prop:cancel}, by which the sampling error of the retained modes cancels in (S) and adds to the fluctuation of the coefficients in (I).

\begin{figure}[t]\centering\includegraphics[width=0.95\textwidth]{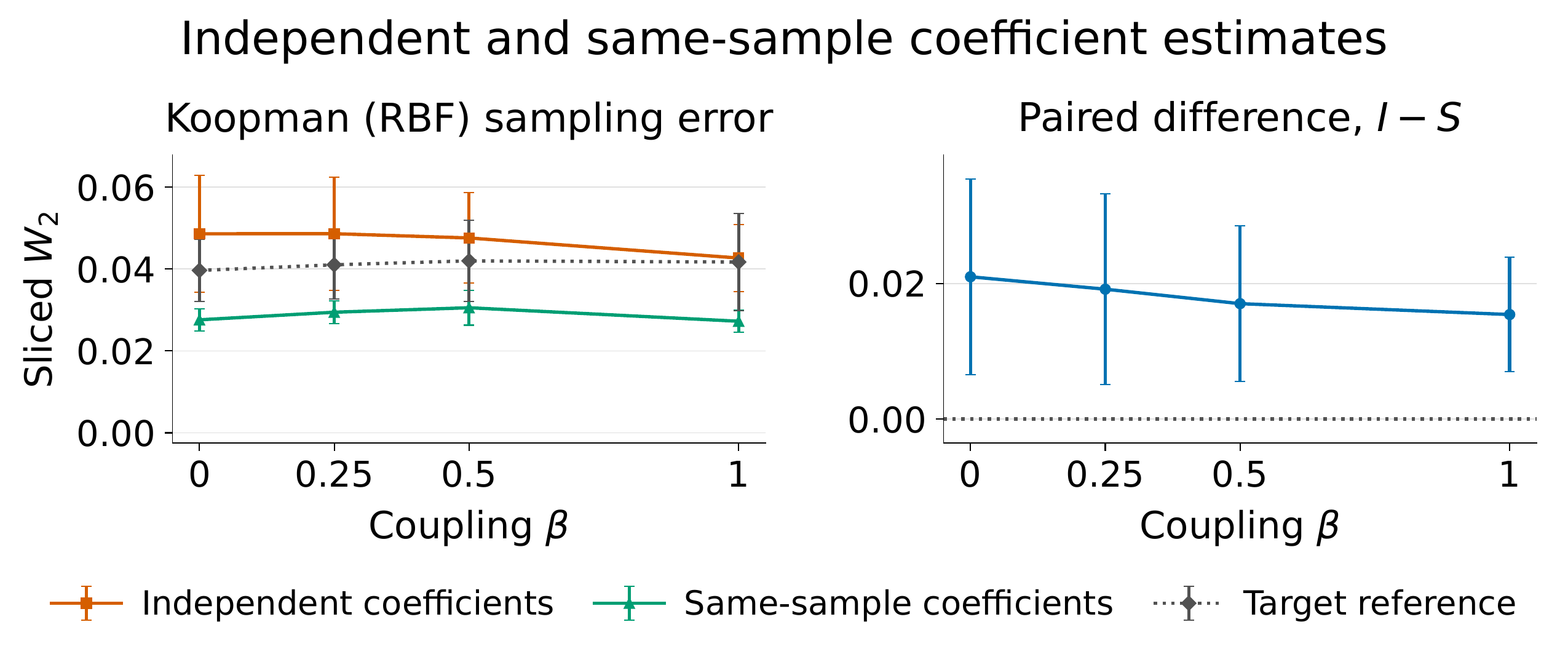}
\caption{Two-dimensional double well. Error with coefficients from the transported particles (regime (S)) and from an independent source sample (regime (I)), mean $\pm$ sd over ten realisations, and reference mean $\pm$ sd over $10$ repetitions.}\label{fig:regimes}
\end{figure}

\paragraph{Sample size.} Table~\ref{tab:samplesize} gives the RBF error at the three sample sizes. It decreases with $n$ in two and in ten dimensions, and the experiments of Table~\ref{tab:cells} use $n=2\cdot10^5$ in two dimensions and $10^5$ in ten.

\begin{table}[t]\centering\footnotesize
\begin{tabular}{lccc}\toprule
system & $n=10^4$ & $n=10^5$ & $n=2\cdot10^5$\\\midrule
2D double-well product, $\beta=0$ & $0.035\pm0.004$ & $0.029\pm0.003$ & $0.028\pm0.003$\\
2D double well, $\beta=0.25$ & $0.036\pm0.004$ & $0.031\pm0.004$ & $0.029\pm0.003$\\
2D double well, $\beta=0.5$ & $0.037\pm0.008$ & $0.031\pm0.004$ & $0.031\pm0.004$\\
2D double well, $\beta=1$ & $0.053\pm0.020$ & $0.028\pm0.002$ & $0.027\pm0.003$\\
2D four-well product & $0.048\pm0.013$ & $0.029\pm0.004$ & $0.028\pm0.002$\\
2D nine wells & $0.087\pm0.024$ & $0.056\pm0.006$ & $0.055\pm0.006$\\
10D double-well product, $\beta=0$ & $0.068\pm0.004$ & $0.061\pm0.002$ & $0.060\pm0.004$\\
\bottomrule
\end{tabular}
\caption{Koopman (RBF) error (mean $\pm$ sd over ten realisations) against the sample size $n$ of the trajectory data.}\label{tab:samplesize}
\end{table}

\paragraph{Comparison with a kernel particle method.} On the double wells at $\beta=0$ and $0.5$, with $n=2\cdot10^5$ pairs and $2000$ output particles (ten realisations, source as above), BKT is compared with the deterministic particle method for diffusion \citep{degond1990,chertock2017}, an interacting particle flow along $\hat b-\nabla\log\hat\mu_t$, where $\hat b$ is the drift estimated by regression of the increments on the gradients of the dictionary and $\hat\mu_t$ a Gaussian kernel density estimate of the current particles (Silverman bandwidth). Table~\ref{tab:comparison} and Figure~\ref{fig:comparison} give the errors and the computation times on one core, estimation included. With the coupling $\beta=0.5$ the kernel flow has a larger mean error and twice the spread ($0.048\pm0.027$ against $0.039\pm0.013$ at $\beta=0$), while the error of BKT is unchanged within the fluctuation between realisations ($0.031\pm0.004$ against $0.028\pm0.003$). At the prescribed observation time the error falls below the reference mean plus one standard deviation in ten realisations out of ten for BKT, and in nine and eight for the kernel flow at $\beta=0$ and $0.5$. Unlike a subsample of stationary data, BKT provides a map from any admissible source to the target, which also applies to non-stationary data.

\begin{table}[t]\centering\footnotesize
\begin{tabular}{lcccc}\toprule
& \multicolumn{2}{c}{sliced $W_2$} & \multicolumn{2}{c}{time (s)}\\
& $\beta=0$ & $\beta=0.5$ & $\beta=0$ & $\beta=0.5$\\\midrule
BKT & $0.028\pm0.003$ & $0.031\pm0.004$ & $9.5$ & $8.7$\\
kernel particle flow & $0.039\pm0.013$ & $0.048\pm0.027$ & $53.6$ & $45.9$\\
reference & $0.040\pm0.008$ & $0.042\pm0.010$ & -- & --\\\bottomrule
\end{tabular}
\caption{Two-dimensional double well, BKT and the kernel particle flow, mean $\pm$ sd over ten realisations, and the reference (mean $\pm$ sd over $10$ repetitions). The times include the estimation of the eigenpairs ($3.1$ and $3.2$ s) and of the drift ($1.3$ s).}\label{tab:comparison}
\end{table}

\begin{figure}[t]\centering\includegraphics[width=0.85\textwidth]{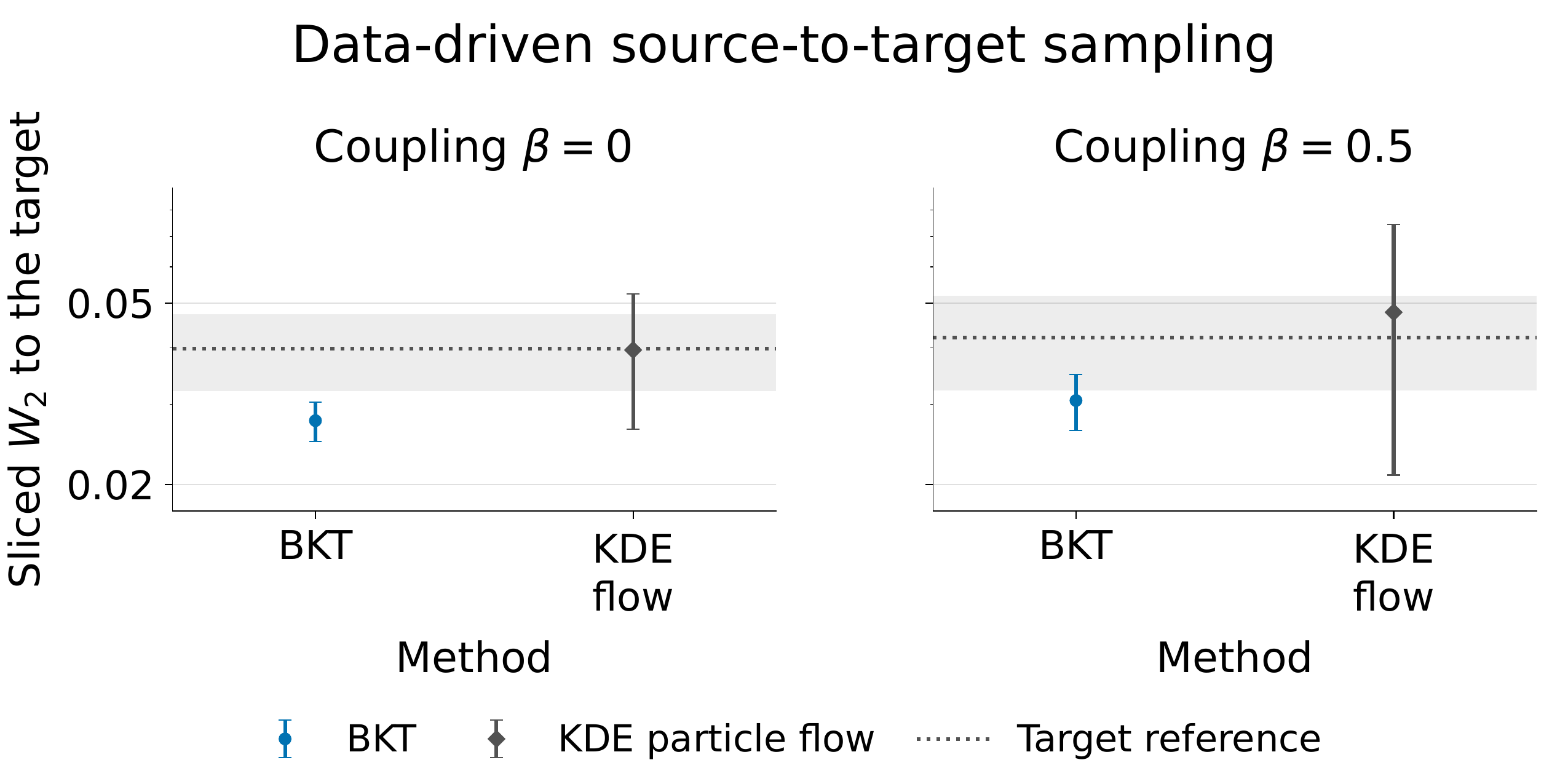}
\caption{Two-dimensional double well at $\beta=0$ and $0.5$. Sliced $W_2$ of BKT and of the kernel particle flow, mean $\pm$ sd over ten realisations. The dotted line and the band show the reference mean $\pm$ sd over $10$ repetitions.}\label{fig:comparison}
\end{figure}

\paragraph{Nine wells.} On this system BKT reaches the reference with the FD spectrum, and the comparison with the estimated spectra isolates the spectral estimation. BKT with the FD-201 eigenpairs gives $0.033\pm0.003$ at $r=64$ ($0.032$ at $r=128$ in the first realisation) against a reference of $0.034\pm0.008$, while Koopman (RBF) gives $0.055\pm0.006$, does not improve beyond $r=32$ (Table~\ref{tab:cells_app}), and the Legendre dictionary fails in half of the realisations (Table~\ref{tab:dict}). A source concentrated in one narrow well needs many modes, and the estimated spectrum supplies about $30$ accurate ones. With accurate eigenpairs BKT thus reaches the reference, and the remaining error is that of the spectral estimation.

\paragraph{Separable products.} For $V=\sum_i(x_i^2-1)^2$ BKT acts coordinate by coordinate with the FD-1401 factor spectrum or with coordinate-wise gEDMD eigenpairs ($n=2\cdot10^5$, $r=16$ per coordinate, $M=2\cdot10^4$, $T=8/\lambda_1=10.69$, step $0.02$), the source and the target being products (Figure~\ref{fig:hd50} for $d=50$). The sliced distances are $0.0151\pm0.0003$ (FD), $0.0195\pm0.0021$ (RBF) and $0.0195\pm0.0021$ (Legendre) against $0.0156\pm0.0014$ for $d=10$, and $0.0158\pm0.0002$, $0.0188\pm0.0012$ and $0.0188\pm0.0012$ against $0.0162\pm0.0008$ for $d=50$. The energy distances are $0.0050\pm0.0010$ for both dictionaries against $0.0052\pm0.0013$ in ten dimensions and $0.0116\pm0.0009$ against $0.0120\pm0.0010$ in fifty. The total variation between the distribution of the number of negative coordinates and its target is $0.0092\pm0.0021$ against $0.0095\pm0.0035$ and $0.0191\pm0.0028$ against $0.0165\pm0.0032$ for the RBF dictionary. The largest marginal errors are $0.0356\pm0.0056$ (RBF) and $0.0354\pm0.0057$ (Legendre) in ten dimensions and $0.0387\pm0.0047$ and $0.0387\pm0.0048$ in fifty, against references of $0.0260\pm0.0055$ and $0.0297\pm0.0045$, with the FD factor at $0.0188\pm0.0014$ and $0.0201\pm0.0009$, so that the excess of the estimated factors comes from the one-dimensional spectral estimation. In each coordinate the paths of $0.12$ to $0.15\%$ of the particles meet the velocity bound. The two dictionaries give the same results, and the number of modes is additive over the coordinates.

\begin{figure}[t]\centering\includegraphics[width=0.9\textwidth]{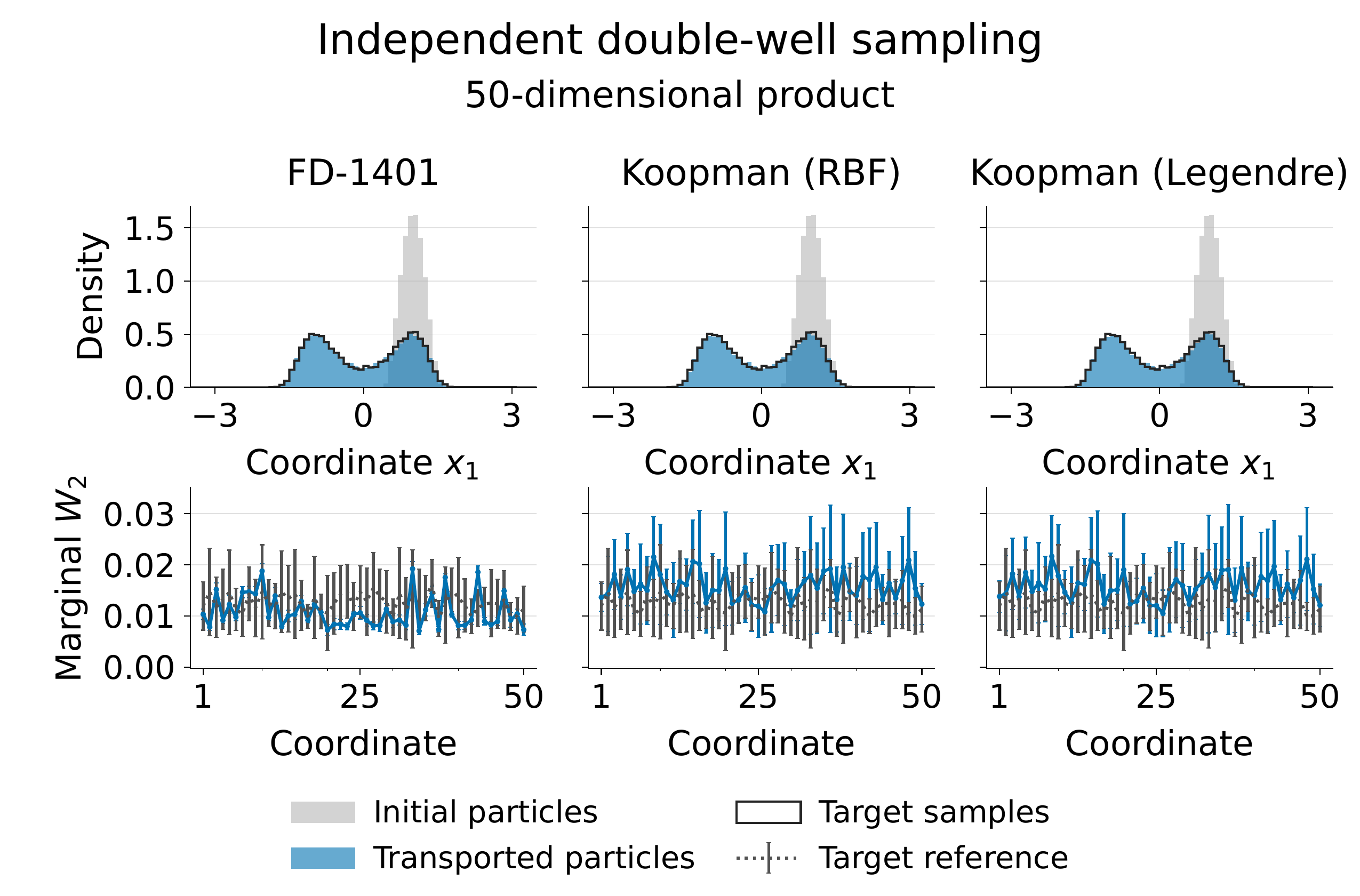}
\caption{Fifty-dimensional double-well product, histograms of the first coordinate of the initial particles, the transported particles and a target sample (top, first realisation) and marginal errors per coordinate, mean $\pm$ sd over ten realisations, with the reference mean $\pm$ sd (bottom), for the FD-1401 factor spectrum and the Koopman (RBF) and Koopman (Legendre) factor spectra.}\label{fig:hd50}
\end{figure}

\subsection{The dihedral marginal of alanine dipeptide}\label{app:ala}

\paragraph{Data and target.} The three public $250$\,ns trajectories of alanine dipeptide in explicit solvent (backbone dihedrals $(\phi,\psi)$ at $1$\,ps, $250{,}000$ frames each, available through \href{https://markovmodel.github.io/mdshare/ALA2/}{\texttt{mdshare}}, \citealp{nuske2017}) are assigned one each to estimation, validation and evaluation. The target is the smoothed empirical marginal of $(\phi,\psi)$ on the torus. The eigenpairs are those of the unit-mobility reversible diffusion with this invariant law, estimated by the Gram and Dirichlet matrices of the estimation frames (smoothed by a wrapped Gaussian of bandwidth $0.10$\,rad) with a real tensor-product Fourier dictionary of degree $28$ in each angle ($3{,}249$ functions), a relative Gram cutoff of $10^{-12}$ and $r=256$ nonconstant modes. The construction uses the Gram and Dirichlet forms of the frames, and the experiment transports one prescribed local source to the smoothed angular distribution.

\paragraph{Transport and evaluation.} The source is a wrapped Gaussian localised in one basin, represented by $1000$ initial points from a scrambled Sobol design. The coefficients are the empirical moments of these points (regime (S)). The equations are integrated to $s=\hat\lambda_1T=8$ with an adaptive eighth-order Runge-Kutta method (DOP853, relative and absolute tolerances $10^{-8}$ and $10^{-10}$, maximum step $0.05$ in $s$). The density threshold is $10^{-3}$ and the velocity bound is $20$ in unscaled time, or $20/\hat\lambda_1$ in $s$. The error is the sliced $W_2$ ($32$ fixed directions) in the embedding $(\cos\phi,\cos\psi,\sin\phi,\sin\psi)$ between the transported particles and $1000$ frames of the evaluation trajectory, and the total variation between the masses of the four Ramachandran basins. The reference is the same distance between two further disjoint clouds of $1000$ evaluation frames, that is, the error of a subsample of the trajectory. Three paired source and reference designs are used with the eigenpairs fixed. Figure~\ref{fig:alanine} shows the final distributions of these runs, which are also those of the comparison below.

\begin{table}[t]\centering\footnotesize\setlength{\tabcolsep}{4pt}
\begin{tabular}{lcccc}\toprule
& draws & sliced $W_2$ & basin-mass TV & reference $W_2$, TV\\\midrule
initial source & 3 & $0.615\pm0.003$ & $0.604\pm0.019$ & --\\
Fourier degree $28$, $r=256$ & 3 & $0.060\pm0.019$ & $0.033\pm0.018$ & $0.051\pm0.021$, $0.023\pm0.014$\\\bottomrule
\end{tabular}
\caption{Alanine dipeptide, sliced $W_2$ and basin-mass total variation of $1000$ transported particles against $1000$ evaluation frames, at $\hat\lambda_1T=8$, mean $\pm$ sd over the three paired designs. The last column gives the reference distances between two further disjoint clouds of evaluation frames.}\label{tab:alanine}
\end{table}

\begin{figure}[t]\centering\includegraphics[width=0.9\textwidth]{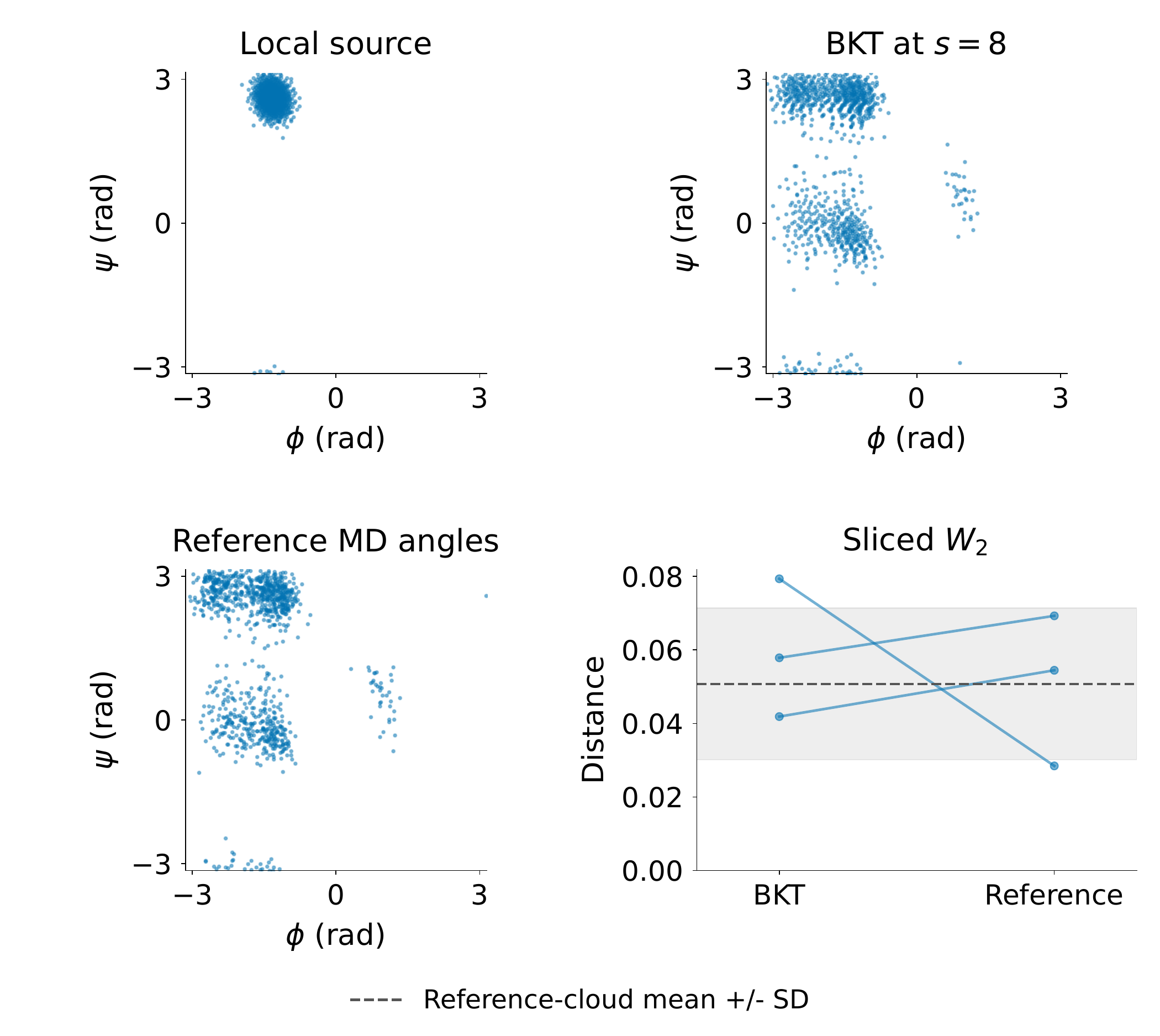}
\caption{Alanine dipeptide with $256$ Fourier modes, the local source, the transported particles at $s=8$ and $1000$ evaluation frames in the $(\phi,\psi)$ plane for the first design, followed by the sliced $W_2$ of the final particles and the paired reference for each of the three designs. The reference band denotes its mean $\pm$ sd.}\label{fig:alanine}
\end{figure}

\paragraph{Results.} The final sliced error is $0.060\pm0.019$ against $0.051\pm0.021$ between reference clouds, and two of the three paired errors are at most their reference. The total variation of the four-basin masses is $0.033\pm0.018$ against $0.023\pm0.014$, and the mean mass of the rare positive-$\phi$ basin is $0.029$ against $0.042$, a difference smaller than the total variation between two reference clouds. The local-width ratio is $1.17\pm0.09$, computed from the median minor-axis standard deviation of twelve-neighbour local clouds, relative to the evaluation cloud (Figure~\ref{fig:alanine}, Table~\ref{tab:alanine}), consistent with the smoothing of the target at bandwidth $0.10$\,rad. With i.i.d.\ initial points in place of the Sobol design the error is $0.058\pm0.021$. The paths of $0.4$ to $0.8\%$ of the particles meet the velocity bound and $0.1$ to $0.2\%$ the threshold.

\paragraph{Comparison with interacting particle methods.} BKT and the Laplacian-adjusted Wasserstein gradient descent (LAWGD) of \citet{chewi2020} are both run with the $256$ nonconstant eigenpairs of the Fourier dictionary of degree $28$, $1000$ initial particles and their empirical coefficients (regime (S)). BKT keeps these coefficients fixed, whereas LAWGD recomputes the current moments in its velocity $-\sum_k\hat\lambda_k^{-1}\nabla\hat\varphi_k\,M^{-1}\sum_i\hat\varphi_k(X^i_t)$. The third method, KDE, adapts the classical diffusion-velocity particle method \citep{degond1990,chertock2017} to periodic kernels. Its target density $\hat\pi$ is fitted once to the $250{,}000$ estimation frames with bandwidth $0.10$\,rad, and its current density $\hat\mu_s$ uses all $1000$ particles, including self terms, with the same bandwidth. In normalised time $s=\hat\lambda_1t$, with $\hat\lambda_1=0.0013776$, its velocity is $(\nabla\log\hat\pi-\nabla\log\hat\mu_s)/\hat\lambda_1$. KDE thus follows the same gradient flow with both densities estimated by kernels and without the generator, whereas BKT and LAWGD take the velocity from the eigenpairs of the generator, which carry the relaxation rates and the metastable structure of the dynamics. The discretisations and integration tolerances of the Fourier and kernel representations are specified in the accompanying code.

Figure~\ref{fig:ala_main} and Table~\ref{tab:alanine_conv} show three paired realisations of the source and of the reference with common estimation frames. BKT reaches the late-time error level by $s=1$ ($0.068$, against $0.210$ for LAWGD and $0.148$ for KDE) and then remains near $0.060$. LAWGD reaches a comparable level at $s=4$, while KDE remains near $0.148$. The two methods built on the spectrum of the generator thus end at $0.06$, and the method built on density estimates at more than twice this error. The kernel method converges to a stationary point of its own velocity $(\nabla\log\hat\pi-\nabla\log\hat\mu_s)/\hat\lambda_1$, at which the logarithmic gradients of the two kernel estimates agree at the particles. The final particles therefore retain the smoothing error of both estimates and the fluctuation of the estimate from $1000$ particles, which are largest where the target density is small, and a longer flow does not reduce them. BKT estimates no density during the transport, since the law of the particles at time $s$ is represented by the propagated coefficients $e^{-\hat\lambda_kt}\hat c_k$. The velocity of BKT uses only the fixed initial coefficients, whereas LAWGD and KDE recompute empirical quantities from the evolving particles. The BKT error reaches its late-time level before that of LAWGD and varies little thereafter. On one thread the estimation of the eigenpairs takes $25$ seconds and the integration to $s=8$ takes $72\pm3$ seconds for BKT, $70\pm4$ for LAWGD and $1056\pm145$ for KDE, and each method first reaches its own late-time level after $34\pm6$, $54\pm5$ and $831\pm90$ seconds.

\begin{table}[t]\centering\footnotesize
\begin{tabular}{lcccc}\toprule
$s=\hat\lambda_1t$ & $0.1$ & $1$ & $4$ & $8$\\\midrule
BKT & $0.091\pm0.011$ & $0.068\pm0.017$ & $0.060\pm0.019$ & $0.060\pm0.019$\\
LAWGD & $0.578\pm0.003$ & $0.210\pm0.004$ & $0.054\pm0.021$ & $0.059\pm0.020$\\
KDE & $0.204\pm0.005$ & $0.148\pm0.004$ & $0.148\pm0.004$ & $0.148\pm0.004$\\\bottomrule
\end{tabular}
\caption{Alanine dipeptide, sliced $W_2$ to the evaluation frames along the normalised flow time, mean $\pm$ sd over three paired realisations, for BKT and the Laplacian-adjusted Wasserstein gradient descent with $256$ eigenpairs, and for the kernel-density particle method with bandwidth $0.10$\,rad.}\label{tab:alanine_conv}
\end{table}

\end{document}